\documentclass[twocolumn]{article}
\usepackage[left=0.75in, right=0.75in, top=1.0in, bottom=1.5in]{geometry}
\usepackage{authblk}
\usepackage{graphicx}
\usepackage{amssymb}
\usepackage{bm}
\usepackage{braket}
\usepackage[dvipsnames]{xcolor}
\definecolor{cobaltblue}{RGB}{0,60,170}
\usepackage{mathtools}
\usepackage{optidef}
\usepackage{physics}
\usepackage[colorlinks=true, linkcolor=cobaltblue,citecolor=ForestGreen, urlcolor=blue]{hyperref}
\usepackage{cleveref}
\usepackage{amsthm}
\usepackage{tikz}
\usetikzlibrary{positioning, 
                quotes}
\usepackage{mathtools}
\usepackage[style=nature]{biblatex}
\usepackage[normalem]{ulem}
\usepackage{changepage}
\usepackage{bbm}
\usepackage{comment}
\usepackage{multicol}
\usepackage{enumitem}
\usepackage{float}
\usepackage{ifthen}
\usepackage{xspace}

\usepackage{adjustbox}
\usepackage{placeins}
\usepackage{caption}

\usepackage{newfloat}
\DeclareFloatingEnvironment[fileext=lof]{suppfigure}
\crefname{suppfigure}{Supplementary Figure}{Supplementary Figures}

\DeclareFloatingEnvironment[fileext=lot]{supptable}

\crefname{supptable}{Supplementary Table}{Supplementary Tables}
\Crefname{supptable}{Supplementary Table}{Supplementary Tables}

\usepackage{algorithm}
\usepackage{algpseudocode}

\usepackage{tcolorbox}

\usepackage[none]{hyphenat}

\crefname{equation}{Eq.}{Eqs.} 
\crefname{section}{Section}{Sections} 
\crefname{figure}{Figure}{Figures} 
\crefname{table}{Table}{Tables} 
\crefname{appendix}{Supplementary Section}{Supplementary Sections} 
\crefname{theorem}{Theorem}{Theorems}
\crefname{conjecture}{Conjecture}{Conjectures}
\crefname{proposition}{Prop.}{Props.}
\crefname{lemma}{Lemma}{Lemmas}
\crefname{corollary}{Corollary}{Corollaries}

\newcommand{\R}{\mathbb{R}}
\newcommand{\C}{\mathbb{C}}
\renewcommand{\H}{\mathcal{H}}
\newcommand\bmt{\bm{\theta}}
\newcommand\mfg{\mathfrak{g}}
\newcommand\mfh{\mathfrak{h}}
\newcommand\mcg{\mathcal{G}}

\newcommand{\cc}[1]{\ensuremath{\mathsf{#1}}}

\newboolean{showcomments}
\setboolean{showcomments}{true}
\DeclareRobustCommand{\HL}[1]{\ifthenelse{\boolean{showcomments}}{{\color{SeaGreen}{#1}}}{}}
\newcommand{\SK}[1]{\ifthenelse{\boolean{showcomments}}{{\color{blue}{SK: #1}}}{}}
\newcommand{\SKtodo}[1]{
    \ifthenelse{\boolean{showcomments}}
    {{\color{red}{SK: #1}}}
    {}
}

\newcommand{\mg}{\mathfrak{g}}
\newcommand{\so}[1]{\mathfrak{so}\hspace{-0.15em}\left(#1\right)}
\newcommand{\su}[1]{\mathfrak{su}\hspace{-0.15em}\left(#1\right)}
\renewcommand{\u}[1]{\mathfrak{u}\hspace{-0.15em}\left( #1 \right)}
\newcommand{\gl}[1]{\mathfrak{gl}\hspace{-0.15em}\left( #1 \right)}

\newcommand{\lie}[1]{\left\langle #1 \right\rangle_{Lie}}
\newcommand{\Gxyc}[1]{\mathcal{G}_{XY}^C(#1)}
\newcommand{\Gxyp}[1]{\mathcal{G}_{XY}^P(#1)}
\newcommand{\Gxyk}[1]{\mathcal{G}_{XY}^K(#1)}
\newcommand{\Gxykz}[1]{\mathcal{G}_{XY,Z}^K(#1)}
\newcommand{\Gxycz}[1]{\mathcal{G}_{XY,Z}^C(#1)}
\newcommand{\Gxypz}[1]{\mathcal{G}_{XY,Z}^P(#1)}

\newcommand{\Gxyczzz}[1]{\mathcal{G}_{XY,Z,ZZ}^C(#1)}
\newcommand{\Gz}[1]{\mathcal{G}_{Z}(#1)}
\newcommand{\Gzz}[1]{\mathcal{G}_{ZZ}(#1)}
\newcommand{\Gzzz}[1]{\mathcal{G}_{Z,ZZ}(#1)}

\newcommand{\gxyp}[1]{
    \ifthenelse{\equal{#1}{}}
    {\mfg_{XY}^{P}}
    {\mfg_{XY}^{P}(#1)}
}

\newcommand{\gxyk}[1]{
    \ifthenelse{\equal{#1}{}}
    {\mfg_{XY}^{K}}
    {\mfg_{XY}^{K}(#1)}
}

\newcommand{\gxyc}[1]{
    \ifthenelse{\equal{#1}{}}
    {\mfg_{XY}^{C}}
    {\mfg_{XY}^{C}(#1)}
}

\newcommand{\gxypz}[1]{
    \ifthenelse{\equal{#1}{}}
    {\mfg_{XY,Z}^{P}}
    {\mfg_{XY,Z}^{P}(#1)}
}

\newcommand{\gxypzz}[1]{
    \ifthenelse{\equal{#1}{}}
    {\mfg_{XY,ZZ}^{P}}
    {\mfg_{XY,ZZ}^{P}(#1)}
}

\newcommand{\gxycz}[1]{
    \ifthenelse{\equal{#1}{}}
    {\mfg_{XY,Z}^{C}}
    {\mfg_{XY,Z}^{C}(#1)}
}

\newcommand{\gxyczzz}[1]{
    \ifthenelse{\equal{#1}{}}
    {\mfg_{XY,Z,ZZ}^{C}}
    {\mfg_{XY,Z,ZZ}^{C}(#1)}
}

\newcommand{\gxykz}[1]{
    \ifthenelse{\equal{#1}{}}
    {\mfg_{XY,Z}^{K}}
    {\mfg_{XY,Z}^{K}(#1)}
}

\newcommand{\gxyce}[1]{
    \ifthenelse{\equal{#1}{}}
    {\mfg_{XY}^{C,E}}
    {\mfg_{XY}^{C,E}(#1)}
}

\newcommand{\gxyco}[1]{
    \ifthenelse{\equal{#1}{}}
    {\mfg_{XY}^{C,O}}
    {\mfg_{XY}^{C,O}(#1)}
}

\newcommand{\gxyzzzk}[1]{
    \ifthenelse{\equal{#1}{}}
    {\mfg_{XY,Z,ZZ}^{K}}
    {\mfg_{XY,Z,ZZ}^{K}(#1)}
}

\newcommand{\AB}[2]{\widehat{#1}_{#2}}
\newcommand{\ABopp}[2]{\widehat{#1}_{#2}^-}

\newcommand{\Popp}[1]{P_{#1}^-}
\newcommand{\Qopp}[1]{Q_{#1}^-}
\newcommand{\Dopp}[1]{D_{#1}^-}
\newcommand{\copp}[1]{c_{#1}^-}
\newcommand{\bopp}[1]{b_{#1}^-}

\newcommand{\Zbar}[1]{Z_{\overline{#1}}}
\newcommand{\Zn}{Z^{\otimes n}}
\newcommand{\bigP}[1]{\left(#1\right)}
\newcommand{\bigB}[1]{\left\{#1\right\}}
\newcommand{\com}[1]{\left[ #1 \right]}
\newcommand{\Span}[1]{ \text{span}\left( #1 \right) }

\newcommand{\hlincon}{\mathfrak{h}_{\text{ker}\left( \text{ad}_{Z^{+}} \right)}}
\newcommand{\glincon}{\mathfrak{g}_{\text{ker}\left( \text{ad}_{Z^+} \right)}}
\newcommand{\redglincon}{\tilde{\mathfrak{g}}_{\text{ker}\left( \text{ad}_{Z^+} \right)}}

\newcommand{\gxykdecompodd}{\bigoplus\limits_{k=1}^{\lfloor n/2 \rfloor} \su{\binom{n}{k}}}
\newcommand{\gxykdecompeven}{\su{\frac{1}{2}\binom{n}{n/2}}^{\oplus 2}\oplus \bigoplus\limits_{k=1}^{n/2 - 1} \su{\binom{n}{k}}}
\newcommand{\gxyke}[1]{\mfg_{XY}^{K,E}(#1)}
\newcommand{\gxyko}[1]{\mfg_{XY}^{K,O}(#1)}

\DeclarePairedDelimiter\ceil{\lceil}{\rceil}
\newcommand{\floor}[1]{\left\lfloor #1 \right\rfloor}

\usepackage{array}
\newcommand\Tstrut{\rule{0pt}{2.6ex}}         %
\newcommand\Bstrut{\rule[-1.5ex]{0pt}{0pt}}   %

\theoremstyle{definition}
\newtheorem{definition}{Definition}[section]

\newtheorem{lemma}{Lemma}
\newtheorem{proposition}{Proposition}
\newtheorem{corollary}{Corollary}

\newtheorem{remark}{Remark}
\newtheorem{conjecture}{Conjecture}
\newtheorem{theorem}{Theorem}

\crefname{manualtheoreminner}{Theorem}{Theorems}

\newcommand{\numstarts}{10}
\newcommand{\numinsts}{84}
\newcommand{\numsteps}{100}

\definecolor{cobaltblue}{RGB}{78,126,192}
\definecolor{cobalttint}{RGB}{230,240,255}

\newtcolorbox[auto counter,
crefname={mybluebox}{blueboxes}]%
{mybluebox}[2][]{colback=cobalttint,colframe=cobaltblue,fonttitle=\bfseries,
title=Procedure 1: #2,#1}

\usepackage{cuted}

\usepackage{titlesec}
\titleformat*{\section}{\raggedright\Large\bfseries\sffamily}
\titleformat*{\subsection}{\raggedright\large\bfseries\sffamily}
\titleformat*{\subsubsection}{\raggedright\bfseries\sffamily}

\begin{document}
\raggedbottom 

\title{The Lie Algebra of XY-mixer Topologies and Warm Starting QAOA for Constrained Optimization}

\author[,1,2]{Steven Kordonowy\thanks{Email: \texttt{skordono@ucsc.edu}}} \author[,1]{Hannes Leipold\thanks{Email: \texttt{hleipold@fujitsu.com}}}

\date{December, 2025}
\affil[1]{Fujitsu Research of America, Santa Clara, CA, USA}
\affil[2]{University of California at Santa Cruz, CA, USA}

\maketitle

\begin{strip}
\vspace{-8em}
\begin{abstract}
The XY-mixer is widely used in quantum computing, particularly in variational quantum algorithms like the Quantum Alternating Operator Ansatz (QAOA). It is effective for solving Cardinality Constrained Optimization problems, a broad class of NP-hard tasks. We provide explicit decompositions of the dynamical Lie algebras (DLAs) for various XY-mixer topologies. When these DLAs decompose into polynomial sized Lie algebras, they are efficiently trainable, such as in cycle XY-mixers with arbitrary RZ. In contrast, all-to-all XY-mixers or including arbitrary RZZ gates leads to exponentially large DLAs, making training intractable. We warm-start optimization by pre-training on smaller, polynomially-sized DLAs via gate generator restriction, improving convergence and yielding better solution quality, for both shared-angle and multi-angle QAOA instances. Multi-angle QAOA is also shown to be more performant than shared-angle QAOA on the instances considered.
\end{abstract}

\vspace{1em}
\end{strip}

\section*{Introduction}\label{sec:intro}

    \begin{figure*}[t]
    \centering 
    \includegraphics[width=0.9\textwidth]{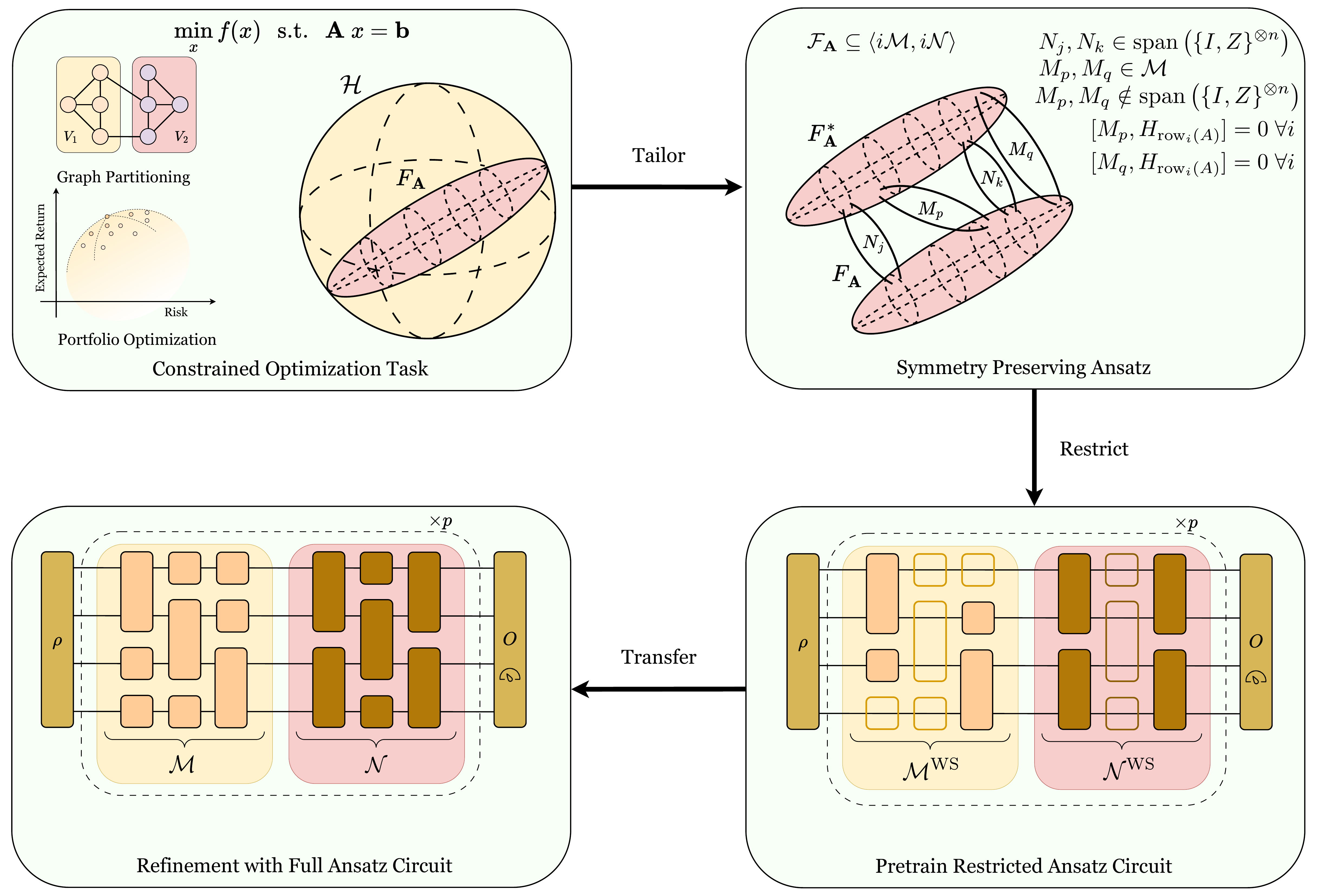}
    \caption{\textbf{Warm starting QAOA for Constrained Optimization.} Given a constrained optimization problem, such as \texttt{Portfolio Optimization}, we (1) tailor a collection of mixing and phase-separating operators, (2) restrict the ansatz and pretrain, (3) transfer the learned parameters for the full circuit, and (4) refine through training with the full parameter set.}
    \label{fig:flowchart}
    \end{figure*}

    Variational quantum algorithms (VQAs)~\cite{cerezo_variational_2021,bharti2022noisy} utilize digital quantum computers to solve a large range of tasks, from simulating fermionic systems~\cite{stanisic2022observing} to solving hard optimization problems~\cite{farhi_quantum_2014,khairy_learning_2020,rieffel2020xymixers,niroula2022constrained}. Parameterized XY operators (XY-mixers) are an important tool for modern quantum algorithms, particularly for constrained optimization problems~\cite{divincenzo2000physical,knill2005quantum,preskill_quantum_2018}. The XY-mixer is generated by the two-body interaction $\ketbra{01}{10} + \ketbra{10}{01}$~\cite{crooks2020gates}, thereby preserving Hamming weight between basis states. This is particularly useful for cardinality constrained optimization problems~\cite{cai2008parameterized}, in which feasible states are known to have a fixed Hamming weight and are friendly to implement on near term devices~\cite{blumel2021power,abrams2019implementation}. When encoding optimization problems into \textit{parameterized quantum circuits} (PQCs), it is important to account for the allowed node interactions, which are often restricted to simple topologies such as one-dimensional paths or cycles, or two-dimensional all-to-all cliques.

    To solve optimization tasks with VQAs, Ref.~\cite{farhi_quantum_2014} introduced the Quantum Approximate Optimization Algorithm (QAOA1), inspired by quantum annealing~\cite{kadowaki_quantum_1998,farhi_quantum_2001} for solving Quadratic Unconstrained Binary Optimization (QUBO)~\cite{kochenberger_unconstrained_2014} problems. Through the introduction of penalty terms, constrained optimization problems can be reformulated as QUBOs and solved by VQAs~\cite{cerezo_variational_2021} or quantum annealing~\cite{farhi_quantum_2000,albash_adiabatic_2018,ronnow_defining_2014}. Inspired by constrained quantum annealing ~\cite{hen_quantum_2016,hen_driver_2016,kudo2018constrained,torggler2019quantum,chancellor2019domain,leipold_constructing_2021,leipold_quantum_2022}, Ref.~\cite{hadfield_quantum_2019} introduced the Quantum Alternating Operator Ansatz (QAOA) as a VQA. The original QAOA1 proposal utilizes the X-mixer, which maintains no constraint symmetry and imposes a list of constraints through designs such as penalty terms, while QAOA may utilize \textit{specialized mixers} to maintain the symmetry imposed by a list of constraints~\cite{leipold2024imposing}. Any linear constraint optimization task can be written:
    \begin{align}\label{eq:op_lincon}
    \min_{x} \; f(x) \,\text{ such that }\, \bm{A} x = \bm{b}, \; x \in \{ 0, 1 \}^n
    \end{align}
    
    \noindent for matrix $ \bm{A} $, vector $ \bm{b} $, and cost function $ f(x) $. The QAOA attempts to minimize the cost $ f(x) $ within the constraint space $ \bm{A} x = \bm{b} $ by (1) encoding the problem into the language of Hamiltonians and (2) applying parameterized mixers and phase-separating operators to approximate the ground state, ensuring the exploration remains within the associated constraint space.
    
    We analyze the \textit{dynamical Lie algebra} (DLA)~\cite{larocca2024review} of these XY-mixer topologies as well as those that include $R_{Z} $ and $ R_{ZZ} $ rotation gates, since these can be utilized by QAOA for constrained optimization. The DLA is the real vector space spanned by the nested commutation of the generators of the parameterized gates in the circuit~\cite{larocca2022diagnosing}. The associated DLA of a PQC, $ \mfg $, has been shown to play a central role in determining the trainability~\cite{holmes_connecting_2022}, simulatability~\cite{cerezo_does_2024}, and expressibility of the PQC. The variance of the cost $\mathcal{L}$ with respect to parameters $\bmt$ of the PQC scales inversely with the dimension of the DLA, roughly $ \text{Var}_{\bmt} \left( \mathcal{L}(\bmt) \right) \sim \mathcal{O}(\dim^{-1} 
    (\mfg)) $~\cite{ragone2024lie, fontana2024characterizing}. A PQC with an exponential DLA, such as $t$-design circuits~\cite{harrow2009random,dankert2009exact,hunter2019unitary}, faces the Barren Plateau phenomenon (BP)~\cite{mcclean2018barren,holmes_connecting_2022,larocca2024review}, which is defined by having exponentially small gradients. This makes training VQAs computationally inefficient as the flatness requires exponentially many samples to determine the direction for the gradient updates (see Ref.~\cite{cerezo_does_2024} for caveats). 

    A simple way for a VQA to be free of the BP issue is when the DLA has polynomial size \cite{cong2019quantum,ragone_representation_2022,meyer2023exploiting,goh_lie-algebraic_2023,nguyen2024theory}. In these instances, the dynamics of the system can be simulated classically in polynomial time given an initial density operator or an observable in the DLA \cite{somma2005quantum,somma2006efficient,goh_lie-algebraic_2023}. In particular, when a DLA is isomorphic to a polynomially sized Lie algebra, certain calculations, such as expectation values, can be vastly simplified from the original operator space over $\C^{2^n \times 2^n}$ to the relevant space over $\C^{poly(n) \times poly(n)}$. This leads to a central dilemma of training VQAs: the generators of a VQA with superpolynomial DLAs often face BPs while generators with polynomial DLAs are often simulable by classical devices~\cite{cerezo_does_2024}.

    Warm-starting circuits is a viable middle ground strategy between the inefficiency of directly training highly expressive circuits and the classical simulability of simpler circuit classes. By identity initializing - and \textit{freezing} - gates that lead to high expressivity, we can train a restricted set of gates inside a polynomial DLA. Specifically, we train the XY-mixer cycle with $R_{Z}$ (but no $R_{ZZ}$), since we show this has a polynomially sized DLA. After optimizing the angles in the restricted circuit, we then use them as initial angles for the full circuit and ``turn on'' the $ R_{ZZ} $ gates. This is particularly fruitful for gate sets where a reasonable restriction leads to a meaningful yet polynomially sized DLA. We showcase that this warm-starting approach outperforms random initialization on three important optimization problems: \texttt{Portfolio Optimization}, \texttt{Sparsest k-Subgraph}, and \texttt{Graph Partitioning}. See \cref{fig:flowchart} as an overview of our strategy and the Methods section for an overview implementation. As we show in the Results section, constrained optimization are fruitful targets for such warm starting. \cref{fig:venndlas} shows our DLA classification and how it informs warm starting QAOA.

    Due to challenges in training large quantum circuits, warm-starting is an important and expanding field of inquiry for VQAs~\cite{cerezo_variational_2021,cerezo_does_2024,mhiri2025unifying}, including for QAOA. Interpolative parameter tuning~\cite{crooks_performance_2018,zhou_quantum_2020,liu2022layer,lee2023depth,leipold2024train} can be seen as repeated warm starting. Through it's connection to adiabatic evolution, initial parameters for QAOA can be selected based on standard Hamiltonian interpolation functions, such as the linear ramp~\cite{mbeng2019quantum,wurtz2021maxcut}. Good initial states can be used to alter the QAOA protocol to act as a warm starting~\cite{egger_warm-starting_2021,tate2023bridging,tate2023warm}. Our approach is related to restricted DLA approaches to warm starting, such as for \texttt{MaxCut}~\cite{goh_lie-algebraic_2023}; by identifying algebraic substructures of the full-circuit DLA that are efficiently trainable, we find optimal parameters inside the restricted manifold that then inform the initial parameters of the full gate set. This strategy scales efficiently in system size.

    \begin{figure}[!t]
        \centering 
        \includegraphics[width=1.0\columnwidth]{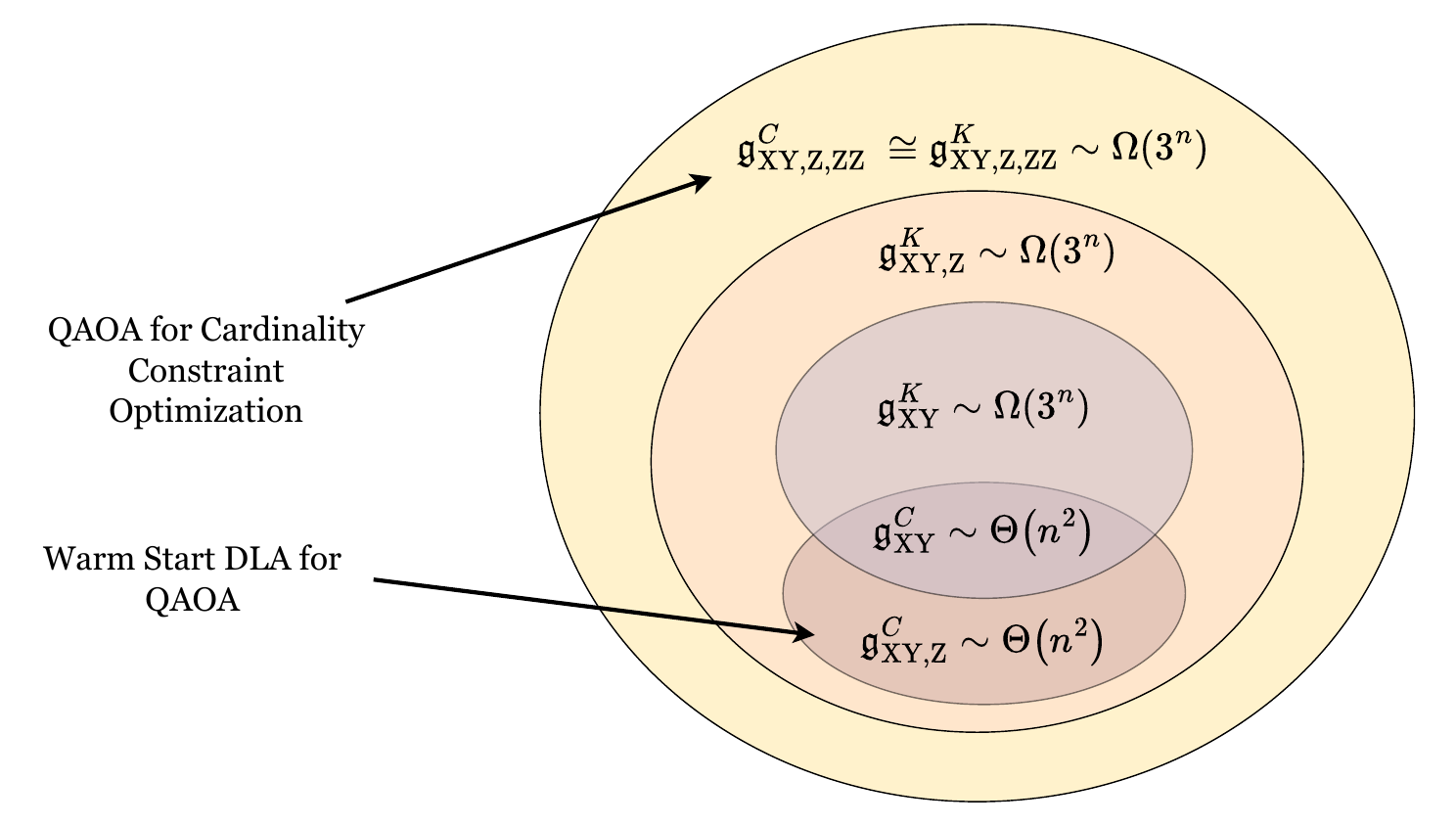}
        \caption{\textbf{Sub-algebras of XY Topologies.} A Venn Diagram comparing the space of linear operators of different XY topologies studied in this manuscript. }
        \label{fig:venndlas}
    \end{figure}

\section*{Results}\label{sec:results}

    \subsection*{Model}\label{subsec:model}
        A parameterized quantum circuit (PQC) on a Hilbert space $\H$ over $n$ qubits is defined by three components: the circuit, the observable measured, and the starting state. We denote the circuit's \textit{gate generators} as $\mcg = \set{G_1, \dots, G_r}$ which are a set of Hermitian operators over $\H$. A depth-$P$ layered circuit is defined as 
        \begin{equation}\label{eq:pqc}
            U(\bmt) = \prod_{p = 1}^P \prod_{k = 1}^K e^{i \,  \theta_{p,k} G_{p,k}},
        \end{equation}
        \noindent where each $G_{p,k} \in \mcg$ is a gate generator and $\bmt \in \R^{P \times K}$ are the trainable parameters. Each $p$ corresponds to a layer in the circuit. For periodic circuits, each layer has a repeating structure. 
        The operator $U(\bmt)$ determines the possible dynamics of the system where each layer can be viewed as one time step. 
        The initial state $\rho_0$ is a density operator on $\H$, such as a pure state $ \rho_0 = \ketbra{\psi_{0}}{\psi_{0}} $, and should be efficiently preparable. Using these components, the parameterized circuit generates the state $ U(\bmt) \, \rho_0 \, U(\bmt)^\dagger $. We are interested in the expectation value of this state with respect to the observable $O$ which in turn defines the parameterized loss:
        \begin{equation} \label{eq:loss}
            \mathcal{L}(\bmt) := \text{Tr}\left( O\, U(\bmt) \, \rho_0 \, U(\bmt)^\dagger \right). 
        \end{equation}

        \noindent Parameters of the gate generators are then trained through gradient or gradient-free methods~\cite{biamonte2017quantum,mitarai2018quantum,bonet2023performance,arrasmith_effect_2021, abbas2023quantum}.
        
        An important collection of operators to utilize as part of $\mathcal{G}$ are the 1- and 2-term Pauli-$Z$ interactions
        \begin{align}
            \Gz{n} &:= \bigB{Z_j : j \in [n]},  \\
            \Gzz{n} &:= \bigB{Z_jZ_k  : 1 \leq, j < k \leq n}, \\ 
            \Gzzz{n} &:= \Gz{n} \cup \Gzz{n}.
        \end{align}

        We refer to shared angle QAOA as SA-QAOA due to parameter sharing between generators. We focus on multi-angle QAOA (MA-QAOA) due its outperformance of SA-QAOA, see \cref{app:warmstart_details} for a comparison. MA-QAOA~\cite{herrman2022multi} is a VQA that stipulates the following structure for \cref{eq:loss}:
        \begin{itemize}[itemsep=1pt, topsep=2pt]
            
            \item[] \textit{Cost Hamiltonian}: the observable $H_{f}$ encodes the cost function $f(x)$ from our optimization task \cref{eq:op_lincon} in the computational basis, such that $ H_{f} \ket{x} = f(x) \ket{x} $ for any $x \in \{0,1\}^n$. Typically, $ H_{f} \in \text{span}\left(\mathcal{G}_{Z,ZZ}\right) $~\cite{lucas_ising_2014}.
            
            \item[] \textit{QAOA Ansatz}: the gate generators of the PQC are split into two groups.
            \begin{itemize}[itemsep=1pt, topsep=2pt]
                \item[] Phase separating generators: a collection $\mathcal{N} $ such that each generator $ N \in \mathcal{N} $ is diagonal in the computational basis. Typically, these operators are selected from one- and two-body interactions $ \mathcal{N} \subseteq \mathcal{G}_{Z,ZZ} $.
                \item[] Mixer generators: a collection $\mathcal{M}$ such that each generator $ M \in \mathcal{M} $ commutes with the \textit{embedded constraint operators}. Given a constraint row $A_{j} $, $ H_{A_{j}} \ket{x} = (A_{j} x) \ket{x} $ and $ [ M, H_{A_{j}} ] = 0 $ for $ M \in \mathcal{M} $. Then any application $ e^{i \, \sum_{j} \gamma_{j} M_{j} } $ evolves the quantum state in the constraint space and so we do not need to mix with the $\{ H_{A_j}\}_j$ directly~\cite{hen_quantum_2016,hadfield_quantum_2017,shaydulin_classical_2020,leipold2024imposing}.
            \end{itemize}

            \item[] \textit{Initial State}: An initial state $\ket{\psi_0}$ for the quantum system within the constraint space~\cite{leipold2024imposing}, such that $H_{A_{j}} \ket{\psi_{0}} = b_{j} \ket{\psi_{0}} $ for every constraint row $A_{j}$ and corresponding constraint value $b_{j}$. 
            
        \end{itemize}

        \noindent Since $ \ket{\psi_{0}} $ is in the \textit{embedded feasible subspace} of each $H_{A_{j}}$ and both generators of the QAOA ansatz \textit{commute} with each $H_{A_{j}}$, the evolution of the wavefunction is restricted to be within the embedded feasible subspace. See \cref{app:warmstart_details} for further exposition. Given the QAOA ansatz, \cref{eq:pqc} can be rewritten:
        \begin{align}\label{eq:ma_qaoa}
            U(\bm{a}, \bm{\beta}) = \prod_{p=1}^{P} \prod_{M_{k} \in \mathcal{M}} e^{i \, \alpha_{pk} \, M_{k}} \prod_{N_{l} \in \mathcal{N}} e^{i \, \beta_{pl} N_{l}} ,
        \end{align}
        with $ \bm{\alpha} \in \R^{|\mathcal{N}| \, P}, \bm{\beta} \in \R^{|\mathcal{M}| \, P} $ such that the PQC is parameterized by $ \bm{\theta} = (\bm{\alpha}, \bm{\beta}) $.

        The canonical example of QAOA is given by phase separators in $\mathcal{G}_{Z,ZZ}$ and the $X$-mixer $ \{  X_{j} \}_{j=1}^{n} $ with \textit{no constraint} ($\bm{A} = 0$), as introduced in Ref.~\cite{farhi_quantum_2014}, although any constrained optimization task can be converted into an unconstrained task through penalty terms (see \cref{app:warmstart_details}). Tailoring QAOA to constrained optimization has shown significant practical improvement over using the X-mixer with penalty terms~\cite{hodson2019portfolio,cook2020quantum,leipold_tailored_2022,golden2023numerical,boulebnane2023peptide} as well as noise resistance~\cite{streif_quantum_2021}.

        In the case of \textit{cardinality constraint optimization}, $\textbf{A} = [1,\ldots,1]$ has a single row constraint, the global linear constraint. The associated embedded constraint operator is $ H_{A} = \frac{1}{2} \left( (n - 2 \, k) \, I + \sum_{j=1}^{n} Z_{j} \right) $. An operator that irreducibly commutes with $H_A$ is the \textit{XY-mixer}:
        \begin{equation}\label{eq:XY_jk}
            XY_{j,k} = \frac{1}{2}(X_j X_k + Y_j Y_k) = \ketbra{10}{01}_{j,k} + \ketbra{01}{10}_{j,k}.
        \end{equation}

        \begin{figure*}[!t]
            \centering
            \begin{tabular}{c c c c c c c c}
            \includegraphics[trim=0.5cm 0 0.5cm 0, clip, width=0.09\textwidth]{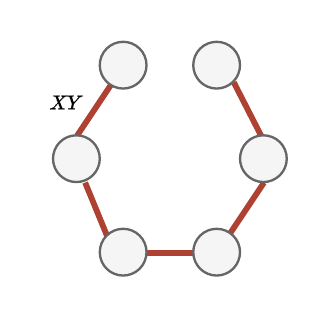} & 
            \includegraphics[trim=0.35cm 0 0.35cm 0, clip, width=0.09\textwidth]{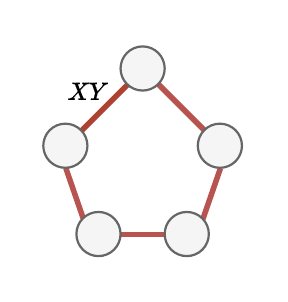} & 
            \includegraphics[trim=0.5cm 0 0.5cm 0, clip, width=0.09\textwidth]{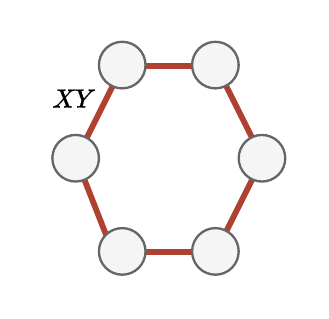} &
            \includegraphics[trim=0.6cm 0 0.6cm 0, clip, width=0.09\textwidth]{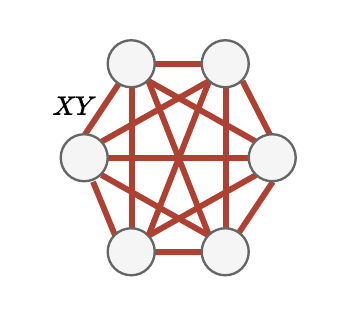} &
            \includegraphics[trim=0.5cm 0 0.5cm 0, clip, width=0.09\textwidth]{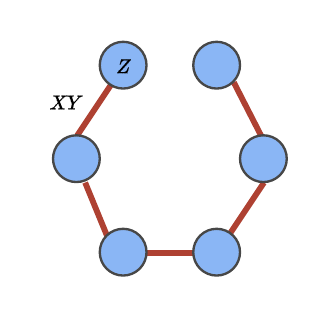} & 
            \includegraphics[trim=0.6cm 0 0.6cm 0, clip, width=0.09\textwidth]{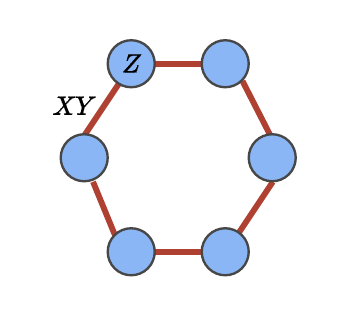} & 
            \includegraphics[trim=0.6cm 0 0.6cm 0, clip, width=0.09\textwidth]{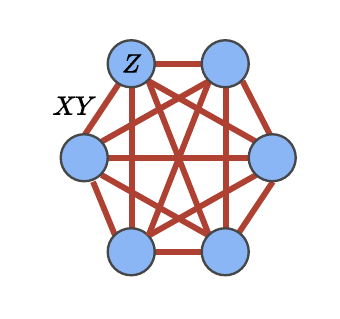} &
            \raisebox{-0.25em}{
            \includegraphics[trim=1cm 0 0.30cm 0, clip, width=0.115\textwidth]{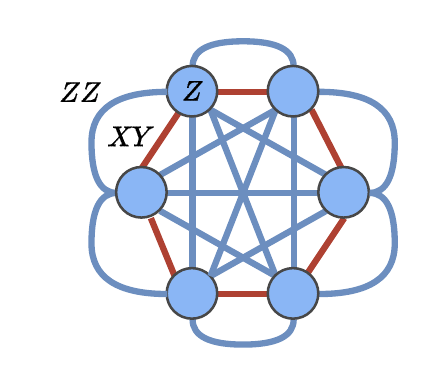}} \\ 
            $ \gxyp{} $ & $ \gxyco{} $ & $ \gxyce{} $ & $ \gxyk{} $ & $\gxypz{} $ & $ \gxycz{} $ & $ \gxykz{}$ & $\gxyczzz{} $ \Bstrut\\
            (a) & (b) & (c) & (d) & (e) & (f) & (g) & (h) \Bstrut
            \end{tabular}
            \caption{\textbf{Depictions of XY Topologies.} Diagrammatic representations of different XY-mixer associated generators and their associated Lie algebra. 
            }
            \label{fig:xytopdiags}
        \end{figure*} 

        \begin{figure}[!t]
            \centering
               \begin{tikzpicture}[node distance = 1.5cm, thick]%
                \node (1) {$\gxyp{n}$};
                \node (2) [below left =1cm and 1cm of 1] {$\gxyce{n}$};
                \node (3) [below=1cm of 1] {$\gxyco{n}$};
                \node (4) [below right =1cm and 1cm of 1] {$\gxypz{n}$};
                \node (5) [below=1cm of 3] {$\gxycz{n}$};
                \draw[->] (1) -- (2);
                \draw[->] (1) -- (3);
                \draw[->] (1) -- (4);
                \draw[->, dashed] (3) -- (4);
                \draw[->] (2) -- (5);
                \draw[->] (3) -- (5);
                \draw[->] (4) -- (5);
            \end{tikzpicture}%
            \caption{\textbf{Containment diagram for the polynomial-sized dynamical XY Lie algebras.} An arrow indicates a Lie sub-algebra. The dashed arrow indicates that $\mfg_{XY}^{C,O}(n)$ is isomorphic to the $\su{n}$ component of $\mfg_{XY,Z}^{P}(n)$ but the basis elements and therefore associated subspaces in $\su{2^n}$ are different.}
            \label{fig:dla-containment}
        \end{figure}
        \noindent The operator $XY_{j,k}$ preserves Hamming weight basis states; over qubits $j$ and $k$, it acts like identity on the states $\ket{00}_{j,k}, \ket{11}_{j,k}$ and bit-swaps the states $\ket{10}_{j,k}, \ket{01}_{j,k}$, thereby preserving of the bit count.

        Note that the canonical starting state for the QAOA, the uniform superposition state $\frac{1}{\sqrt{2^n}} \sum_{x\in\{0,1\}^n} \ket{x} $, is distributed across several constraint spaces and therefore fails to satisfy the constraint condition. Instead, we restrict to a uniform superposition over basis states of fixed Hamming weight $w$, known as the Dicke state:
        \begin{equation}
            \ket{D^n_w} = \frac{1}{\sqrt{\binom{n}{w}}}\sum_{x : \abs{x} = w} \ket{x},
        \end{equation}
        
        \noindent Typically, $w = \Theta(n)$, such as $w = n/2$. Dicke states can be prepared efficiently \cite{dicke_prep}. Then the loss function in \cref{eq:loss} for our VQA, which we wish to minimize to find (approximate) solutions to our optimization task, is given by:
        \begin{align}\label{eq:loss_qaoa}
            \mathcal{L}(\bm{\theta}) &= \bra{D^n_{n/2}} U(\bmt)^{\dagger} \, H_{f} \, U(\bmt) \ket{D^n_{n/2}}. \nonumber
        \end{align}

        The \textit{dynamical Lie algebra} (DLA) of a set of Hermitian operators $\mcg$, denoted $\mfg := \text{span}_{\R} \langle i \mathcal{G} \rangle_{Lie}$, is the linear span of nested commutations of elements from $\mcg$. Understanding $\mfg$ is useful as it encapsulates all of the unitaries that a quantum circuit can realize. Consider any gate generators $G_1, G_2 \in \mcg$ corresponding to a PQC. There exists some element $i G \in \mfg$ such that
        \begin{equation}
            e^{iG_1}e^{iG_2} = e^{i G}.
        \end{equation}

        The dimension of the DLA of QAOA with the X-mixer has been shown to typically be exponential~\cite{larocca2022diagnosing}, although highly structured instances may not~\cite{allcock_dynamical_2024}. For one- and two-dimensional Pauli spin topologies, The dimension of the DLA was found to grow $\mathcal{O}(n), \mathcal{O}(n^2),$ or $ \mathcal{O}(4^{n})$~\cite{wiersema_classification_2023}.

    \subsection*{XY-Mixer Topologies}\label{subsec:def_XY_mixer_sets}    
    
        The following generator sets are studied in this manuscript:
        \begin{enumerate}
            \item $\Gxyp{n} := \bigB{XY_{1,2}, \dots, XY_{n-1,n}}$.
            \item $\Gxyc{n} := \Gxyp{n} \cup \bigB{XY_{n,1}}$.
            \item $\Gxypz{n} := \Gxyp{n}  \cup \Gz{n}$.
            \item $\Gxycz{n} := \Gxyc{n} \cup \Gz{n}$.
            \item $\Gxyk{n} := \bigB{XY_{j,k} : 1 \leq j < k \leq n}$.
            \item $\Gxykz{n} := \Gxyk{n} \cup \Gz{n} $.
            \item $\Gxyczzz{n} := \Gxyc{n} \cup \Gz{n} \cup \Gzz{n} $.
        \end{enumerate}
    
         The super-script denotes the topology of the $XY$-terms: a path (P), a cycle (C), or a clique (K). The subscript corresponds to the type of operators used. Generators $(1)-(4), (6)$ correspond to 1-dimensional interaction graphs of $XY$ interactions. Note that in $(7)$, we have all-to-all $ZZ$ connectedness so the full interaction graph is technically not 1-dimensional though the $XY$-connectivity is. An example of each is depicted in \cref{fig:xytopdiags}.

    \subsection*{DLA Decompositions}\label{subsec:def_dlas_decomps}
    
        DLAs generated from gate generators in the previous section are spans of nested commutations of ($i$ times) these elements, such as $[ i XY_{3,4}, i XY_{4,5}]$ or $[ i XY_{9,10}, [ i Z_{10}, i XY_{10,11} ] ]$. If there is no overlap of indices, then the operators commute. The notation $\mfg_\alpha^\Gamma = \lie{i \mcg_\alpha^\Gamma}$ denotes the DLA resulting from the nesting commutation procedure of a given generator set.
        
        We begin by describing the DLAs with polynomial dimension.
    
        \begin{theorem}\label{thm:poly_dlas_decomp}
                    \begin{align}
                        &\gxyp{n} \cong \so{n}\label{eq:gxyp_decomp}\\
                        &\gxyc{n} \cong \begin{cases}
                            \so{n} \oplus \so{n} &\text{even } $n$\\
                            \su{n} &\text{odd } $n$
                        \end{cases}  \label{eq:gxyc_decomp}\\
                        &\gxypz{n}  \cong \u{1} \oplus \su{n} \label{eq:gxypz_decomp}\\
                        &\gxycz{n}  \cong \u{1} \oplus \su{n} \oplus \su{n}\label{eq:gxycz_decomp}
                    \end{align}
    
                 \noindent In particular, for one-dimensional systems with 2-body $XY$ interactions and single $Z$ terms, the DLA is of size $\mathcal{O}(n^2)$. 
             \end{theorem}

        All of these DLAs are decomposed into one or more copies of $\so{n}$ or $\su{n}$ (possibly along with a one-dimensional center). Here, $\su{n}$ ($\so{n}$) is the Lie algebra of skew-Hermitian (real antisymmetric) operators. 
        
        Once we consider either all-to-all XY-mixer connectivity or allow $ZZ$ interactions, the DLAs grow exponentially in the number of qubits.
        
        \begin{theorem}\label{thm:dlasummmary-big}
            The following gate generators have a DLA is of size $\Omega(3^n)$.
            \begin{enumerate}
                \item[(a)] 1-dimensional systems with single Pauli $Z$ terms, neighboring $XY$ interactions, and neighboring $ZZ$ interactions. This is true for open and closed boundary conditions.
                \item[(b)] Fully-connected systems (i.e. the interaction graph is complete) with neighboring $XY$ terms. Adding single qubit $Z$ or neighboring $ZZ$ interactions only serves to increase the DLA.
            \end{enumerate} 
        \end{theorem}

        In the polynomially sized DLAs, we rigorously understand the Lie algebra decompositions. In the remaining cases, we have formal arguments for the $\Omega(3^n)$ lower bound but only provide conjectures for the explicit decomposition. We need more sophisticated tools when dealing with the larger DLAs than simply directly computing the basis elements. In particular, from computational verification for small $n$ and based on the stability of classification with increasing $n$, we have the following central conjecture.  
        \begin{conjecture}\label{conj:expliealg}
            \begin{align}
                &\gxyk{n} \cong \begin{cases}
                    \gxykdecompeven &\text{even } n \vspace{0.5em}\\
                    \gxykdecompodd & \text{odd } n \vspace{0.5em}
                \end{cases}\\
                &\gxykz{n} \cong \u{1} \oplus  \bigoplus_{k=1}^{n-1} \su{\binom{n}{k}}  \\ 
                &\gxyczzz{n} \cong \u{1}^{\oplus 2} \oplus \bigoplus_{k=1}^{n-1} \su{\binom{n}{k}}  
            \end{align} 

            \noindent These DLAs have dimension $ 4^{n - \Theta\left( \log(n) \right)} $. 
        \end{conjecture}

        \noindent The conjecture postulates that $\gxykz{}$ and $\gxyczzz{}$ both respect Hamming weight subspaces of fixed $k$. That is, they contain the maximal Lie algebra such that every element commutes with $ \sum_{j=1}^{n} Z_{j} $ (up to a center), described further in Methods. This means that the generator sets $\mathcal{G}_{XY,Z}^{K}$ and $\mathcal{G}_{XY,Z,ZZ}^{C}$ can, up to phasing, express all unitaries that commute with $\sum_{j=1}^{n} Z_{j} $. For their centers, both contain $\sum_{j=1}^{n} Z_{j} $ while the center of $\gxyczzz{}$ also contains $\sum_{j=1}^{n} \sum_{j<k}^{n} Z_{j} Z_{k} $. 

        \cref{fig:venndlas} shows a Venn Diagram for some of the polynomial and exponential sized DLAs studied in manuscript while \cref{fig:dla-containment} shows the containment hierarchy for the polynomial sized DLAs. QAOA with the XY-mixer therefore has a Barren Plateau~\cite{larocca2024review} while restricting the VQA to $\gxycz{}$ is Barren Plateau free. \cref{app:controlgrad} reviews how the dimension of the DLA directly impacts trainability.

    \subsection*{Warmstarting QAOA on Constrained Optimization}\label{subsec:warmstart_qaoa}

    To demonstrate how warm starting inside the DLA $\gxycz{}$ can dramatically improve QAOA inside the DLA $\gxyczzz{}$, we develop the warm starting procedure for QAOA on three optimization tasks in the Methods section. Here we report the results from numerical experiments. %
    
    For comparison, we consider two performance metrics for each problem: the approximation ratio and the success probability. Let $ E_{\text{min}} $ ($ E_{\text{max}}$) be the minimum (maximum) energy of the system and $ \langle H_{f} \rangle $ the expected energy of the cost Hamiltonian $H_{f}$ after applying the PQC $U(\theta)$ to the input state. The approximation ratio is the ratio of the difference in the expected energy with the maximum energy over the total energy range, 
    \begin{align} 
        \text{AR} = \left(\langle H_f \rangle - E_{\text{max}} \right)/ \left( E_{\text{min}}-E_{\text{max}} \right);
    \end{align} 
    
    \noindent this ratio reaches its minimum value of $0$ when $\langle H_f \rangle = E_{\text{max}}$ and its maximal value of $1$ when $\langle H_f \rangle = E_{\text{min}}$.
    Let $P_{\text{min}}$ be the projector into the solution space, such that $P_{\text{min}}^2 = P_{\text{min}}$ and $\text{Tr}(H_f P_{\text{min}}) = n_s E_{\text{min}}$ where $n_s$ is the number of minimum energy states. %
    The success probability is the support of the final wavefunction over the solution space:
    \begin{align}
        \text{Pr}[\text{Success}] = \bra{\psi_f} P_{\text{min}} \ket{\psi_f}. 
    \end{align}

    \texttt{Portfolio Optimization} is an optimization task for a financial portfolio given a risk profile and historical data on the returns and correlations between the assets. Given $n$ assets $ \mathcal{A} = \{ A_{1}, \ldots, A_{n} \} $ from which to choose precisely $ k $ to be in the portfolio, define a binary variable $ x \in \{0,1\}^{n} $ that determines a selection: $ A_{i} $ will be in the portfolio if $ x_{i} = 1$.

    \cref{fig:multiqaoa_sp500} shows the approximation ratio and success probability for $p$-depths $\leq 40$ and $n\in\{12,14,16\}$ for this task using QAOA. As the depth and thus number of tunable parameters increases, so does performance, which is to be expected~\cite{larocca2023theory,zhang2022quantum}. We see that warm starting improves the discovery of high-quality approximate solutions as well as the optimal solution, with the gap widening dramatically for larger number of assets. Note that the approximation ratio tends to $1.0$ faster than the success probability, because QAOA more readily finds high-quality approximate solutions.
    
    \cref{fig:sp500multistep} shows the approximation ratio and success probability during training of QAOA with $p=10$ with and without warm starting. Warm starting starts with a significantly better approximation ratio in (a) and smoothly trends towards higher quality solutions, including the optimal solution as seen in (b).

    \begin{figure*}[!t]
        \centering
        \begin{tabular}{cc} 
        \includegraphics[width=0.48\textwidth]{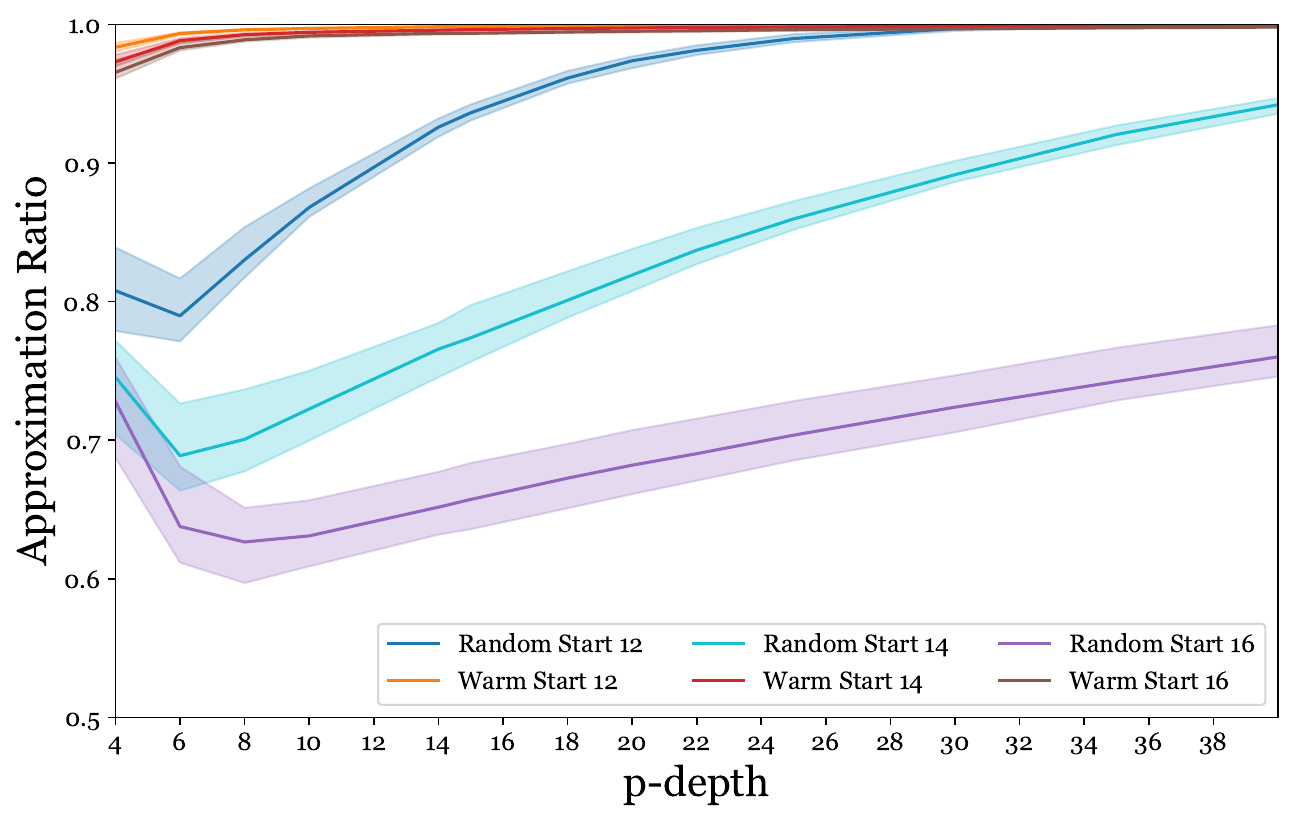} & \includegraphics[width=0.48\textwidth]{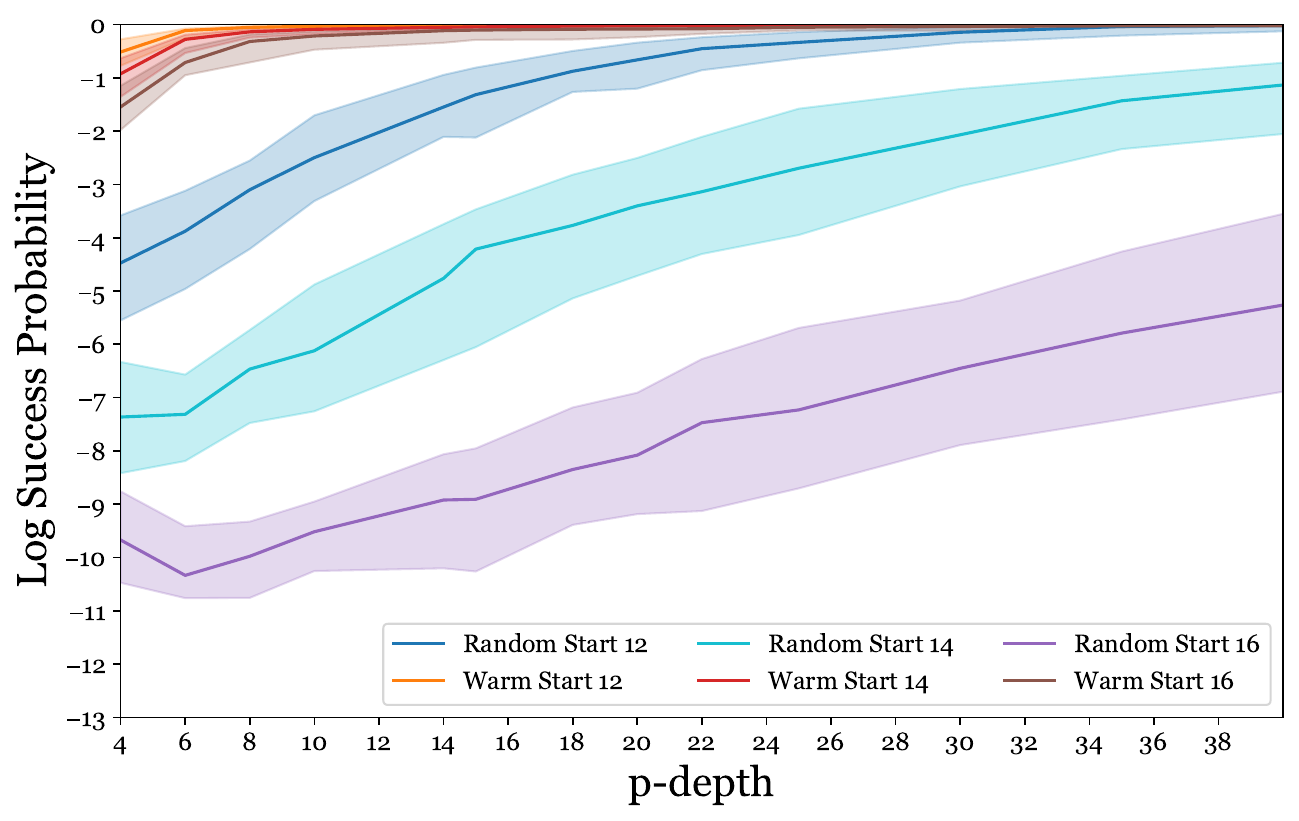} 
        \end{tabular}
        \caption{\textbf{Warm Starting QAOA for Portfolio Optimization on the SP500.} Median performance (solid lines) of the best of $\numstarts$ randomly initialized QAOA-$p$ circuits for different depth $p$ for portfolio optimization across $\numinsts$ monthly SP500 instances. Ribbons delineate the lower and upper quartile.}
        \label{fig:multiqaoa_sp500}
    \end{figure*}
    \begin{figure*}[!t]
        \centering
        \begin{tabular}{cc}
        \includegraphics[width=0.48\textwidth]{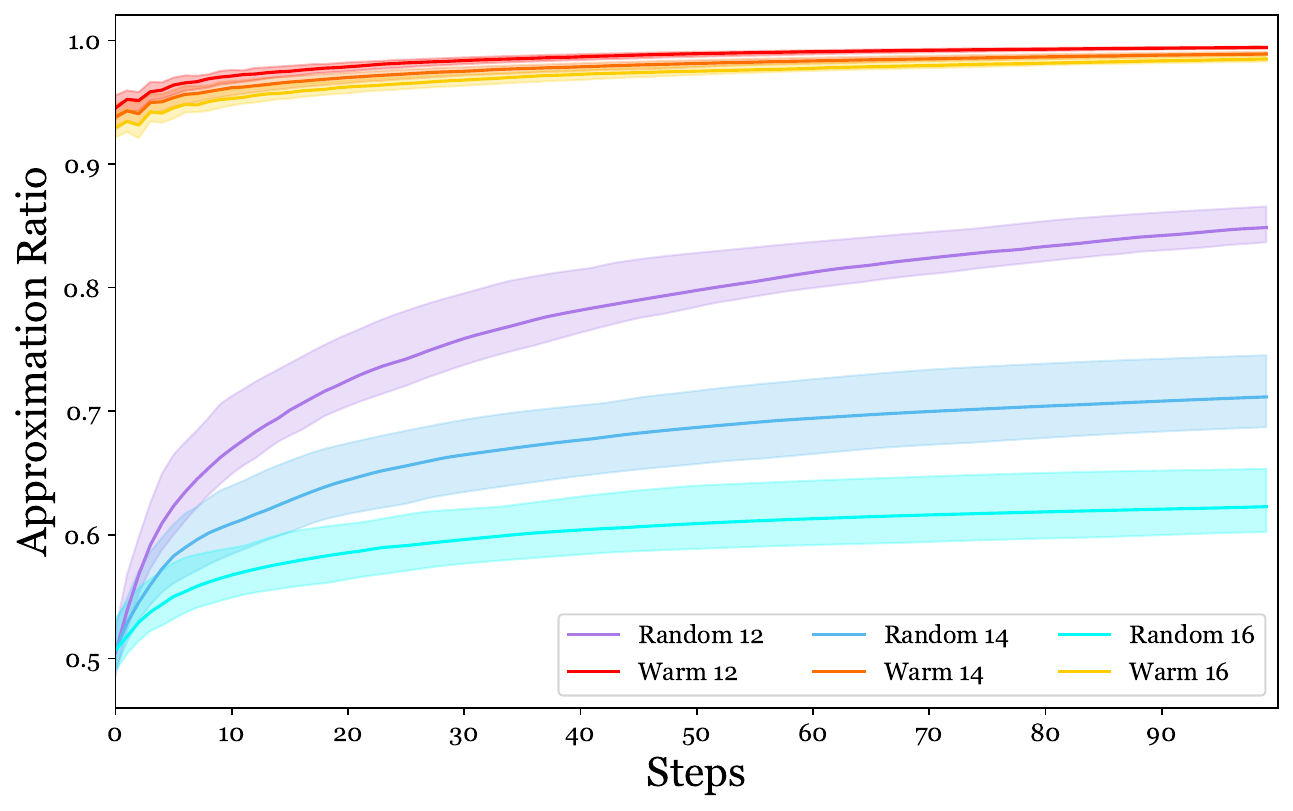} & \includegraphics[width=0.48\textwidth]{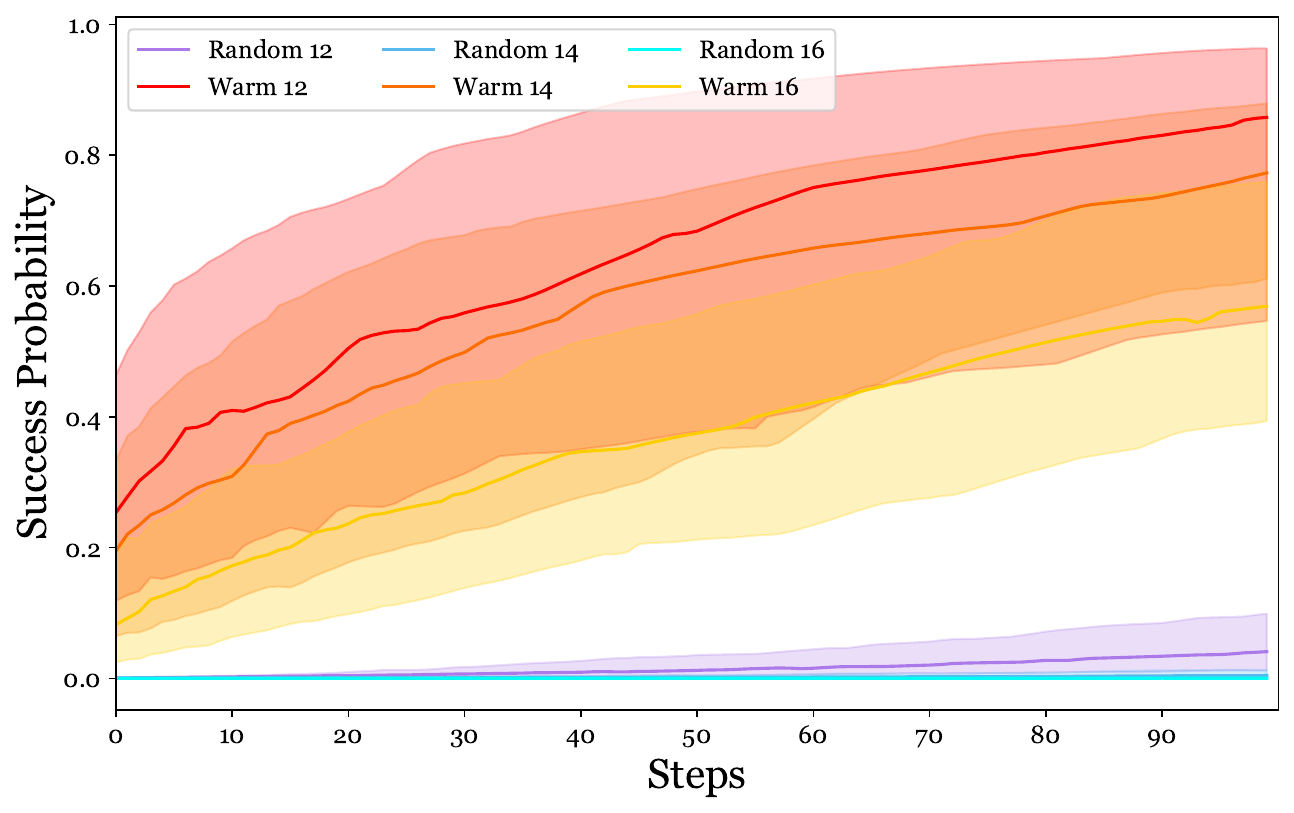}
        \end{tabular}
        \caption{\textbf{QAOA Performance during Training on SP500.} Median performance (solid lines) of the best of $\numstarts$ randomly initialized QAOA-$10$ circuits for each step in the training over $\numinsts$ instances. Ribbons delineate the lower and upper quartile.}\label{fig:sp500multistep}
    \end{figure*}

    \cref{fig:multiqaoa_graph} shows the performance of warm-started QAOA compared to random starts for two additional combinatorial problems. \texttt{Graph Partitioning} is the task of partitioning nodes in a graph into two equal sides such that the cut value across them is minimized. \texttt{Sparsest k-Subgraph} asks to identify the set of vertices that have the smallest number of edges in the induced subgraph. Across each problem domain, we see a dramatic improvement in both metrics by warm starting QAOA.

    Our results reflect empirical findings over meaningful instance families, but it may be possible to construct instances where the restricted ansatz manifold is not as significantly better than a random initialization point. The improvement from warm starting may also be less significant for much larger instance sizes. Our results show that tailored ans{\"a}tze for constrained optimization problems provide a fertile ground for high quality warm starting QAOA.
    
    \begin{figure*}[!t]    
        \centering
        \begin{tabular}{cc}
        \includegraphics[width=0.48\textwidth]{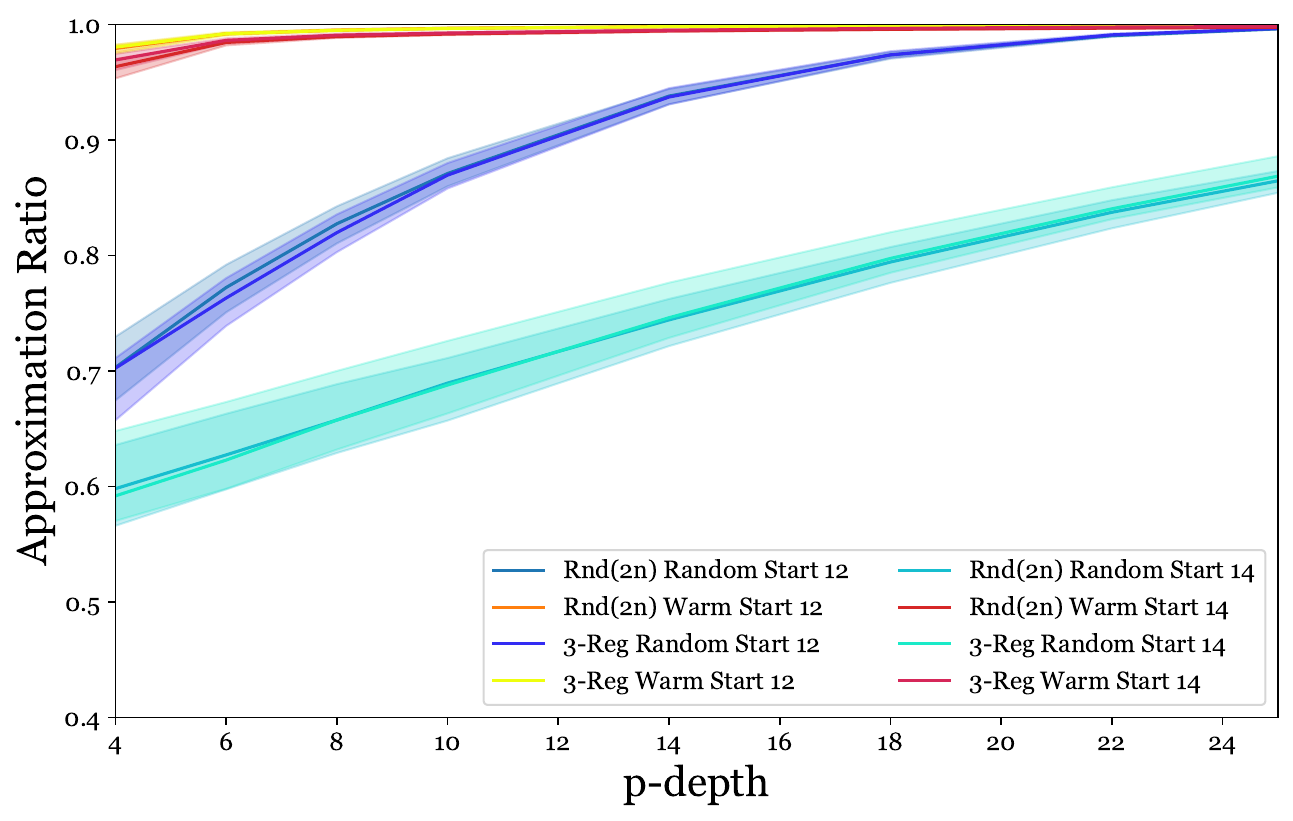} &
        \includegraphics[width=0.48\textwidth]{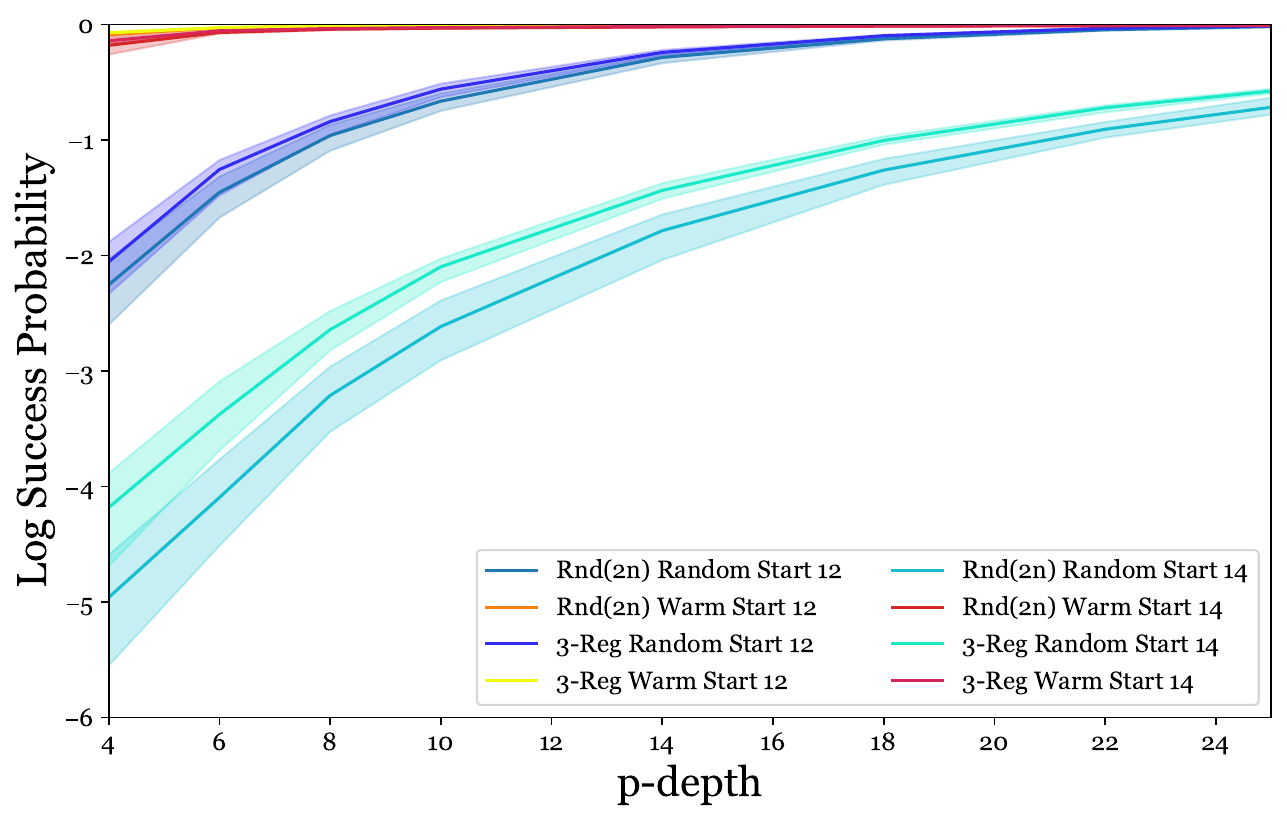} \\ 
        \multicolumn{2}{c}{(a) Graph Partitioning} \\ 
        \includegraphics[width=0.48\textwidth]{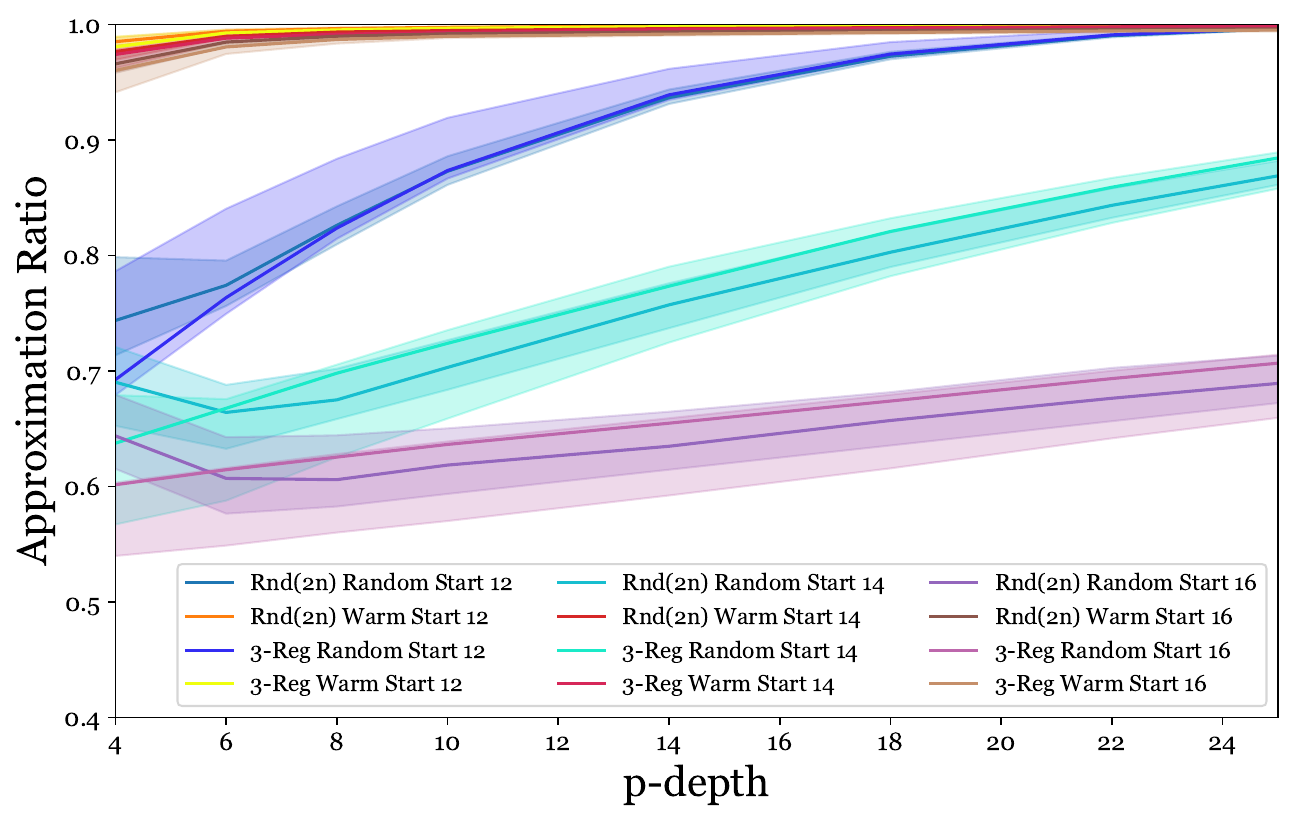} &
        \includegraphics[width=0.48\textwidth]{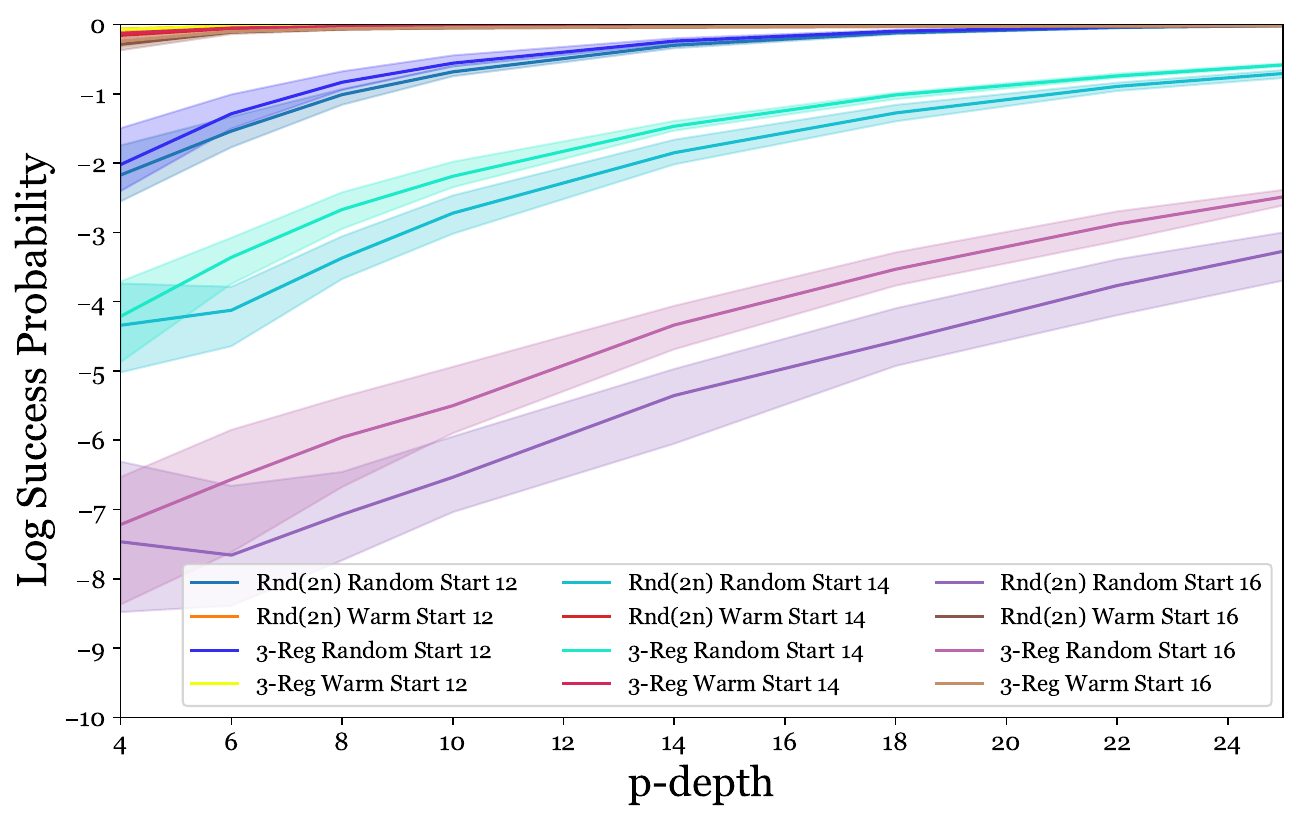} \\ 
        \multicolumn{2}{c}{(b) Sparsest k-Subgraph}
        \end{tabular}
        \caption{\textbf{Warm Starting QAOA for Graph Partitioning and Sparsest k-Subgraph.} Median performance (solid lines) of the best of $\numstarts$ randomly initialized QAOA-$p$ circuits for different depth $p$, across graph partition and random sparsest subgraph. For graphs labeled $\text{Rnd}(2n) $ we considered random graphs with $2\,n$ edges and $n$ vertices and for graphs labeled $\text{Reg-}3$ we considered 3-regular random graphs. Ribbons delineate the lower and upper quartile.}
        \label{fig:multiqaoa_graph}
    \end{figure*}

\section*{Discussion}\label{sec:conclusion}

    Parameterized XY operators (XY-mixers) are a mainstay utility for quantum algorithms, including VQAs. The topology of XY-mixers utilized for a VQA has a significant impact on the associated dimension of the Dynamic Lie algebra, which is known to impact trainability~\cite{larocca2022diagnosing,holmes_connecting_2022,fontana2024characterizing,ragone2024lie,larocca2024review} and overparameterization~\cite{larocca2023theory} for such approaches. Open and closed 1-dimensional XY-mixer topologies have dimensions that scale as quadratically in $ n $ as shown in \cref{thm:poly_dlas_decomp} while the all-to-all connectivity scales $ \Omega(3^{n}) $, as shown in \cref{thm:dlasummmary-big}. We conjecture the true scaling is $ 4^{n - \Theta\left(\log(n)\right)}$, summarized in \cref{conj:expliealg}. Inclusion of single qubit Pauli-Z rotations $ R_{Z} $ leads to the same asymptotic dimensions as shown in \cref{thm:dlasummmary-big} for the 1-dimensional cases specifically and scales at least as exponential for the clique. Moreover, the inclusion of the clique of two qubit Pauli-Z rotations $ R_{ZZ} $ leads to an exponential DLA for all topologies. 
    
    Given the quadratic DLA spaces associated with the XY-mixer cycle with single qubit Pauli-Z rotations, we considered warm starting the associated circuits for the shared angle and multiple angle QAOA formalisms for three NP-Hard constrained optimization tasks for which the XY-mixer is utilized: \texttt{Portfolio Optimization}, \texttt{Sparsest k-Subgraph} and \texttt{Graph Partitioning}. For each problem, we find significantly higher quality parameters leading to dramatic performance benefits with respect to success probability and approximation ratio, shown in \cref{fig:multiqaoa_sp500,fig:multiqaoa_graph} for MA-QAOA and \cref{app:warmstart_details} for SA-QAOA. MA-QAOA generally outperformed SA-QAOA with and without warm starting, likely due to more granular control inside the DLA for fixed depth.
    
    This work opens avenues for several promising directions. Many constrained optimization problems have tailored mixer and phase-separation operators; this rich interplay allows for the consideration of restriction criteria that can aid in designing warm starting. The XY-mixer is also utilized for fermionic simulation on qubit devices~\cite{stanisic2022observing} and such VQAs are also amendable to similar warm starting strategies. Furthermore, there is work being done to better understand constrained Hamiltonian problems that are stoq-MA- and QMA-hard \cite{parekh2024constrained,rayudu2024fermionic}. Understanding when these problems give rise to similar DLA behavior would be helpful for using variational algorithms to solve them.
    
    There is also exciting directions for future work to understand the behavior of XY-mixer PQCs better. While this work presents concrete conjectures about the decompositions of the exponentially-large DLAs, it would be fruitful to formally prove these claims. Moreover, understanding the $\mfg$-purities of the Dicke state and objective Hamiltonians could provide more concrete understanding of the barren plateau behavior when training these circuits. 

\section*{Methods}\label{sec:methods}

    \subsection*{Constructing the Dynamic Lie Algebra}

        An orthonormal basis for a Lie algebra $\mfg$ can be constructed by a simple iterative algorithm~\cite{larocca2022diagnosing}. Let 
        \begin{equation}
        \mathcal{B}_0 = \text{Gram-Schmidt}(\{ i G_{1}, \ldots, i G_{|\mcg| } \})
        \end{equation} 
        for $ G_{i} \in \mcg $ be an orthonormal basis of $ i \mcg $. Define 
        \begin{equation} 
        \Delta_{k} = \bigB{ \com{ i B_{kj}, i G_{l} } :  i B_{kj} \in \mathcal{B}_{k}, G_{l} \in \mcg }
        \end{equation} 
        to construct $ \mathcal{B}_{k+1} = \text{Gram-Schmidt}(\mathcal{B}_{k} \cup \Delta_{k}) $ by finding an orthonormal basis for the new linear span. Follow this construction until $ \mathcal{B}_{K+1} = \mathcal{B}_{K} $ for some $ K $. It follows that the span of the orthonormal basis is equivalent to the Lie algebra and hence the cardinality of the set equal to the dimension:
        \begin{align}
        \mathcal{B} &= \{ i B_{1}, \ldots, i B_{\dim(\mfg)} \}, \\ 
        \mfg &= \text{span}_{\R}\left(\mathcal{B}\right). 
        \end{align}

    \subsection*{Proof Sketch of \cref{thm:poly_dlas_decomp}}
        \noindent An isomorphism between two Lie algebras is a vector space isomorphism that preserves commutation relations. In order to prove Lie algebra equivalences, we give an explicit mapping of $\mathcal{B}$ to known bases of $\so{n}$ or $\su{n}$~\cite{bertlmann2008bloch,Bossion2022_sun_structure_constants}. See \cref{subsec:sosu} for a definition of these Lie algebras and their commutation relationships.

    \subsection*{Example: $\gxyp{n} \cong \so{n}$}    
        $\gxyp{n}$ is the DLA generated by $XY_{j,j+1}$ terms defined on a path graph. New terms in the iterative DLA process are generated by chains of nesting commutations of neighboring terms along a subpath. These terms are defined as $ P_{j,k} := \com{iXY_{j,j+1}, \com{iXY_{j+1,j+2}, \dots, \com{iXY_{k-2,k-1}, iXY_{k-1,k}}}} $ for $1 \leq j < k \leq n$, with $P_{j,j+1} = XY_{j,j+1}$ as the base case. The $\{P_{j,k}\}$ are the only possible terms that can be generated in this process and so $\gxyp{n} = \text{span}\left( \{ P_{j,k} \}_{j<k} \right)$. Moreover, they satisfy skew-symmetric relationships. As an example, for $1 \leq j < k < \ell \leq n$, they satisfy $\com{P_{j,k},P_{k,\ell}} = P_{j,\ell}$. Therefore $\gxyp{n} \cong \so{n}$.

    \subsection*{Example: $\gxyc{n} \cong \su{n}$ for odd $n$}
        $\gxyc{n}$ is the DLA generated by $XY_{j,j+1}$ terms defined on a cycle graph. Since $\Gxyp{n} \subset \Gxyc{n}$, we have that $\gxyp{n} \subset \gxyc{n}$ and so all of the $P_{j,k}$ defined above are in $\gxyc{n}$. Indeed, $\so{n}$ is a Lie subalgebra of $\su{n}$. However, $\Gxyc{n}$ also contains the term $XY_{1,n}$ that completes the cycle. By taking nested commutations of neighboring XY terms starting and ending at the same node $j$, a diagonal element $\Dopp{j,j+1}$ is generated. These $\{D_{j,j+1}\}$ form the Cartan subalgebra of $\su{n}$, the largest set of mutually commuting elements in the algebra. When the indices overlap, commutations of $P_{j,k}$ and $\Dopp{a,b}$ form the $\Popp{c,d}$ terms. Alternatively, these terms can also be formed by taking nested commutations around the cycle in the opposite direction of the $P_{c,d}$ terms. The $P_{j,k}, \Popp{c,d},$ and $\Dopp{a,b}$ satisfy the skew-symmetric, symmetric, and diagonal commutation relations of $\su{n}$. Thus, the linear span of these operators is isomorphic to $\su{n}$ as a Lie algebra.
        
    \subsection*{Proof Sketch of \cref{thm:dlasummmary-big}}
    We begin with the upper bound on the dimension of the exponential-sized DLAs. 
    
    Given a Lie algebra $\mfg$, the \textit{adjoint action} of $h \in \mfg$ is an endomorphism on $\mfg$ given by $\text{ad}_h(g) = [h,g]$. Then $\text{ker}(\text{ad}_h)$ are the elements in $\mfg$ that commute with $h$. Consider the operator
    \begin{equation}
        Z^+ = \sum_{j =1}^n Z_j,
    \end{equation}
    
    \noindent which is associated with preserving Hamming weight of computational basis states. Of particular interest are the \textit{centralizers} of $Z^+$, the Lie sub-algebras of operators that commute with $Z^+$:   
    \begin{align}
        \hlincon &:= \bigB{iH \in \u{2^n} : \com{iH,Z^+} = 0} \\
        \glincon &:= \bigB{iH \in \su{2^n} : \com{iH,Z^+} = 0} \label{eq:glincon}
    \end{align}

    We prove the following decompositions.
        
    \begin{theorem}\label{thm:lincon_decomp}
        \begin{align}
            \hlincon &\cong \u{1}^{\oplus n} \oplus \bigoplus_{k=1}^{n-1} \su{\binom{n}{k}}  \label{eq:cong_hlincon} \\ 
            \glincon &\cong \u{1}^{\oplus n-1} \oplus  \bigoplus_{k=1}^{n-1} \su{\binom{n}{k}}  \label{eq:cong_glincon}
        \end{align}

        \noindent Consequently, $\text{dim}\left(\hlincon\right) = \binom{2\,n}{n} - 1$ and $\text{dim}\left(\glincon\right) = \binom{2\,n}{n} - 2$ which are both $\Theta \left(\frac{4^n}{\sqrt{n}}\right) = 4^{n - \Theta(\log n)} $.
    \end{theorem}

    The proof is as follows. Each eigenspace $F_{k}$ of $Z^{+}$ with eigenvalue $k$ is spanned by computational basis states with Hamming weight $k$. The dimension of this space is $\binom{n}{k}$. The Lie algebra $\hlincon$ can be broken into $\oplus_{k=0}^{n} \hlincon^{k} $ where each $\hlincon^{k}$ is over $F_{k}$ because each Lie algebra is mutually commutative. By dimensionality, we have $\hlincon^k = \u{1} \oplus \su{\binom{n}{k}} $, which proves \cref{eq:cong_hlincon}. \cref{eq:cong_glincon} follows from the definition of $\glincon$ and identities of projection matrices. This establishes \cref{thm:lincon_decomp}. The dimensions follow from the fact that $\text{dim}(\su{a}) = a^2 - 1$.
     
    Since the XY mixer and any Pauli-Z string commutes with $Z^{+}$, the associated Lie algebra of these operators are subalgebras of $\glincon$. Then we know their dimension is upper bound by $\text{dim}\left(\glincon\right)$. 

    We now consider the lower bound. We find linearly independent operators that must be in the DLA of the XY clique $\gxyk{}$, thereby lower bounding the dimension. 

    For the same reason the cycle has $P_{jk} $, there exists a term $P_{\sigma_{jk}} $ for every sequence $ \sigma_{jk} = (j,v_{2},\ldots,v_{N},k)$ with $j < v_{2} <\cdots<v_{N}<k$ in the clique such that it contains a $Z_{v_{j}}$ for every position $v_{j}$ and either $XY_{j,k}$ or $YX_{j,k} := \frac{1}{2}(Y_j X_k - X_j Y_k)$. Moreover, consider a location $w_{1}$ such that $w_{1}$ is not in the sequence. Then by commutating $P_{\sigma_{jk}}$ with $XY_{v_{1},w_{1}}$, now have the operator $YX_{v_{1},w_{1}}$ instead of $Z_{v_{1}}$. This can similarly be done in any location except $j$ and $k$. By counting the number of such constructions, we find that there are at least $(3^{n}-1)/2 = \Omega(3^{n})$ linearly independent terms in $\gxyk{}$. By containment, it immediately follows for $\gxykz{}$. 
    
    To show these terms are contained in $\gxyczzz{}$, we show that $XY_{j,k}$ for any $j,k$ is contained inside $\gxyczzz{}$. Begin with $P_{j,k}$ and commute with $Z_{k-1}Z_{k}$ to remove $Z_{k-1}$ from the term. This process can be repeated to remove any $Z_{j}$ in $P_{j,k}$ with $Z_{j} Z_{k}$. Each application sends the $X_{j,k}$ in $P_{j,k}$ to $YX_{j,k}$, but $Z_{k}$ can be applied to flip it back to $XY_{j,k}$.  

    \begin{figure*}[!t]
    \centering \includegraphics[width=0.75\textwidth]{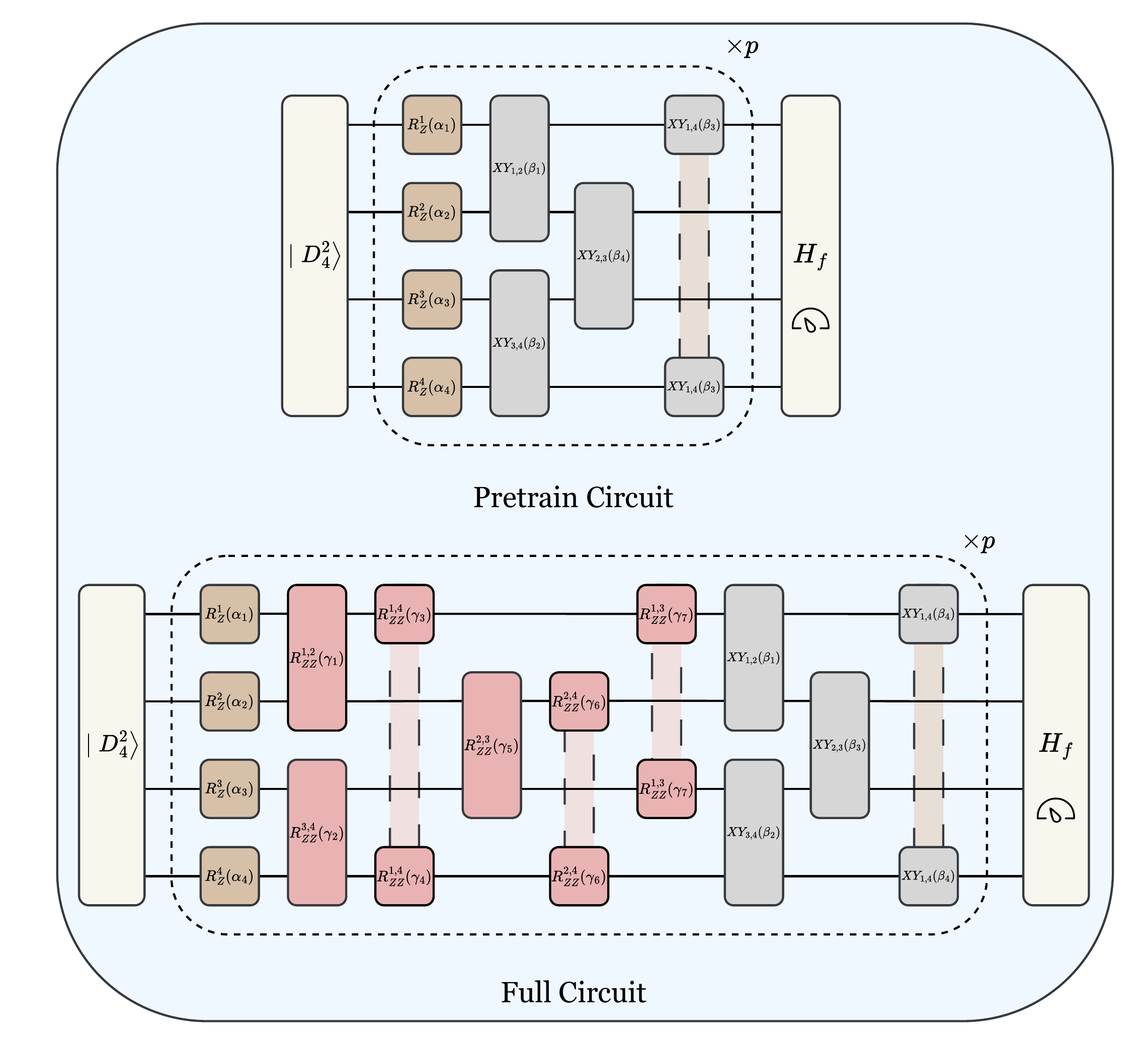}
    \caption{\textbf{Warm starting QAOA.} The figure on the top shows the circuit utilized to warm start QAOA without the $R_{ZZ}$ terms. The figure on the bottom shows the full circuit for QAOA. These corresponds to a polynomial and exponential sized DLA, respectively.}\label{fig:circ_maqaoa}
    \end{figure*}

    \subsection*{Warm Starting Procedure for QAOA}\label{subsec:qaoawxy} 

    \begin{figure*}
        \begin{mybluebox}[label={not_used}]{Warm Start Procedure}
            \textbf{Input}: Gate sets $\mcg$ and $\mcg' \subset \mcg$.
            
            \textbf{Output:} $\bmt_f \in [0,2\pi]^{\abs{\mcg}}$.
            \begin{enumerate}
                \item \textbf{Restrict:} Initialize tuple vector $\bmt_i^{WS} \in [0,2\pi]^{\abs{\mcg'}} $. 
                \item \textbf{Pretrain:} Train the circuit by varying the $\abs{\mcg'}$ angles to learn $\bmt_f^{WS}$.
                \item \textbf{Transfer:} Initialize tuple vector $\bmt_i = (\bmt_f^{WS}, \vec{0}) $ with length $\abs{\mcg}$.
                
                \item \textbf{Refine:} Run an efficient training schedule on all gates in $\mcg$ to achieve $\bmt_f $. 
                
                Return $\bmt_f$.
            \end{enumerate}
        \end{mybluebox}
        \caption{\textbf{Warm Start Procedure.} The warm-starting procedure used throughout this paper. The DLA $\mfg = \lie{\mcg}$ is assumed to be exponentially-large whereas $\mfg' = \lie{\mcg'}$ should be efficiently trainable.}
        \label{fig:warmstart_alg}
    \end{figure*}
    
    The warm starting circuit consists of $ \mathcal{G}_{XY}^{C} $ as the mixing operators and $ \mathcal{G}_{Z} $ as the phase-separating operators:
    \begin{align}\label{eq:maqaoa-ws}
        U^{WS}(\bm{\alpha}^{WS},\, &\bm{\beta}^{WS}) =\nonumber\\
        & \prod_{j=1}^{p} \prod_{M_{k} \in \mathcal{G}_{XY}^{C}} e^{i \beta_{jk}^{WS} M_{k}} \prod_{N_{k} \in \mathcal{G}_{Z}} e^{i \alpha_{jk}^{WS} N_{k}}. 
    \end{align}
    
    \noindent $ U^{WS} $ has $n$ $R_{Z}$ rotation gates and $n$ XY-mixer gates per layer leading to $ \bmt^{WS} \in \R^{2 \, n \, p} $ where $\bmt^{WS} = (\bm{\alpha}^{WS}, \bm{\beta}^{WS} )$. In comparison the full QAOA circuit with $ \mathcal{N} = \mathcal{G}_{Z,ZZ} $ is given by
    \begin{align}
        U(\bm{\alpha}, \bm{\beta}) = \prod_{j=1}^{p} \prod_{M_{k} \in \mathcal{G}_{XY}^{C}} e^{i \beta_{jk} M_{k}} \prod_{N_{k} \in \mathcal{G}_{Z,ZZ}} e^{i \alpha_{jk} N_{k}}. 
    \end{align}
    
    \noindent $ U $ is composed of $ n $ $R_{Z}$ gates, $ n $ XY-mixer gates, and $\binom{n}{2}$ $R_{ZZ}$ gates, and so $ \bmt = (\bm{\alpha}, \bm{\beta}) \in \R^{(n^2+3\,n) \, p/2} $. See \cref{fig:circ_maqaoa} for a depiction of these two circuits.
    
    During both the warm start and training with the full circuit, we use the Dicke state $ \ket{D^n_{n/2}} $ as the initial state. The observable is the optimization task cost Hamiltonian $ H_{f} $ in both scenarios. For all three combinatorial problems considered, $ H_{f} \in \text{span}\left( \mathcal{G}_{Z,ZZ} \right) $. Then the respective loss functions are 
    \begin{align}
        \mathcal{L}^{WS}(&\bmt^{WS}) = \nonumber\\ 
        &\bra{D^n_{n/2}} U^{WS}(\bmt^{WS})^{\dagger} \, H_{f} \, U^{WS}(\bmt^{WS}) \ket{D^n_{n/2}}, \nonumber\\
        \mathcal{L}(\bmt) &= \bra{D^n_{n/2}} U(\bmt)^{\dagger} \, H_{f} \, U(\bmt) \ket{D^n_{n/2}}. \nonumber 
    \end{align}

    To warm start the full circuit, we train the restricted ansatz circuit $ U^{WS} $ with a random initial set of parameters $ \bmt^{WS}_i \in [0,2\pi]^{2\,n\,p} $. After training for $\numsteps$ steps minimizing loss $\mathcal{L}^{WS}$, we learn the angles $ \bmt^{WS}_f $. Then we parameter transfer to the full circuit: $R_{Z}$ and XY-mixer gates are set by $ \bmt^{WS}_f $ and the $R_{ZZ}$ gates are set to zero. We then train the full circuit for $\numsteps$ steps with respect to the loss $ \mathcal{L} $. For comparison, we consider random initialization for the full circuit  $ \bmt \in [0,2\pi]^{(n^2+3\,n)\,p/2} $. For all experiments we used the Adam optimizer~\cite{kingma2014adam} in Pennylane framework~\cite{bergholm2018pennylane}. An overview is given in \cref{fig:warmstart_alg}.

    \subsection*{Example: Portfolio Optimization}\label{subsec:method_port_opt}

    Let $ \bm{p} \in \R^{n} $ represent the expected returns (e.g. based on historical data) such that $ \bm{p}_{i} = \mathbb{E}[ p(A_{i}) ] $ for asset $ A_{i} $ and a (normalized) returns function $ p $. Then $ \bm{p}^T x $ is the total expected return of a specific portfolio selection. Let $\bm{C}$ represented the associated covariance matrix for the returns such that $ \bm{C}_{ij} = \mathbb{E}[ (p(A_{i}) - \bm{p}_{i}) (p(A_{j}) - \bm{p}_{j}) ] $. For example, $ \bm{p} $ and $ \bm{C} $ may be constructed through sampling from an underlying population of historical asset pricing data (such as a stock index). 
    
    Given $ \bm{p} $ and $ \bm{C} $, we wish to solve the binary quadratic program with linear constraint (cast as a minimization):
    \begin{align}\label{eq:po_iqp}
        (IQP)\quad &\min_{x \in \{0,1\}^n}  - \bm{p}^{T} x + q \, x^{T} \bm{C}  x \nonumber \\ 
        &\text{s.t.} \quad \abs{x} = k
    \end{align}
    \noindent  Only allowing for $ |x| = k $ constrains the problem such that $ k $ assets are selected for the portfolio. Here, $q \in \R^+ $ is a risk-aversion parameter such that larger $q$ corresponding to higher (lower) risk aversion (appetite). The inner product with $ \bm{p} $ rewards selecting higher return assets while $ x^{T}  \bm{C} x $ is the associated risk profile of the selection. The parameter $q$ trades off the possible return of a selection with the associated risk of the selection. Since there are $\Theta(n^k)$ many binary strings $ x $ for a specific $k$, if $ k = \mathcal{O}(1) $, the problem is solvable in polynomial time. However, if $k = \Omega(n)$, such as $k = n/2$, the space we optimize over is of size  $\Theta(2^n/\sqrt{n}) = 2^{n-\Theta(\log(n))}$ and the problem becomes \cc{NP}-hard for arbitrary $\bm{p}, \bm{C}$~\cite{cai2008parameterized}.

    Let $\sigma^0 = \ketbra{0}$ and $\sigma^1 = \ketbra{1}$ be the projections onto the standard 1-qubit basis vectors and $ \sigma_{k}^{0}$  and $ \sigma_{k}^{1} $ be those operators on a specific qubit $k$. We can embed \cref{eq:po_iqp} as a cost Hamiltonian in the standard way. Define the return and covariance Hamiltonians as
    \begin{align}
        H_p = -\sum_{i = 1}^n p_i \, \sigma^1_i, \; H_{covar} = \sum_{i,j = 1}^n C_{ij} \, \sigma^1_i \sigma^1_j,
    \end{align}
    \noindent and let $H_{f} = H_p + q \, H_{covar}$.
    \begin{lemma}\label{lem:hamportop}
        $H_{f}$ encodes the IQP cost function \cref{eq:po_iqp}, for $x \in \{0,1\}^n$,  $\mel{x}{H}{x} = - \bm{p}^{T} x + q \, x^{T}  \bm{C}x$.
    \end{lemma}
    
    \noindent \cref{app:isingembed} discusses Ising embedding for optimization and provides the explicit form of $ H_{f} $ over $ \mathcal{G}_{Z,ZZ} $.
    
    Standard and Poor's 500 is stock market index that tracks the performance of the 500 largest companies, founded in 1957. Given a price $ q_{a}^{i} $ for an asset $ a \in \mathcal{A} $ on a day index $ i $, the return is $ p_{a}^{i} = (q_{a}^{i} - q_{a}^{i-1})/q_{a}^{i-1} $. A daily return is defined for everyday a stock is traded in a monthly interval $ d_{m} = \{ 1, \ldots, |d_{m}| \} $ (excludes weekends and holidays). Then each day serves as a sample for that month, such that the average monthly return is $ p_{a}^{m} = \sum_{i} p_{a}^{i} / (|d_{m}| - 1) $. Then the average covariance between assets $ a $ and $ b $ over the month is given by $ C_{a,b} = \sum_{i} (r_{a}^{i} - p_{a}^{i})(r_{b}^{i} - p_{b}^{i}) / (|d_{m}|- 1) $. 
    
    Given $\numinsts$ months of publicly available trading data of the SP500 between 2010 and 2018, we select the $ n $ top performing stocks based on average return per month to define the returns vector $ \bm{p}^{m} = (p_{1}^{m}, \ldots, p_{n}^{m}) $ and associated covariance matrix $ \bm{C}^{m} $. Then the task is to find the best $ k=n/2 $ assets for our portfolio. 

    \subsection*{Example: Graph Partitioning}\label{subsec:method_graphpart}

    \texttt{Graph Partitioning} is a famous NP-Hard optimization problem that is well studied for Quantum Annealing and QAOA~\cite{lucas_ising_2014,hen_driver_2016}. Given a graph $ G = (V,E) $, we wish to find a partition, $ V_{1}, V_{2} $, such that the partition is proper ($ V_{1} \cap V_{2} = \{ \} $), equal sized ($ |V_1| = |V_2| $), and minimizes the number of edges in the crossing, $ E(V_1, V_2) = \{ (v_1, v_2) | v_1 \in V_1, v_2 \in V_2, (v_1, v_2) \in E \} $ over all possible partitions. 
    
    Let $ n = |V| $ and $ A $ the adjacency matrix of $ G $, then \texttt{Graph Partitioning} can be written as a quadratic program with a linear constraint of the form:
    \begin{align}
    \min_{x} \; (\mathbf{1} - x^{T}) \mathbf{A} x \text{ s.t. } |x| = \frac{n}{2},
    \end{align}
    where $ \mathbf{1} = (1,\ldots,1) $ such that an edge $(v_{i},v_{j})$ contributes to the count of the cut when $ x_{i} \neq x_{j} $. As shown in \cref{app:isingembed}, $ H_{f} = \frac{|E|}{2} I + \frac{1}{2} \sum_{i<j}^{n} Z_{i} Z_{j} $. Note that $ \text{Proj}_{\Gz{n}}\left( H_{f} \right) = 0 $.

    \subsection*{Example: Sparsest k-Subgraph}\label{subsec:method_sparsest}
       
    \texttt{Sparsest k-Subgraph} is a famous NP-Hard (via \texttt{Independent Set}) optimization problem that does not admit a polynomial time approximation scheme (PTAS)~\cite{khot2006ruling,bhaskara2010detecting}. Given a graph $ G = (V,E) $ and a constant $k$, we wish to find $ S \subseteq V $ such that $ |S| = k $ and the number of edges on the subgraph defined by $ S $ is minimized $ E(S) = \{ (v_1, v_2) : v_1, v_2 \in S, \, (v_1, v_2) \in E \} $~\cite{srivastav1998finding}.

    Let $ n = |V| $ and $ A $ the adjacency matrix of $ G $, then \texttt{Sparsest $\frac{n}{2}$-Subgraph} can be written as a quadratic program with a linear constraint of the form:
    \begin{align}
    \min_{x} \; x^{T} \mathbf{A} x \text{ s.t. } |x| = \frac{n}{2} 
    \end{align}
    
    As shown in \cref{app:isingembed}, $ H_{f} = \sum_{i<j}^{n} A_{ij} \sigma_{i}^{1} \sigma_{j}^{1} $.

\section*{Acknowledgements}

The authors thank Sarvagya Upadhyay, Yasuhiro Endo, and Hirotaka Oshima for their insightful comments. %

\section*{Author Contributions}

SK and HL developed the theoretical aspects of this work. SK and HL conducted the numerical experiments. SK and HL contributed equally to the writing of the manuscript.

\section*{Data Availability}

The data used in numerical experiments are available from the authors upon request.

\section*{Code Availability}

Further implementation details are available from the authors upon request, requests can be made to \texttt{hleipold@fujitsu.com}. 

\section*{Competing Interest}

The authors declare no competing interests.

\renewcommand*{\bibfont}{\fontsize{10}{10}\selectfont}
\printbibliography

\appendix 
\pagebreak 
\clearpage 
\newpage

\setcounter{page}{1}
\onecolumn 

\renewcommand{\thesection}{S.\arabic{section}}
\renewcommand{\thesubsection}{S.\arabic{section}.\arabic{subsection}}
\crefname{section}{Supplementary Section}{Supplementary Sections}
\renewcommand{\theequation}{S\arabic{equation}}
\setcounter{equation}{0}

\begin{center}
\textbf{\LARGE Supplementary Information}
\end{center}
\begin{center}
\textbf{\Large Contents}
\end{center}

\smallskip
\noindent \ref*{app:controlgrad}.~~\hyperref[app:controlgrad]{\bf Control and Training through Gradient Descent} \dotfill\textbf{\pageref*{app:controlgrad}}
\medskip 

\noindent \ref*{app:warmstart_details}.~~\hyperref[app:warmstart_details]{\bf Warm Starting Shared Angle QAOA versus Multi-Angle QAOA} \dotfill\textbf{\pageref*{app:warmstart_details}}
\medskip 

\noindent \qquad \begin{minipage}{\dimexpr\textwidth-0.72cm}
\ref*{app:subsec_saqaoa}.~~\hyperref[app:subsec_saqaoa]{Shared Angle QAOA}\dotfill\text{\pageref*{app:subsec_saqaoa}}

\qquad\ref*{subsec:warmstart_saqaoa}.~~\hyperref[subsec:warmstart_saqaoa]{Warm Starting Shared Angle QAOA}\dotfill\text{\pageref*{subsec:warmstart_saqaoa}}

\ref*{subsec:maqaoa_vs_saqaoa}.~~\hyperref[subsec:maqaoa_vs_saqaoa]{Shared Angle QAOA versus Multi-Angle QAOA on Portfolio Optimization}\dotfill\text{\pageref*{subsec:maqaoa_vs_saqaoa}}
\end{minipage}
\medskip 

\noindent \ref*{sec:appendix_liealgs}.~~\hyperref[sec:appendix_liealgs]{\bf Preliminaries and Lie Algebra}\dotfill\textbf{\pageref*{sec:appendix_liealgs}}
\medskip

\noindent \qquad \begin{minipage}{\dimexpr\textwidth-0.72cm}
\ref*{subsec:notation}.~~\hyperref[subsec:notation]{Notation}\dotfill\text{\pageref*{subsec:notation}}

\qquad\ref*{subsec:defs_paulis_XY}.~~\hyperref[subsec:defs_paulis_XY]{The XY Mixer}\dotfill\text{\pageref*{subsec:defs_paulis_XY}}

\ref*{subsec:lie_algebras}.~~\hyperref[subsec:lie_algebras]{Lie Algebras}\dotfill\text{\pageref*{subsec:lie_algebras}}

\qquad\ref*{subsec:sosu}.~~\hyperref[subsec:sosu]{The Special Orthogonal and Unitary Subalgebras}\dotfill\text{\pageref*{subsec:sosu}}

\end{minipage}
\medskip

\noindent \ref*{sec:appendix_DLAs_description}.~~\hyperref[sec:appendix_DLAs_description]{\bf Lie Algebras of XY-mixer Topologies}\dotfill\textbf{\pageref*{sec:appendix_DLAs_description}}
\medskip

\noindent \qquad \begin{minipage}{\dimexpr\textwidth-0.72cm}
\ref*{subsec:poly_dlas_desc}.~~\hyperref[subsec:poly_dlas_desc]{Polynomially-Sized XY-mixer DLAs}\dotfill\text{\pageref*{subsec:poly_dlas_desc}}

\qquad\ref*{subsubsec:poly_basis_els}.~~\hyperref[subsubsec:poly_basis_els]{Basis Elements}\dotfill\text{\pageref*{subsubsec:poly_basis_els}}

\qquad\ref*{subsubsec:poly_dlas_decomp}.~~\hyperref[subsubsec:poly_dlas_decomp]{Lie Algebra Decompositions}\dotfill\text{\pageref*{subsubsec:poly_dlas_decomp}}

\ref*{subsec:exp_dlas_desc}.~~\hyperref[subsec:exp_dlas_desc]{Exponentially-Sized XY-mixer DLAs}\dotfill\text{\pageref*{subsec:exp_dlas_desc}}

\qquad\ref*{app:glincon}.~~\hyperref[app:glincon]{Lie Algebra of Cardinality Invariance}\dotfill\text{\pageref*{app:glincon}}

\qquad\ref*{subsubsec:exp_basis_els}.~~\hyperref[subsubsec:exp_basis_els]{Basis Elements}\dotfill\text{\pageref*{subsubsec:exp_basis_els}}

\qquad\ref*{subsubsec:exp_basis_decomp}.~~\hyperref[subsubsec:exp_basis_decomp]{Lie Algebra Decompositions}\dotfill\text{\pageref*{subsubsec:exp_basis_decomp}}

\end{minipage}
\medskip

\noindent 
\ref*{sec:appendix_DLAs_calcs}.~~\hyperref[sec:appendix_DLAs_calcs]{\bf Derivation of Lie Algebras}\dotfill\textbf{\pageref*{sec:appendix_DLAs_calcs}}
\medskip

\noindent \qquad \begin{minipage}{\dimexpr\textwidth-0.72cm}
\ref*{subsec:dla_xy_path}.~~\hyperref[subsec:dla_xy_cycle]{Lie Algebra of the $XY$ Path}\dotfill\text{\pageref*{subsec:dla_xy_path}}

\ref*{subsec:dla_xy_cycle}.~~\hyperref[subsec:dla_xy_cycle]{Lie Algebra of the $XY$ Cycle}\dotfill\text{\pageref*{subsec:dla_xy_cycle}}

\qquad\ref*{subsubsec:even_cycle}.~~\hyperref[subsubsec:even_cycle]{Even $n$}\dotfill\text{\pageref*{subsubsec:even_cycle}}

\qquad\ref*{subsubsec:odd_cycle}.~~\hyperref[subsubsec:odd_cycle]{Odd $n$}\dotfill\text{\pageref*{subsubsec:odd_cycle}}

\ref*{subsec:dla_xy_z_path}.~~\hyperref[subsec:dla_xy_z_path]{Lie Algebra of the $XY$ Path with $R_z$}\dotfill\text{\pageref*{subsec:dla_xy_z_path}}

\ref*{subsec:dla_xy_z_cycle}.~~\hyperref[subsec:dla_xy_z_cycle]{Lie Algebra of the $XY$ Cycle with $R_z$}\dotfill\text{\pageref*{subsec:dla_xy_z_cycle}}

\ref*{subsec:exp_dlas}.~~\hyperref[subsec:exp_dlas]{Exponentially Large $XY$-mixer Dynamic Lie Algebra}\dotfill\text{\pageref*{subsec:exp_dlas}}

\qquad\ref*{subsubsec:xyclique}.~~\hyperref[subsubsec:xyclique]{Fully-connected $XY$ Interactions}\dotfill\text{\pageref*{subsubsec:xyclique}}

\qquad\ref*{subsubsec:xyclique_z}.~~\hyperref[subsubsec:xyclique_z]{Fully-connected $XY$ Interactions with $R_{z}$}\dotfill\text{\pageref*{subsubsec:xyclique_z}}

\qquad\ref*{subsubsec:xyring_zz}.~~\hyperref[subsubsec:xyring_zz]{$XY$ Cycle with $R_z$ and Fully-conencted $R_{zz}$ Interactions}\dotfill\text{\pageref*{subsubsec:xyring_zz}}

\end{minipage}
\medskip

\noindent 
\ref*{app:isingembed}.~~\hyperref[app:isingembed]{\bf Ising Embeddings of Linear and Quadratic Costs} \dotfill\textbf{\pageref*{app:isingembed}}
\medskip

\noindent 
\ref*{sec:com_calcs}.~~\hyperref[sec:com_calcs]{\bf Common Commutation Calculations} \dotfill\textbf{\pageref*{sec:com_calcs}}
\medskip

\section{Control and Training through Gradient Descent}\label{app:controlgrad}

Gradient descent optimizes the parameters $ \bmt $ of a classically controlled system through iterative updating the parameters (as a vector or position) through the associated gradient of the loss function. At timestep $t$ and current vector $\bmt^{t}$,  the partial derivative with respect to angle $ \theta_{i} $ is given by $ \partial_{\theta_{i}} \mathcal{L}(\bmt^t) $ and the gradient of the loss function is $ \grad \mathcal{L}(\bmt^t) = ( \partial_{\theta_{1}} \mathcal{L}(\bmt^t), \ldots, \partial_{\theta_{PK}} \mathcal{L}(\bmt^t) ) $. This gradient points in the direction of greatest ascent around the current position $\bmt^{t}$. Then $ \bmt^{t+1} = \bmt^{t} - \alpha \, \grad \mathcal{L}(\bmt^t) $ defines a ``close'' vector by moving in the direction of greatest descent. Efficient but approximate estimation of the gradient underlies most local update policies in quantum and classical computing, while other local updates without the explicit intention tend to be amendable to analysis through gradients~\cite{arrasmith_effect_2021}. 

For Quantum Machine Learning (QML)~\cite{biamonte2017quantum}, direct computation of the gradient on quantum devices may be expensive due to wavefunction collapse associated with each partial derivative (therefore hybrid quantum classical systems lack a direct analog to backpropagation~\cite{abbas2023quantum}), leading to approaches such as Parameter Shift~\cite{mitarai2018quantum} and SPSA~\cite{bonet2023performance} while simulation of quantum systems allow for backpropagation and the self-adjoint method. 

As noted in Ref.~\cite{larocca2023theory}, understanding the structure of \cref{eq:loss} is tightly connected to the DLA and its size. The loss equation can be decomposed into to explicit include intermediate mappings as:
\begin{equation}
    \R^{P \cdot K} \xrightarrow{(1)} \mathcal{U}(2^n) \xrightarrow{(2)} \H \xrightarrow{(3)} \R,
\end{equation}
(1) is the mapping of $ \bmt $ to the associated PQC $ U(\bmt) $, which is an element of $\mathcal{U}(2^{n})$ - the set of all unitaries over $\mathbb{C}^{2^{n}}$. The DLA $\mfg$ has an associated dynamical Lie group $e^\mfg \subset \mathcal{U}(2^n)$ which we interpret as the reachable unitaries this circuit can represent. The larger the circuit, either in width $K$ or depth $P$, the more control one has over the possible unitaries one can reach in $ e^\mfg $. Note that this comes at a possible cost of needing more samples to properly train. After fixing an initial state $ \rho_0 $, map (2) then produces a state in Hilbert space $ \rho_{1} = U(\bmt)\rho_0U(\bmt)^\dagger $. Exploring the whole state space with the goal of minimizing loss becomes possible when $PK$ is on the order of $\text{dim}(\mfg)$. Lastly, (3) is the measurement of observable $ O $ on $ \rho_{1} $ which results in sampling from the associated loss function. The loss landscape can be understood by analyzing the Hessian $ \nabla^2 \mathcal{L}(\bmt) $. In particular, the eigenvalues of the Hessian are associated with the local curvature around $\bmt$ and scale inverse to $\text{dim}(\mfg)$, meaning that the landscape looks more flat for higher dimensional Lie algebras. %
See \cite{larocca2023theory} Section C for more thorough details of these decompositions.

The Lie correspondence is that for any set of parameters $\bmt$, there exist $ \bm{\omega} \in \R^{\dim(\mfg)} $ such that $ U(\bmt) = e^{i \sum_{j=1}^{\dim(\mfg)} \omega_{j} B_{j} } $. The dimension of the Lie algebra, $ \dim(\mfg) $, is equivalent to the effective number of free parameters, $ \bm{\omega} $, for the image of $ \bmt $ given a collection of generators $ \mathcal{G} $ at a sufficient depth for the generators to become a 2-design~\cite{harrow2009random,dankert2009exact,hunter2019unitary}. The variance of a real function around a random point in the image space decays in the dimension of the space~\cite{ragone2024lie}, which leads to a curse of dimensionality~\cite{larocca2024review} for highly expressive circuits (with $\dim(\mfg)$ growing superpolynomial in the number of qubits) characterized by the cost landscape appearing very barren (lacking highly distinguished peaks) referred to as the Barren Plateau~\cite{mcclean2018barren,larocca2024review} due to vanishing gradients. 

\section{Warm Starting Shared Angle versus Multi-Angle QAOA}\label{app:warmstart_details}

    \begin{suppfigure}[!t]
        \centering
        \begin{tabular}{c c}
        \includegraphics[width=0.48\linewidth]{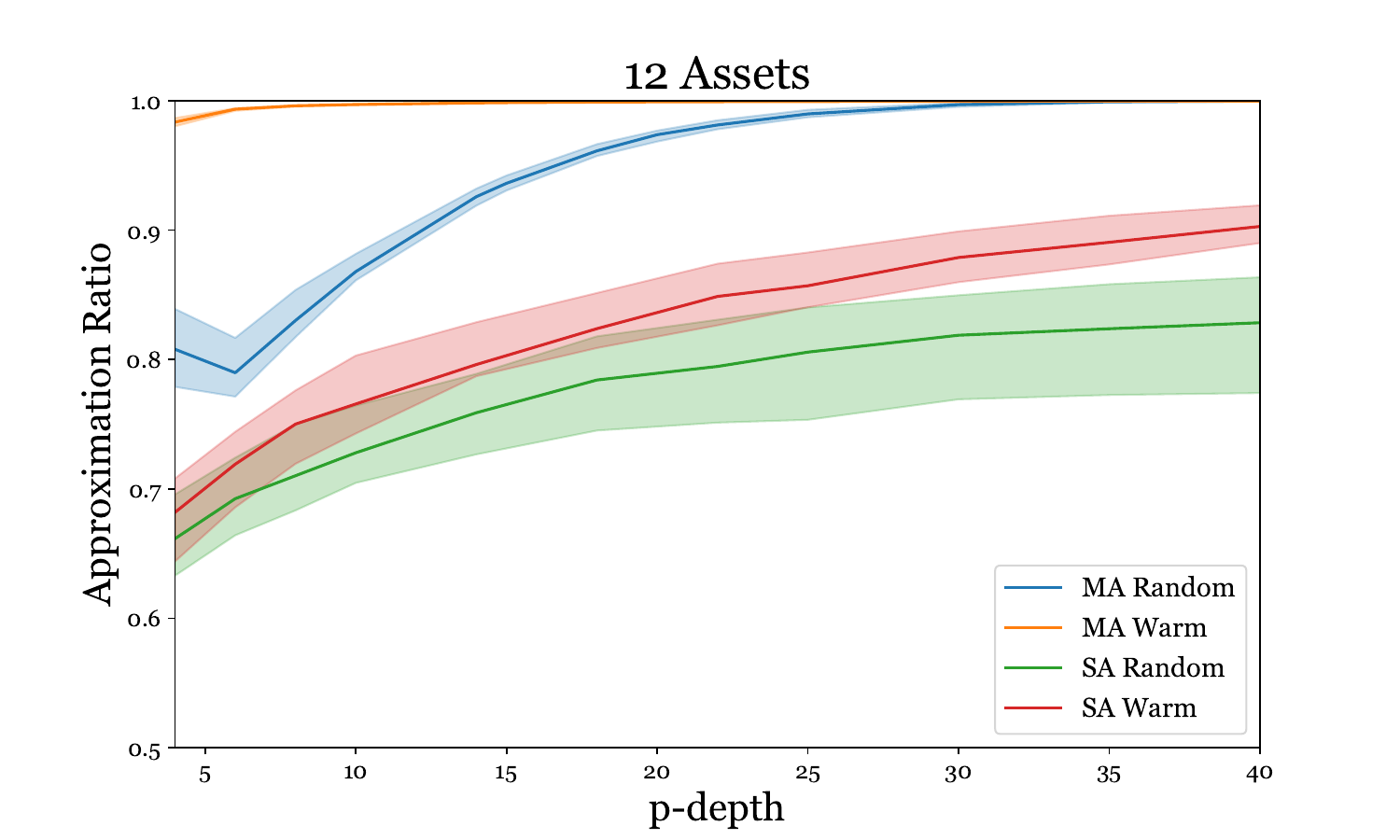} & 
        \includegraphics[width=0.48\linewidth]{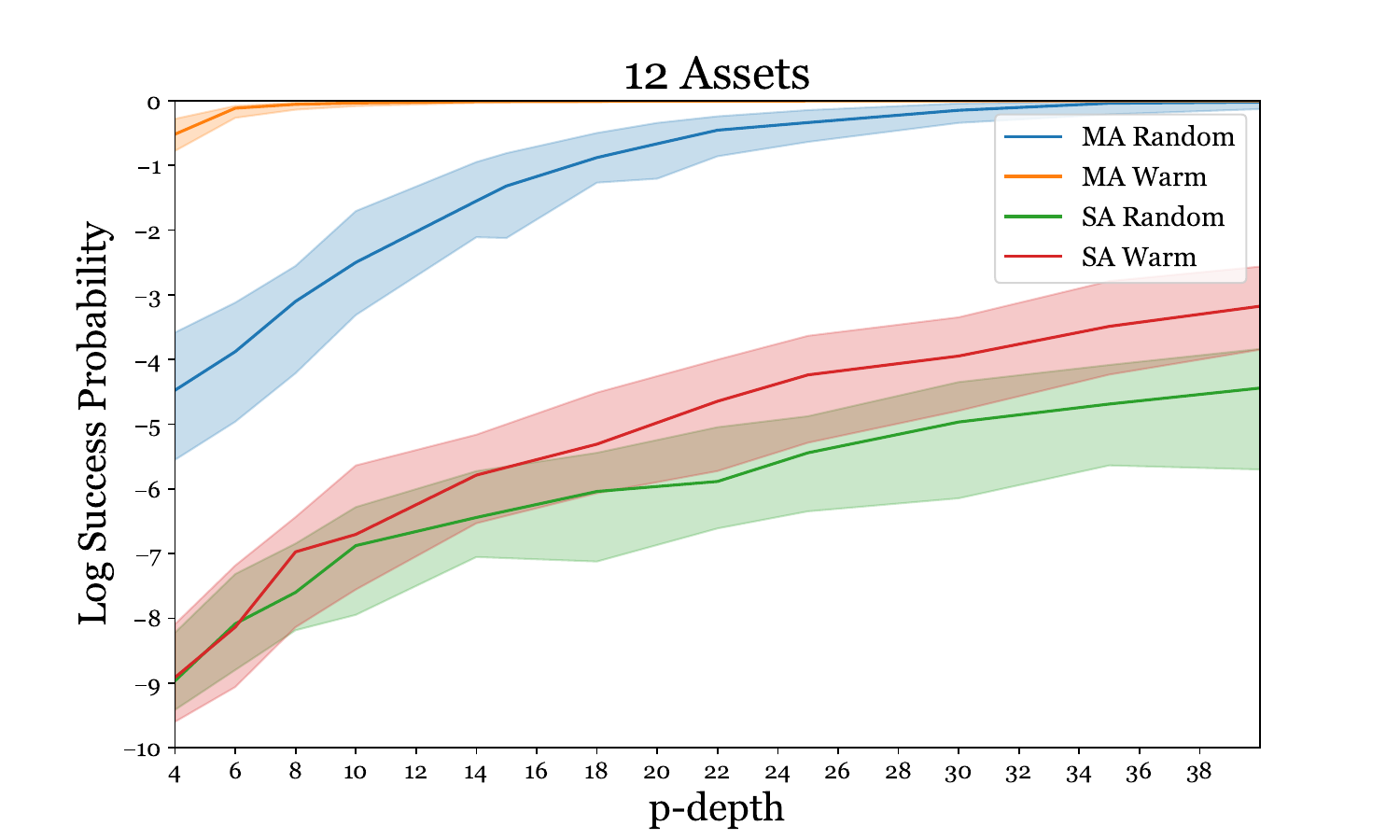} \\ 
        (a) & (b) \\
        \includegraphics[width=0.48\linewidth]{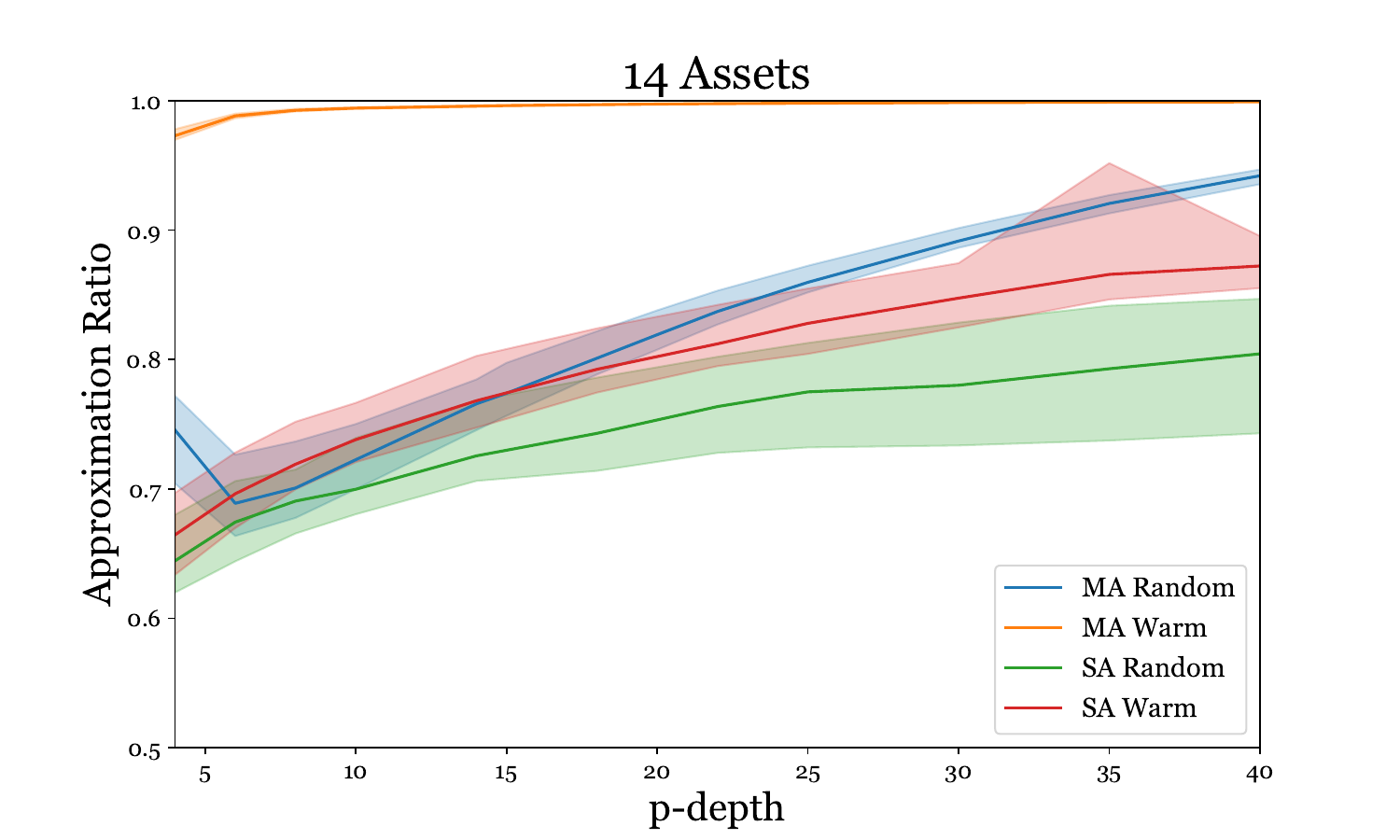} & 
        \includegraphics[width=0.48\linewidth]{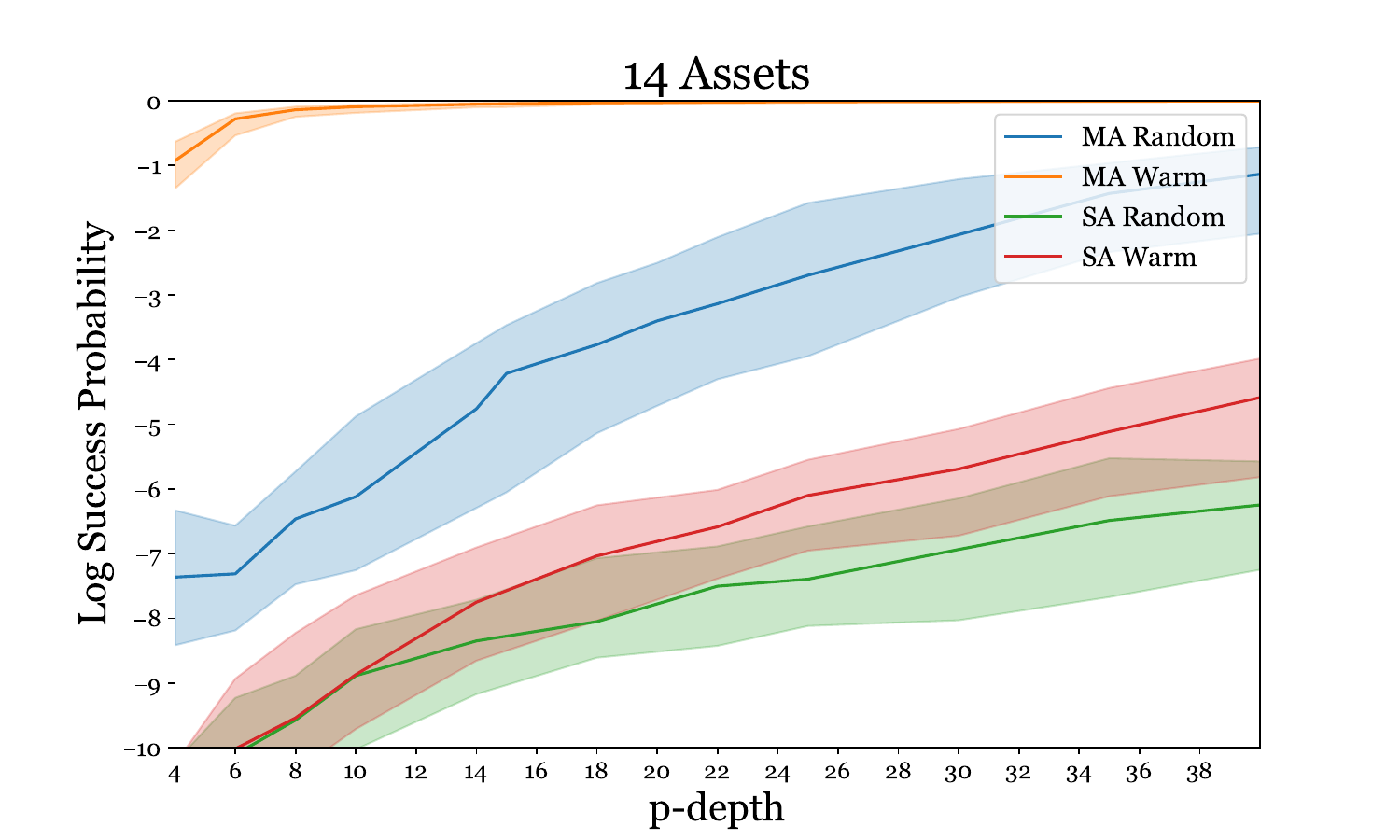} \\ 
        (c) & (d) \\
        \includegraphics[width=0.48\linewidth]{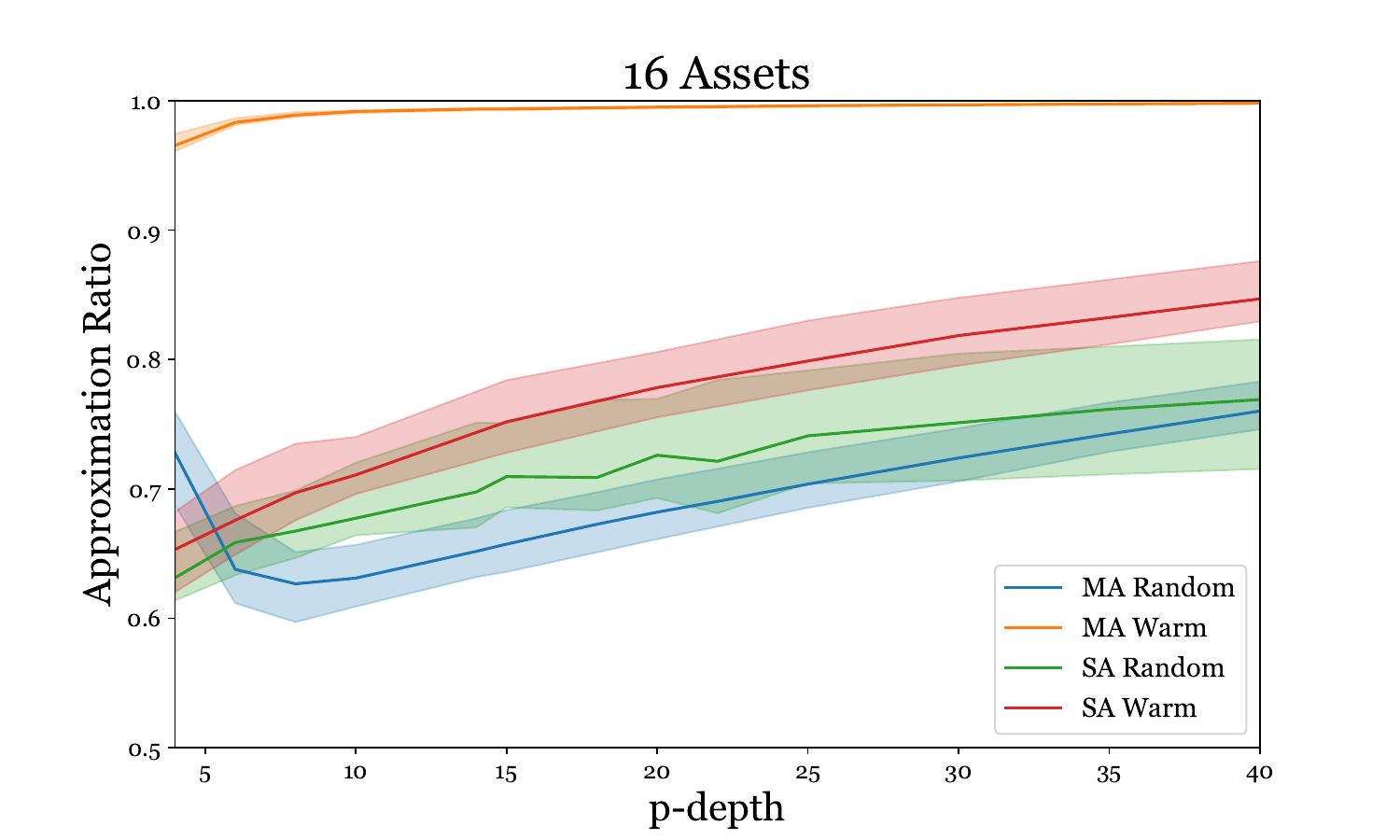} & 
        \includegraphics[width=0.48\linewidth]{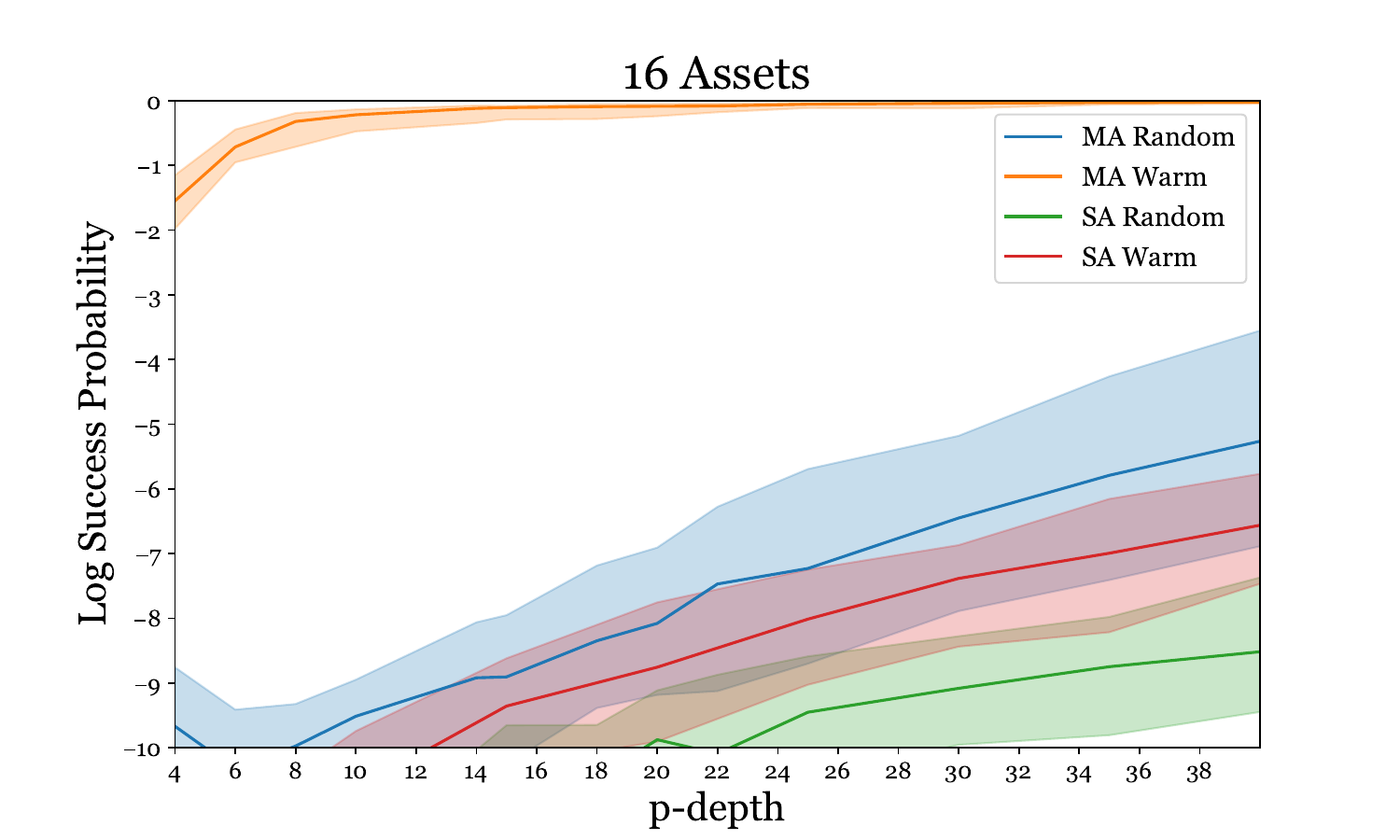} \\ 
        (e) & (f) %
        \end{tabular}
        \caption{\textbf{Warm Starting SA-QAOA vs. MA-QAOA for Portfolio Optimization.} Median performance (solid lines) of the best of $\numstarts$ randomly initialized QAOA-$p$ circuits for different depth $p$ over $\numinsts$ instances. Ribbons delineate the lower and upper quartile. MA-QAOA and SA-QAOA are shown $n=12,14,16$.}
        \label{fig:sp500qaoa}
    \end{suppfigure}

    \subsection{Shared Angle QAOA}\label{app:subsec_saqaoa}

        Due to the significantly better performance, the main text focuses on multi-angle QAOA, in which every gate is individually parameterized. The original protocols for QAOA~\cite{farhi_quantum_2014,hadfield_quantum_2019} delineate a much more \textit{structured} ansatz. In this section, we develop the \textit{shared angle} QAOA (SA-QAOA) which has only $2$ parameters per layer\footnote{$3$ parameters if we have a catalyst Hamiltonian.}. To solve optimization tasks with VQAs, Refs.~\cite{farhi_quantum_2014} introduced the Quantum Approximate Optimization Algorithm (QAOA1) paradigm. By introducing penalty terms, any constrained binary optimization task can be transformed into a quadratic unconstrained binary optimization (QUBO) problem. Using standard embedding techniques~\cite{lucas_ising_2014}, an Ising Hamiltonian can be constructed that encodes the energy spectrum of the QUBO cost function, $H_{QUBO}$. As in the Introduction section, we have a general cost function $f(x)$ subject to a constraint matrix $\bm{A}$:
        \begin{align*}
        \min_{x} \; f(x) \; \text{s.t.} \; \bm{A}  x = \bm{b}, \; x \in \{ 0, 1 \}^n 
        \end{align*}

        Recall that $ \sigma^{0} = \ketbra{0}{0} $ and $ \sigma^{1} = \ketbra{1}{1} $. Then each row of $\bm{A}$ is an embedded constraint operator:
        \begin{align} \label{eq:defembedcon}
        H_{A_{j}} = \sum_{jk} A_{jk} \, \sigma_{k}^{1}.
        \end{align}

        In the case of cardinality constrained optimization~\cite{cai2008parameterized}, we have:
        \begin{align}
        H_{A} = \sum_{k} \sigma_{k}^{1},
        \end{align}
        since $ \bm{A} = [1,\ldots,1] $ with dimension $ 1 \times n $. The cost function $f(x)$ is a general polynomial that can be written as:
        \begin{align} \label{eq:defcostfunc}
        f(x) = \sum_{j} a_{j} \prod_{k \in I_{j}} x_{k},
        \end{align}
        where $I_{j} \subseteq [ n ]$ is a list of indices. Then $f(x)$ can be embedded by
        \begin{align}\label{eq:defembedcostfunc}
        H_{f} = \sum_{j} a_{j} \prod_{k \in I_{j}} \sigma_{j}^{1}.
        \end{align}

        In particular $ f(x) = \bra{x} H_{f} \ket{x} $~\cite{leipold2024imposing}. Typically the cost function is at most quadratic: 
        \begin{align}
        H_{f} = h_{0}^{f} \, I + \sum_{j} h_{j}^{f} \, Z_{j} + \sum_{j<k}^{n} J_{jk}^{f} Z_{j} Z_{k}
        \end{align}
        Then we can define the total QUBO cost Hamiltonian using penalty terms:
        \begin{align}
        H_{QUBO} = H_{f} + B \sum_{j} \left( H_{A_{j}} - b_{j} \, I \right)^2,
        \end{align}
        where a sufficiently large $ B $ guarantees that each constraint is satisfied for the lowest energy state(s). In the case of cardinality constraint optimization:
        \begin{align}
        H_{QUBO} &= H_{f} + B \left( \sum_{j} \frac{1}{2} \left( 1 - Z_{j} \right) - k \, I \right)^2 \\
        &= H_{f} + B \left( \left(\frac{n}{2} - k \right) \, I - \sum_{j} \frac{1}{2} Z_{j} \right)^2 \\
        &= H_{f} + B \left( \left(\left( \frac{n}{2} - k \right)^2 + \frac{n}{4} \right) \, I - \left( \frac{n}{2} - k \right) \sum_{j} Z_{j} + \frac{1}{2} \sum_{j<k} Z_{j} Z_{k} \right)
        \end{align}
        
        Then by the adiabatic theorem~\cite{farhi_quantum_2000,albash_adiabatic_2018}, the slowly evolving quantum system initialized in the uniform superposition state,
        \begin{align} 
        \ket{s} = \frac{1}{\sqrt{2^{n}}} \sum_{x \in \{0,1\}^{n}} \ket{x},
        \end{align}
        will approximately prepare the ground state\footnote{A state in the ground space if degenerate.} by a dynamic Hamiltonian that interpolates between the transverse field and the Ising Hamiltonian encoding the problem:
        \begin{align}
        (TFQA)\quad H(s(t)) &= (1-s(t)) \, H_{TF} + s(t) \, H_{QUBO}, \\ 
        H_{TF} &= -\sum_{j} X_{j}, 
        \end{align}
        where $t \in (0,T]$ is the clock time and $ s(t) \in [0,1] $ is the dimensionless time. Transverse field quantum annealing (TFQA) is a mainstay procedure~\cite{ronnow_defining_2014} for quantum optimization with specialized devices. Ref.~\cite{hen_quantum_2016} introduced Constrained Quantum Annealing (CQA). Assuming the initial state $\ket{\psi(0)} $ is prepared such that $ H_{A_{j}} \ket{\psi(0)} = b_{j} \ket{\psi(0)} $ and we have a driver Hamiltonian $ H_{d} $ such that $ [H_{d}, H_{A_{j}} ] = 0 $, then we can use a quantum annealing schedule without needing penalty terms $(H_{A_{j}} - b_{j} \, I)^2$:
        \begin{align}
        (CQA)\quad H(s(t)) &= (1-s(t)) \, H_{d} + s(t) \, H_{f} \\ 
        [ H_{d} , H_{A_{j}} ] &= 0 \text{ for each } A_{j}
        \end{align}

        In general, finding such a $H_{d}$ is NP-Hard~\cite{leipold_constructing_2021,leipold2024imposing}, but for many important constraint classes mixers have been developed. In the case of the cardinality constraint, we can use an XY driver of (for example) cycle topology:
        \begin{align}
        (XYQA)\quad H(s(t)) &= (1-s(t)) \, H_{XYcycle} + s(t) \, H_{f} \\ 
        H_{XYcycle} &= -\sum_{j} \sigma_{j}^{x} \, \sigma_{j+1}^{x} + \sigma_{j}^{y} \, \sigma_{j+1}^{y}  
        \end{align}

        The Quantum Approximate Optimization Algorithm (QAOA1) is defined by the application of parameterized X mixer and a phase-separating operator associated with $H_{QUBO}$ as:
        \begin{align}
        (QAOA1) \quad U(\bm{\alpha}, \bm{\beta}) &= \prod_{j=1} U_{X}(\beta_{j}) U_{c}(\alpha_{j}) \\
        U_{X} &= \prod_{k=1}^{n} e^{i \beta_{j} X_{k}} \\ 
        U_{QUBO} &= e^{i \alpha_{j} H_{QUBO}} \\ 
        &= e^{i \alpha_{j} H_{f}} \prod_{A_{j}}^{n} e^{i \, \alpha_{j} \, \left( H_{A_{j}} - b_{j} \, I \right)^2},
        \end{align}
        where $\alpha,\beta \in \R^{p}$ are angles for application.

        The closest analog to CQA is H-QAOA where circuit is defined by repeated application of $ U_{c}(\alpha) = e^{i \alpha H_{f}} $ (as the phase-separating operator) followed by $ e^{i \beta \sum_{j} M_{j} } $ (as the mixing operator) for single layer parameters $ \alpha, \beta $. In general, $ M_{k}, M_{\ell} $ may not commutate, making the (approximate) application $ e^{i \beta \sum_{j} M_{j} } $ expensive. Instead, $ \mathcal{M} $ is typically ordered for mutually commuting operators to allow simultaneous application. In SA-QAOA, we have angles $ \bm{\alpha}, \bm{\beta} \in \R^{p} $ and apply $ U_{c}(\alpha_{\ell}) $ followed by $ U_{d}(\beta_{\ell}) = \prod_{M\in\mathcal{M}} e^{i \beta_{\ell} M} $ for layer $\ell$ in the circuit. SA-QAOA is defined by applying a $p$-depth circuit with each layer corresponding a the phase-separating and mixing operator: 
        \begin{align}\label{eq:saqaoa}
        (SA\text{-}QAOA) \quad U_{SA}(\bm{\alpha}, \bm{\beta}) &= \prod_{j=1}^{p} U_{d}(\beta_{j}) U_{c}(\alpha_{j}) \\ 
        U_{d}(\beta_{j}) &= \prod_{M\in\mathcal{M}} e^{i \beta_{j} M} \\ 
        U_{c}(\alpha_{j}) &= e^{i \alpha_{j} H_{f}}
        \end{align}
        
        \noindent This typically leads to between $\mathcal{O}(n)$ and $\mathcal{O}(n^2)$ gates per layer and only 2 parameters per layer, which is a highly structured ansatz.

        \subsubsection{Warm Starting Shared Angle QAOA}\label{subsec:warmstart_saqaoa}

        The DLA of $ U $ (for MA-QAOA) defined over $ \mathcal{M} =  \{  X_{j} \}_{j=1}^{n} $ and $ \mathcal{N} = \mathcal{G}_{Z,ZZ} $ is $\su{2^n} $ with dimension $4^{n}-1$ (universal for quantum computing over $n$ qubits). Note that the DLA of $ U_{SA} $ (for SA-QAOA) is strictly less expressive than $ U $ of MA-QAOA if $ H_{f} \in \text{span}\left(\mathcal{N}\right) $. Striking though is that for typical instances of interest the DLA of $ U_{SA} $ saturates to the same dimension as for the DLA of $ U $ despite how structured the ansatz is. For problems with high structure, such \texttt{Maximum Cut} on complete graphs or cycles~\cite{allcock_dynamical_2024}, $ U_{SA} $ has a polynomial sized DLA while $ U $ exhibits an exponential sized DLA. This showcases the subleties that can be involved in analyzing the DLA of highly structured ans{\"a}tze such as $ U_{SA} $ for VQAs. 
        
        Intuitively, families of instances that exhibit hardness tend to exhibit superpolynomial or exponential DLAs. We found for the instances we considered, the DLA of $U_{SA}$ grows exponentially as expected and so warm starting is an interesting approach for SA-QAOA.

        For SA-QAOA, $ U_{d} $ is utilized for both the restricted and full VQA, similar to MA-QAOA discussed in the main text. What is restricted is the associated phase-separating operator: 
        \begin{align}
            H_{f}^{WS} &= \text{Proj}_{\mathcal{G}_{Z}}( H_{f} ), \\ 
            U_{c}^{WS}(\alpha) &= e^{i \, \alpha H_{f}^{WS}}.
        \end{align}

        \noindent Suppose $ H_{f} = \sum_{j} h_{j} Z_{j} + \sum_{j<k} J_{jk} Z_{i} Z_{j} $ as is typical for Ising embeddings, then $ H_{f}^{WS} = \sum_{j} h_{j} Z_{j} $. Then the circuit for warm starting SA-QAOA of depth $p$ is:
        \begin{align}
            U_{SA}^{WS}(\bm{\alpha}, \bm{\beta}) = \prod_{j=1}^{p} U_{d}(\beta_{j}) U_{c}^{WS}(\alpha_{j}).
        \end{align}
        
        \noindent Then full circuit of SA-QAOA now with three parameter vectors is:
        \begin{align}
            U_{SA}(\bm{\alpha}, \bm{\beta}, \bm{\gamma}) = \prod_{j=p}^{1} U_{d}(\beta_{j}) U_{c}^{WS}(\alpha_{j}) U_{c}^{FC}(\gamma_{j}),
        \end{align}
        
        \noindent where $ U_{c}^{FC}(\gamma) = e^{i \gamma (H_{f} - H_{f}^{WS}) }$. We considered parameter transfer of $ \alpha_{j} $ to $ U_{c}(\alpha_{j}) $ to retain two angles, but this structured transfer was always worse than identity initialization of $ U_{c}^{FC} $ by setting $ \gamma = [ 0, \ldots, 0 ] \in \R^{p} $. Then loss function, Eq.~\ref{eq:loss}, of the warm start for SA-QAOA and MA-QAOA is:
        \begin{align}
            \mathcal{L}_{SA}^{WS}(\bmt^{WS}) &= \bra{D^n_k} U_{SA}^{WS}(\bmt^{WS})^{\dagger} \, H_{f} \, U_{SA}^{WS}(\bmt^{WS}) \ket{D^n_k}, \\ 
            \mathcal{L}^{WS}(\bmt^{WS}) &= \bra{D^n_k} U^{WS}(\bmt^{WS})^{\dagger} \, H_{f} \, U^{WS}(\bmt^{WS}) \ket{D^n_k}.
        \end{align}
        
        \noindent The loss function of the full circuit for SA-QAOA and MA-QAOA is:
        \begin{align}
        \mathcal{L}_{SA}(\bmt) &= \bra{D^n_k} U_{SA}(\bmt)^{\dagger} \, H_{f} \, U_{SA}(\bmt) \ket{D^n_k}, \\ 
        \mathcal{L}(\bmt) &= \bra{D^n_k} U(\bmt)^{\dagger} \, H_{f} \, U(\bmt) \ket{D^n_k}.
        \end{align}

    \subsection{Shared Angle QAOA versus Multi-Angle QAOA on Portfolio Optimization}\label{subsec:maqaoa_vs_saqaoa}

    \cref{fig:sp500qaoa} shows the approximation ratio and success probability for MA-QAOA and SA-QAOA with and without warm starting on instances of the \texttt{Portfolio Optimization} task. We see that warm starting MA-QAOA significantly outperforms all other approaches for every depth and that this improvement increases for larger number of assets. Due to this significant outperformance, we focused on MA-QAOA in the main text. 

\section{Preliminaries}\label{sec:appendix_liealgs}

    \subsection{Notation}\label{subsec:notation}
        We let $[m] := \bigB{1,2,\dots,m}$ and $[j,k] := \bigB{j, j+1, \dots, k}$ be intervals of integers. For $x \in \{0,1\}^n$ we let $\abs{x}_1 := \# \{i : x_i = 1\}$.
        
        The standard Hermitian Pauli spin operators over a single qubit are given by:
            \begin{alignat}{2}
                X &= \ketbra{0}{1} + \ketbra{1}{0} &&= \begin{pmatrix} 0 & 1 \\ 1 & 0  \end{pmatrix},\label{pauli-X} \\
                Y &= -i \bigP{\ketbra{0}{1} - \ketbra{1}{0}} &&= \begin{pmatrix} 0 & -i \\ i & 0  \end{pmatrix},\label{pauli-Y} \\
                Z &= \ketbra{0}{0} - \ketbra{1}{1} &&= \begin{pmatrix} 1 & 0 \\ 0 & -1  \end{pmatrix}. \label{pauli-Z}
            \end{alignat}
    
        \noindent Let $A$ be a 1-qubit operator. For $M \subseteq [n]$, we use $A_M$ to be the $n$-qubit operator that applies $A$ to qubits $j \in M$ and identity to those $j \not \in M$. $M_k$ is used for the singleton case. We use $M_{\overline{k}}$ as shorthand for $M_{[n] \setminus \{k\} }$.

        For $A,B \in \{X,Y\}$, We use the shorthand
        \begin{align}
            \AB{AB}{j,k} &:= A_j Z_{[j+1:k-1]}B_k \\
            \ABopp{AB}{j,k} &:= Z_{[1:j-1]}A_jB_kZ_{[k+1:n]}.
        \end{align}
        
        \noindent for any $1 \leq j < k \leq n$. Commutation relations between these types of terms are given in \cref{sec:com_calcs}.

    \subsubsection{The XY Mixer}\label{subsec:defs_paulis_XY}

        \begin{suppfigure}
        \centering 
        \begin{tabular}{cc}
        \includegraphics[width=0.30\textheight]{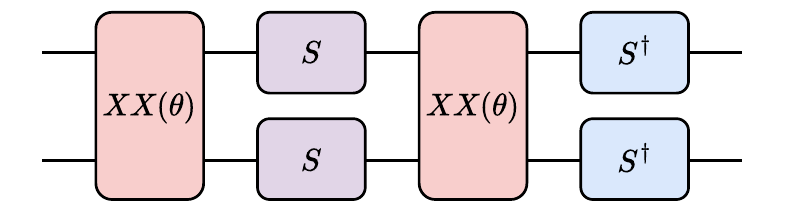} & \includegraphics[width=0.39\textheight]{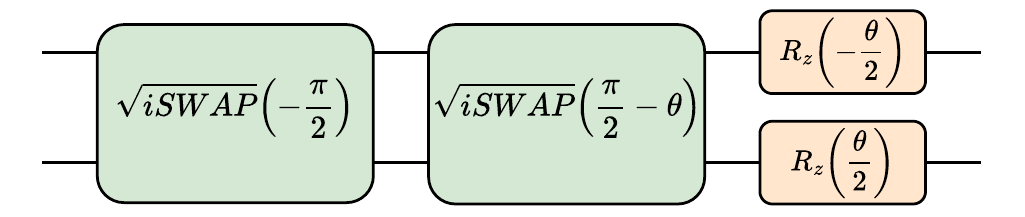} \\ 
        (a) & (b) 
        \end{tabular}
        \caption{\textbf{Circuit implementations of the XY Mixer.}  The XY mixer can be implemented with two qubit-pair gates and single qubit rotations. (a) shows an implementation well suited for trapped-ion devices that have a native XX gate~\cite{blumel2021power} while (b) was proposed by Ref.~\cite{abrams2019implementation} to exploit controlled phased $\sqrt{i\text{SWAP}}$ gates on superconducting devices.}
        \label{fig:xy_circuit}
        \end{suppfigure}
    
        The Pauli spin operators are known to satisfy simple commutation relationships:
        \begin{align}
        [X,Y] &= -[Y,X] = 2 i \, Z, \\
        [Y,Z] &= -[Z,Y] = 2 i \, X, \\ 
        [Z,X] &= -[X,Z] = 2 i \, Y.
        \end{align}
        
        \noindent They are widely used as a traceless orthogonal basis in quantum theory to write Hermitian operators, since $ \su{2} = \text{span}\left(\{ iX, iY, iZ \}\right) $ and $ \su{2^{n}} = \text{span}\left( i \{ I_{2}, X, Y, Z \}^{\otimes n} \setminus \{ I \} \right) $, where $ I $ is the identity operator over $\C^{2^{n}}$ and $ I_{2} $ is the 2 by 2 matrix identity over $\C^{2}$. 
        
        Over two qubits, the XY-mixer can be written in the standard basis as a real symmetric matrix:
        \begin{align}
        XY = \frac{1}{2}(X \otimes X + Y \otimes Y) = 
        \begin{pmatrix}
        0 & 0 & 0 & 0 \\
        0 & 0 & 1 & 0 \\
        0 & 1 & 0 & 0 \\
        0 & 0 & 0 & 0
        \end{pmatrix}. \label{eq:xy4by4}
        \end{align}
        
        \noindent The YX-mixer over two qubits in the standard basis can be written as an imaginary skew-symmetric matrix:
        \begin{align}
        YX = \frac{1}{2}(X \otimes Y - Y \otimes X) = 
        \begin{pmatrix}
        0 & 0 & 0 & 0 \\
        0 & 0 & -i & 0 \\
        0 & i & 0 & 0 \\
        0 & 0 & 0 & 0
        \end{pmatrix}. \label{eq:yx4by4}
        \end{align}
        
        \noindent Relevant commutations for the subsequent sections are:
        \begin{align}
        [XY, YX] &= i \, (I \otimes Z - Z \otimes I), \\ 
        [XY, Z \otimes I ] &= -[XY, I \otimes Z] = 2 i \, YX, \\ 
        [YX, I \otimes Z ] &= -[YX, Z \otimes I] = 2 i \, XY, \\
        [XY, Z \otimes Z] &= [YX,Z \otimes Z] = 0.
        \end{align}

        The XY mixer generator gives rise to the parameterized XY gate:
        \begin{align}
            XY(\theta) = e^{i \, \theta (XX + YY)}. 
        \end{align}
        Because $[X \otimes X,Y \otimes Y] = 0$, the XY-mixer has a straightforward circuit implementation~\cite{crooks2020gates} and is well suited for near term devices~\cite{abrams2019implementation,blumel2021power}. \cref{fig:xy_circuit}(a) shows an implementation suited for trapped-ion quantum computers, but it can also be decomposed further into the standard $\text{CNOT}, R_{z}, S, H $ gates by recognizing that $ZZ(\theta) = \text{CNOT}(I\otimes R_{z}(\theta))\text{CNOT}$ and $XX(\theta) = (H \otimes H) ZZ(\theta) (H\otimes H)$. 

        The QAOA requires an ordering of non-commuting mixers based on the sequence they should be applied on the device. Note that $ [ XY_{j,k}, XY_{p,q} ] \neq 0 $ if one of the two qubit indices overlap (see \cref{subsec:defs_paulis_XY} for two qubit commutation relations). The natural decomposition to apply the $n$ operators of $\mathcal{G}_{XY}^{C}$ in $2$ rounds is to split the operators into collections of commuting operators: $ \mathcal{G}_{XY}^{C} = \mathcal{G}_{XY}^{C,0} \sqcup \mathcal{G}_{XY}^{C,1} $ with $ \mathcal{G}_{XY}^{C,0} = \{ XY_{2 \, i,2 \, i+1} \}_{i=1}^{n/2} $ and $ \mathcal{G}_{XY}^{C,1} = \{ XY_{2 \, i - 1,2 \, i} \}_{i=1}^{n/2} $. In this way, $e^{i\sum_j \beta_j M_j} = \prod_j e^{i \beta_j M_j}$ for each $ \mathcal{G}_{XY}^{C,i} $. We then define the mixing operator
        \begin{align}
        U_{d}(\beta) = \prod_{M_{j} \in \mathcal{G}_{XY}^{C}} e^{i \, \beta M_{j} } \equiv \prod_{M_{j} \in \mathcal{G}_{XY}^{C,0}} e^{i \, \beta M_{j} } \prod_{M_{k} \in \mathcal{G}_{XY}^{C,1}} e^{i \, \beta M_{k} }, 
        \end{align}
        such that $ \mathcal{G}_{XY}^{C} $ is an ordered set and we always apply it's operators based on the underlying order. 

    \subsection{Lie Algebras}\label{subsec:lie_algebras}

        A \textit{Lie algebra} is a vector space $\mfg$ over a field $F$ along with a bilinear operator $[\cdot, \cdot]: \mfg \times \mfg \rightarrow \mfg$ that satisfies:
        \begin{enumerate}[label=(LA\arabic*), left=1.0em, labelsep=1.0em, itemsep=0.5em]
            \item Skew-symmetry: $[a,b] = -[b,a]$ for all $a,b \in \mfg$.
            \item Jacobi identity: $[a,[b,c]] = [[a,b],c] + [b,[a,c]]$ for all $a,b,c \in \mfg$.
        \end{enumerate}

        \noindent In this paper, we are interested in matrix Lie algebras over the field $\C$, where the commutator is defined as $[A,B] = AB - BA$. The \textit{general linear Lie algebra}, $\gl{m}$, is just $\C^{m \times m}$ with this commutator. $\gl{m}$ is a $m^2$-dimensional Lie algebra whose basis elements are
        
        \begin{equation}
            e_{j,k} := \ketbra{j}{k}
        \end{equation}
        
        \noindent for $j,k \in [m]$. Then $\gl{m} = \text{span}_{\C} \left( \{ e_{j,k} \}_{j,k}^{m} \right) $.

        For a fixed element $h \in \mfg$, it is useful to define the \textit{adjoint action}, denoted $ad_h : \mfg \rightarrow \mfg$ and given by $ad_h(g) := [h,g]$. Since commutation is bilinear\footnote{$[A+B,C] = [A,C] + [B,C]$ and $[A,B+C] = [A,B] + [A,C]$}, the adjoint action is a linear transformation over matrices. We say that $a, b \in \mfg$ \textit{commute} if  $[a,b] = 0$. By skew-symmetry, every element commutes with itself. $\mfg$ is \textit{abelian} if and only if every element commutes with one another. The \textit{center} of $\mfg$, denoted $Z(\mfg)$, is the set of operators that commute with every other element:
        
        \begin{equation}
            Z(\mfg) := \{g \in \mfg : [g,h] = 0\ \forall \, h \in \mfg\}
        \end{equation}
        
        A \textit{Lie sub-algebra} $\mfh \leq \mfg$ is a linear subspace closed under the bracket: $[h_1, h_2] \in \mathfrak{h}$ for any $h_1, h_2 \in \mathfrak{h}$. A Lie sub-algebra $\mfh$ is an \textit{ideal} if it absorbs under commutation: $[h,g] \in \mfh$ for any $h \in \mfh$ and $g \in \mfg$. The center $Z(\mfg)$ is always an ideal since for any $g \in \mfg$ and $h \in Z(\mfg)$, $[h,g] = 0 \in Z(\mfg)$. An ideal is said to be \textit{trivial} if it is the empty span $\{0\}$ or the full Lie algebra $\mfg$.

        Given two sets $A,B \subset \mfg$, extend the commutator operator to sets as
        
        \begin{equation}
            [A, B] = \text{span}\bigP{\bigB{\com{a, b} : a \in A, b \in B }}
        \end{equation}
        
        \noindent Many Lie algebras can be decomposed into sums of smaller Lie sub-algebras. $\mfg$ is a direct sum of Lie sub-algebras $\mfg_1, \dots, \mfg_n$, written as $\mfg = \mfg_1 \oplus \cdots \mfg_n$, if $\mfg$ is the direct sum of the $\mfg_i$ as a vector space and $[\mfg_i, \mfg_j] = \{0\}$ for all $i \neq j$. A Lie algebra $\mfg$ is \textit{simple} if it is non-abelian and contains no non-trivial ideals. $\mfg$ is \textit{semisimple} if it is a direct sum of simple Lie sub-algebras. Many Lie algebras of interest are semisimple.
        
        The study of Lie algebras is tightly connected to the study of Lie groups. We omit the rigorous connections but note that $\mfg$ is \textit{compact} if its associated Lie group is topologically compact. Further, the Lie group of special (orthogonal) unitary operators is generated by the Lie algebra of (real symmetric) skew-Hermitian operators ($\mathfrak{so}$) $\mathfrak{su}$ through exponentiation. Both $\so{m}$ and $\su{m}$ are compact Lie algebras and are defined formally later in this section. The following theorem is well-known; see Ref.~\cite{wiersema_classification_2023} proposition A.1. for a formal proof. 
        
        \begin{theorem}
            Let $\u{n} := \bigB{M \in \C^{n \times n} : M^\dagger = -M}$ be the Lie algebra of skew-Hermitian matrices. Any sub-algebra of $\u{n}$ is either abelian or a direct sum of simple, compact Lie algebras and a center as
        
            \begin{equation}
                \mfg = Z(\mfg) \oplus \mfg_1 \oplus \cdots \oplus \mfg_k
            \end{equation}
        \end{theorem}
        
        \noindent This decomposition theorem allows for the classification of Lie algebras defined over skew-Hermitian operators as the direct sum of known sub-algebras. The work in the following section is to show the specific decomposition of different Lie algebras into sums of $\so{m}$ and/or $\su{m}$ - the Lie algebras of skew-symmetric real operators and skew-Hermitian operators respectively - sometimes with a non-trivial center.
        
        Two Lie algebras can formally be related through congruence. A Lie algebra \textit{homomorphism} is a linear map $\phi : \mfg_1 \rightarrow \mfg_2$ that preserves the Lie bracket (commutation):
        
        \begin{equation}
            \phi([g_1,g_2]_1) = [\phi(g_1), \phi(g_2)]_2 \quad \forall g_1, g_2 \in \mfg_1; 
        \end{equation}
        
        \noindent $\phi$ is an \textit{isomorphism} if it is bijective, in which case we say the two Lie algebras are congruent: $\mfg_1 \cong \mfg_2$. Therefore the most straightforward way to show that two Lie algebras are isomorphic is to show (1) a bijection between linear bases that (2) preserves commutation relations.
        
        As discussed in \cref{app:controlgrad}, given a set of operators (the generators), the Lie algebra can constructed as the span over the nested commutation of its elements.
        
        \begin{definition}[Dynamical Lie algebra]
            Consider a set $\mathcal{A} = \{iH_1, \dots, iH_k\}$ of DLA generators. The \textit{dynamical Lie algebra (DLA)} is the (real) Lie algebra that is spanned by all nested commutators of $\mathcal{A}$, denoted by $\lie{\mathcal{A}}$. Equivalently, this is the smallest Lie algebra that contains $\mathcal{A}$.
        \end{definition}
        
        \noindent Given a PQC defined by layers of gate generators $\mathcal{G}$ as in \cref{eq:pqc}, their resulting DLAs define all the unitaries expressible at sufficient depth. 

        \begin{remark}
            Using the Jacobi identity and skew-symmetry, it can be shown that we only need consider nested commutations that are ``nested from the right'':
                \begin{equation}
                    [iH_1, [iH_2, [\dots,[iH_k, iH_{k+1}]\dots]]]
                \end{equation}
        \end{remark}      
        
        \begin{remark}
            If $\mg$ is a DLA generated by a set of generators $\mathcal{A}$, then an element $h \in Z(\mg)$ if and only if $[h,a] = 0$ for all $a \in \mathcal{A}$. Thus, in order to understand the center of a DLA, one must only ask about elements that commute with the generator set. (See Ref.~\cite{allcock_dynamical_2024} lemma 16 for a more general statement about symmetries and gate genetators.)
        \end{remark}     

        \subsubsection{The Special Orthogonal and Unitary Subalgebras}\label{subsec:sosu}
        In this manuscript, we use the unsubscripted span to be over the reals, such that $ \Span{A} := \text{span}_{\R}(A)$ and clarify if we represent a span over a different field through subscripts (i.e. $\text{span}_{\C} $ for a span over complex numbers). 
        
        \begin{enumerate}
            \item $\so{m} := \{M \in \gl{m} : M^T = -M\}$ is the Lie algebra of \textit{skew-symmetric} matrices. Through exponentiation, this Lie algebra generates the Lie group of orthogonal matrices. A basis for this space is given by 
        
            \begin{equation}
                a_{j,k} = \ketbra{j}{k} - \ketbra{k}{j}
            \end{equation}
        
            \noindent for $j < k$ in $[m]$ such that $ \so{m} = \text{span}_{\R}\left(\{ a_{j,k} \}_{j<k}^{m} \right) $. This space has dimension $\binom{m}{2} = \frac{m(m-1)}{2}$ (select an element in the upper triangle and its negation is selected in the lower triangle) and the non-zero commutation relations of $\so{m}$ are given by:
            \begin{align}
                &[a_{j,k}, a_{k,\ell}] = a_{j,\ell}, \tag{SO1}\label{eq:SO1}\\
                &[a_{j,\ell},a_{j,k}] = a_{k, \ell}, \tag{SO2}\label{eq:SO2}\\
                &[a_{k,\ell},a_{j,\ell}] = a_{j,k}, \tag{SO3}\label{eq:SO3}
            \end{align}
            
            \noindent for $j < k < \ell$.
        
            \item $\su{m} := \{M \in \gl{m} : M^\dagger =  -M,\ \tr(M) = 0\}$ is the set of skew-Hermitian matrices with trace 0. A real-valued skew-symmetric matrix is skew-Hermitian and so $\so{m} \subset \su{m}$. The structure of the Lie algebra $\su{m}$ is well-studied~\cite{Bossion2022_sun_structure_constants}. A typical basis for $\su{m}$ composes of three types of operators, denoted (real) skew-symmetric, (imaginary) symmetric, and diagonal. One way to define these operators is as follows:
            
            \begin{enumerate}
                \item Skew-symmetric: $a_{j,k}$ as above.
                \item Symmetric: $s_{j,k} := i \ketbra{j}{k} + i \ketbra{k}{j}$ for $j < k$ in $[n]$.
                \item Diagonal: $d_{j,k} := i \ketbra{j}{j} - i \ketbra{k}{k}$ for $j < k$ in $[n]$.
            \end{enumerate}
    
            Summing up the diagonal operators telescopes and so we only need to take $d_{1,2}, \dots, d_{n-1,n}$ to span the full set. These $m^2-1$ operators form a basis for $\su{m}$. They obey the following commutation relationships (all others are zero or can be derived from these, see \cref{rmk:SU_generates_others}). Here, let $j < k < \ell$ be indices in $[m]$:
            \begin{multicols}{2}
            \begin{enumerate}[label=(SU\arabic*), left=1.0em, labelsep=1.0em, itemsep=0.5em]
                \item $[a_{j,k}, a_{k,\ell}] = a_{j,\ell}$
                \item $[a_{j,\ell},a_{j,k}] = a_{k, \ell}$
                \item $[a_{k,\ell},a_{j,\ell}] = a_{j,k}$
                \item $[s_{j,k}, s_{k,\ell}] = -a_{j,\ell} $
                \item $[s_{j,\ell}, s_{j,k}] = a_{k,\ell} $
                \item $[s_{k,\ell}, s_{j,\ell}] = a_{j,k} $
                \item $[a_{j,k}, s_{j,k}] = 2 \, d_{j,k}$
                \item $[d_{j,k}, a_{j,k}] = 2 \, s_{j,k}$
                \item $[s_{j,k}, d_{j,k}] = 2 \, a_{j,k}$
                \item $[d_{j,k}, a_{k,\ell}] = -s_{k, \ell}$
                \item $[d_{j,r}, a_{j,k}] = s_{j,k}$ for $r \neq k$
                \item $[a_{j,k}, d_{k,\ell}] = s_{j,k}$
                \item $[d_{r,k}, a_{j,k}] = -s_{j,k}$ for $r \neq j$
            \end{enumerate}            
            \end{multicols}

            Any matrix in this space is orthogonal to the identity matrix $ I $, i.e. $ \u{m} = \text{span}( i \, I) \oplus \su{m} $. The Lie algebra of skew-Hermitian matrices, $\su{m}$, generates the Lie group of special unitary matrices through exponentiation. 
            
            Note that the $d_{j,j+1}$ terms defined here is slightly different than the normal Cartan generators of the $\su{m}$ basis. In fact, these are not orthogonal with respect to the trace inner product. However, they are simpler to handle for our purposes. There is a simple mapping between the terms defined here and the normal diagonal Cartan generators; see \cref{rmk:gellman-relationship}.
        \end{enumerate}

        \begin{remark}\label{rmk:SU_generates_others}
            These 13 commutation relationships are sufficient to describe the dynamics of $\su{n}$ but not exhaustive. However, other relationships can be generated from these. For example, we also have that $\com{s_{j,\ell}, a_{k,\ell}} = -s_{j,k}$:
            \begin{align}
                \com{s_{j,\ell}, a_{k,\ell}} &= \com{\com{d_{j,m}, a_{j,\ell}}, a_{k,\ell}} &\text{(SU11)} \\
                &=\com{d_{j,m}, \com{a_{j,\ell}, a_{k,\ell}}} - \com{a_{j,\ell}, \com{d_{j,m}, a_{k,\ell}}} &\text{(Jacobi)} \\
                &=\com{d_{j,m}, \com{a_{j,\ell}, a_{k,\ell}}} & \com{d_{j,m}, a_{k,\ell}} = 0\\
                &=-\com{d_{j,m}, a_{j,k}} &\text{(SU3)} \\
                &=-s_{j,k} &\text{(SU11)}
            \end{align}
        \end{remark}

        \begin{remark}\label{rmk:gellman-relationship}
            The symmetric and skew-symmetric terms correspond to the Gell-Mann basis for off-diagonal terms (see e.g. Ref.~\cite{bertlmann2008bloch}). The diagonal terms, known as Cartan generators, of the Gell-Mann basis are:
            \begin{align}
            G_{k} = \frac{i}{\sqrt{2 \, k \, (k+1)}} \left( -k \, \ketbra{k+1}{k+1} + \sum_{j=1}^{k} \ketbra{j}{j} \right). \label{eq:gellman}
            \end{align} 
            
            In general, the diagonal elements form a \textit{Cartan sub-algebra} of $\su{n}$ which is the maximal abelian sub-algebra. The Gell-Mann diagonal operators have the nice property of being traceless, mutually commutative, and orthogonal, while $ \{ d_{j,j+1} \}_{j=1}^{n-1}$ is traceless, mutually commutative, but non-orthogonal. However, we can relate our basis terms to the Cartan generators through the simple relation:
            \begin{align}
            G_{k} &= \frac{i}{\sqrt{2 \, k \, (k+1)}} \left( -k \, \ketbra{k+1}{k+1} + \sum_{j=1}^{k} \ketbra{j}{j} \right)  \\ 
            &= \frac{i}{\sqrt{2 \, k \, (k+1)}} \big( \ketbra{1}{1} + (2-1) \ketbra{2}{2} + \ldots + (k-(k-1)) \ketbra{k}{k} - k \ketbra{k+1}{k+1} \big) \\ 
            &= \frac{i}{\sqrt{2 \, k \, (k+1)}} \big( \left( \ketbra{1}{1} - \ketbra{2}{2} \right) + 2 \, \left( \ketbra{2}{2} - \ketbra{3}{3} \right) + \ldots + k \left(\ketbra{k}{k} - \ketbra{k+1}{k+1} \right) \big) \\ 
            &= \frac{i}{\sqrt{2 \, k \, (k+1)}} \sum_{j=1}^{k} j \, \left( \ketbra{j}{j} - \ketbra{j+1}{j+1} \right) \\ 
            &= \frac{1}{\sqrt{2 \, k \, (k+1)}} \sum_{j=1}^{k-1} j \, d_{j,j+1} \label{eq:gellman-alt}
            \end{align}
            and vice versa:
            \begin{align} 
            d_{k,k+1} &= i\left( \ketbra{k}{k}-\ketbra{k+1}{k+1} \right)  \\ 
            &= i/k \big( (1-1) \ketbra{1}{1} + .... + (1-1)\ketbra{k-1}{k-1} \\ 
            &\phantom{= 1/k \big(} + (-(k-1)-1)\ketbra{k}{k} + -k \ketbra{k+1}{k+1} \big) \\ 
            &= \frac{1}{k} \left( \sqrt{2 \, (k+1) \, (k+2)} \, G_{k+1} - \sqrt{2 \, k \, (k+1)} \, G_{k} \right)
            \end{align}
            
            \noindent In particular, these relationships can be used to describe an orthonormal basis for the DLA which is widely used. For our purposes, using $d_{j,k}$ as defined is fruitful since it is a more direct relationship with the terms generated below.
        \end{remark}

        \begin{remark}
            A common way to characterize a basis for an algebra is through its structure constants. See ~\cite{Bossion2022_sun_structure_constants} for the $\su{m}$ structure constants. We opt to use the commutation relations explicitly in this manuscript as these naturally correspond to elements constructed via the DLA process (and our operators are not necessarily orthogonal as per the previous remark).
        \end{remark}

    \section{Lie Algebras of XY-mixer Topologies}\label{sec:appendix_DLAs_description}
        In this section, we give a high-level overview of the XY-mixer DLAs. Rigorous calculations are provided in \cref{sec:appendix_DLAs_calcs}. Recall the 2-qubit terms $XY_{j,k} := \frac{1}{2}\bigP{X_jX_k + Y_jY_k}$ and $YX_{j,k} := \frac{1}{2}\bigP{X_jY_k - Y_jX_k}$.
        
        \subsection{Polynomially-Sized XY-mixer DLAs}\label{subsec:poly_dlas_desc}
                
            \subsubsection{Basis Elements}\label{subsubsec:poly_basis_els}
            
                Here we define the basis elements for the different Lie algebras associated with $XY$-mixers, with and without single $Z$ terms.  Normalizing coefficients are ignored; see \cref{sec:appendix_DLAs_calcs} for full details. Let $1 \leq j < k \leq n$ be fixed indices. \cref{fig:soubasis} presents visual descriptions of these basis elements on the left and right side of the figure.
                
                \begin{suppfigure}[t!]
                \centering 
                \includegraphics[width=0.9\linewidth]{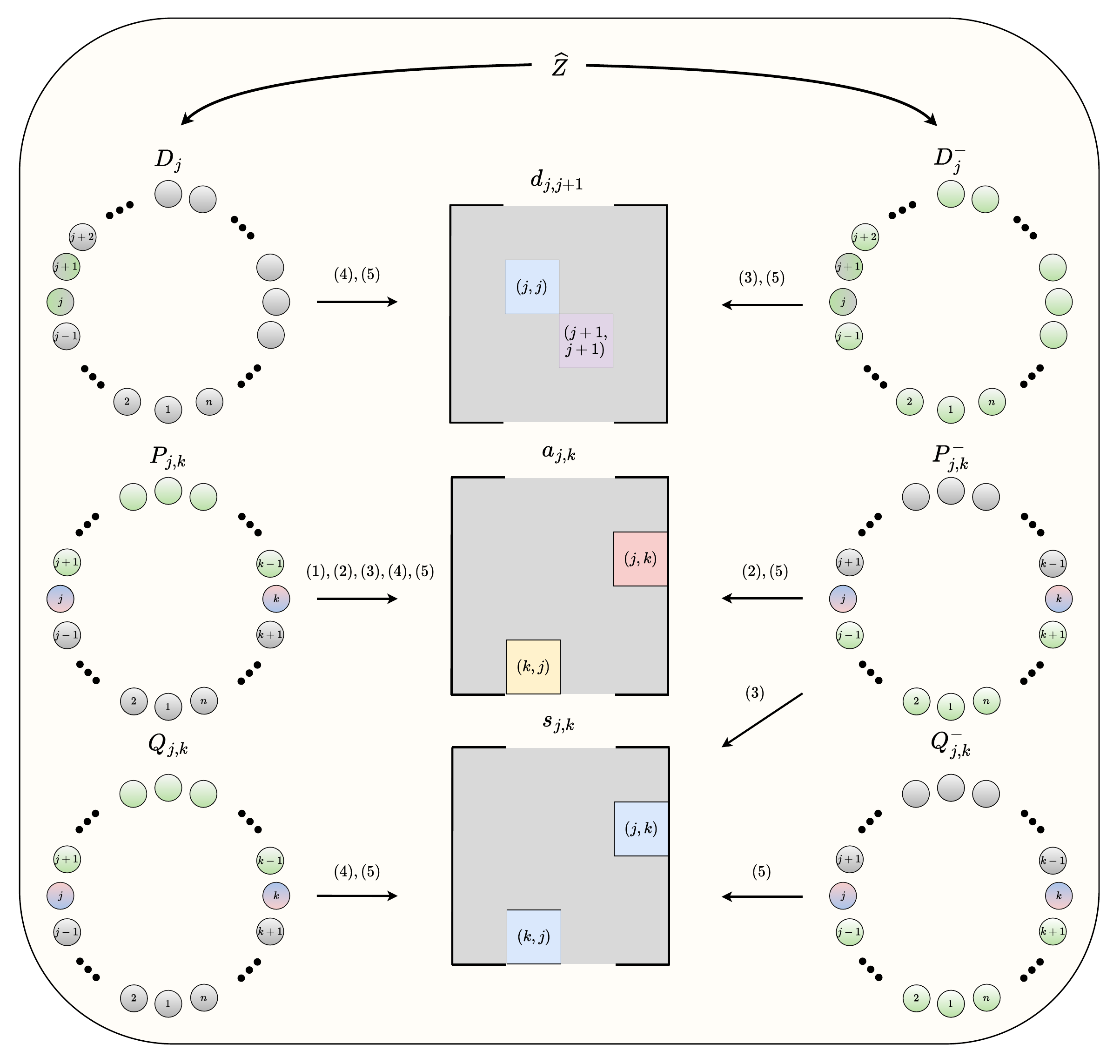}
                \caption{\textbf{Mappings from $\su{2^n}$ to $\su{n}$.} Left and right are visual representations of basis elements of the Lie algebras summarized in \cref{tab:dlasummary-1d} while center are basis elements of $\su{n}$. Labeled arrows show mappings of basis elements in $\su{2^n}$ to basis elements in $\su{n}$. For example, $\mfg_{XY}^{C,O} \cong \su{n}$ with $D_{j}^{-}$ mapped to $d_{j,j+1}$, $P_{j,k}$ mapped to $a_{j,k}$, and $P_{j,k}^{-}$ mapped to $s_{j,k}$ while $\mfg_{XY}^{C,E}{n} \cong \mathfrak{so}(n) \oplus \mathfrak{so}(n)$ with $P_{j,k}$ and $P_{j,k}^{-}$ mapped to $a_{j,k}$ (explicit mapping uses $P_{j,k} + P_{j,k}^{-}$ and $P_{j,k} - P_{j,k}^{-}$). Operators visualized here are defined in \cref{subsubsec:poly_basis_els} and in \cref{sec:appendix_DLAs_calcs}. }
                \label{fig:soubasis}
                \end{suppfigure}

                \begin{enumerate}
                \item $P_{j,k} := 
                        \begin{cases}\label{Pjk_simple}
                            X_j Z_{[j+1:k-1]} X_k + Y_j Z_{[j+1:k-1]}Y_k &(k-j \text{ odd}) \\ 
                            X_j Z_{[j+1:k-1]} Y_k - Y_j Z_{[j+1:k-1]} X_k &(k-j \text{ even})
                        \end{cases}$
                    
                    These terms are generated by nested commutations of neighboring $XY$ terms from $j$ up to $k$ (up to coefficient):
                    \begin{equation}
                        P_{j,k} = [iXY_{j,j+1}, [\dots [iXY_{k-2,k-1}, iXY_{k-1,k}]\dots]]
                    \end{equation}

                \item $\Popp{j,k} :=\label{Pjk_minus_simple}                \begin{cases}
                        Z_{[1:j-1]} \bigP{X_j X_k + Y_j Y_k} Z_{[k + 1:n]} &(n - k + j \text{ odd})\\ 
                        Z_{[1:j-1]} \bigP{X_j Y_k - Y_j X_k} Z_{[k + 1:n]} &(n - k + j \text{ even})\\
                    \end{cases}$
                    
                    These terms are generated by nested commutations of neighboring $XY$ terms $j$ \textit{down} to $k$, looping around zero:
                    \begin{equation}
                        \Popp{j,k} = [iXY_{j,j-1}, [iXY_{j-1,j-2}, [\dots [iXY_{k+2,k+1}, iXY_{k+1,k}]\dots]]]
                    \end{equation}
                    In other words, this follows the path around the cycle in the opposite direction of $P_{j,k}$. 
                \item $Q_{j,k} := \begin{cases}\label{Qjk_simple}
                            X_j Z_{[j+1:k-1]} Y_k - Y_j Z_{[j+1:k-1]} X_k &(k-j \text{ odd}) \\ 
                            X_j Z_{[j+1:k-1]} X_k + Y_j Z_{[j+1:k-1]}Y_k &(k-j \text{ even})
                        \end{cases}$
                
                    These terms are generated as $Q_{j,k} = [D_{j,k}, P_{j,k}]$.
                    
                \item $\Qopp{j,k} :=\label{Qjk_minus_simple}                 \begin{cases}
                        Z_{[1:j-1]} \bigP{X_j Y_k - Y_j X_k} Z_{[k + 1:n]} &(n - k + j \text{ odd})\\ 
                        Z_{[1:j-1]} \bigP{X_j X_k + Y_j Y_k} Z_{[k + 1:n]} &(n - k + j \text{ even})\\
                    \end{cases}$
                
                    These terms are generated as $Q_{j,k} = [D_{j,k}, \Popp {j,k}]$.
                \item \label{Djk_simple} $D_{j,k} := Z_j - Z_{k}$. These terms are in $\text{span}(\Gz{n})$.
                
                \item $\Dopp{j,k} = \Zbar{j} - \Zbar{k}$ where $ \Zbar{\ell} = Z_{\ell} \prod_{j=1}^{n} Z_{j} $ is the operator that has a Pauli-$Z$ on all qubits except identity at the $\ell^{th}$ location. The $\Dopp{j,j+1}$ terms are generated by nested commutations of neighboring $XY$ terms beginning and ending at $j$ for \textit{odd} $n$: 
                \begin{equation}
                    \Dopp{j,j+1} = [iXY_{j,j+1}, [\dots [iXY_{j-2,j-1}, XiY_{j-1,j}]\dots]]
                \end{equation}
    
                Note that $\Dopp{1,2}, \dots, \Dopp{n-2,n-1}$ are sufficient to linearly span all the $\bigB{\Dopp{jk}}_{j,k}$.
                    
                \end{enumerate}
                
                \noindent A few remarks. The $P$ and $P^-$ terms are identical except the ordering is flipped based on the parity of $k-j$; we have similar behavior for $Q$ and $Q^-$. Define the operator
    
                \begin{equation}
                    \hat{Z} := (-1)^{\floor{\frac{n}{2}}} \Zn 
                \end{equation}
                
                This operator tends to take $A \mapsto A^-$ and vice versa. This is helpful when both are in our Lie algebra as $P$ and $P^-$ may not commute but $P + P^-$ and $P - P^-$ do; see \cref{sec:appendix_DLAs_calcs} for specific details.

            \subsubsection{Lie Algebra Decompositions}\label{subsubsec:poly_dlas_decomp}
                \begin{supptable}[t]
                \centering 
                \begin{tabular}{|c|c|c|c|c|}
                    \hline
                    & \textit{DLA} & \textit{Topology}  & \textit{Decomposition} & \textit{Dimension} \\
                    \hline
                    1 & $\gxyp{n}$ & Path  & $\so{n}$ & $n(n-1)/2$ \Tstrut\Bstrut\\
                    \hline 
                    2 & $\gxyce{n} $ & Cycle (even $n$) &  $ \so{n} \oplus \so{n} $  & $ n(n-1) $ \Tstrut\Bstrut\\
                    \hline 
                    3 & $\gxyco{n} $ & Cycle (odd $n$) &  $\su{n}$  & $ n^2 - 1 $ \Tstrut\Bstrut\\
                    \hline 
                    4 & $\gxypz{n}$ & Path &   $\su{n} \oplus \u{1} $& $n^2$ \Tstrut\Bstrut\\
                    \hline 
                    5 & $\gxycz{n}$ & Cycle &  $\su{n} \oplus \su{n} \oplus \u{1}$ &$ 2n^2 - 1$ \Tstrut\Bstrut\\ 
                    \hline
                \end{tabular}
                \caption{\textbf{Polynomial-sized $XY$-mixer Lie Algebras.} A summary of the decompositions and dimensions of relevant polynomial-sized DLA studied in this manuscript. \cref{fig:xytopdiags}(a)-(c),(e),(f) provide visual representations of each.}
                \label{tab:dlasummary-1d}
                \end{supptable}
        
                Here we show the semisimple Lie algebraic decomposition of the polynomial-sized DLAs. Most proofs are omitted; see \cref{sec:appendix_DLAs_calcs} for full details. See \cref{fig:soubasis} for representations of how basis elements from each $ \mathfrak{g} $ are mapped to basis elements of $ \so{n} $ or $ \su{n} $. The structure of $\lie{i\Gxyc{n}}$ depends on the parity of $n$ so let $\gxyce{n}$ and $\gxyco{n}$ be the even and odd versions of this DLA, respectively.
                
                \begin{enumerate}
                \item $\gxyp{n} \cong \so{n}$. 
                
                    In this first case, the only way to generate new elements from the gate generators is to nest neighboring $XY$ terms. So, a basis for $\gxyp{n}$ is given by the $P_{j,k}$ for $1 \leq j < k \leq n$. These obey the skew-symmetric commutation relations $SO1, SO2, SO3$. Thus, we have isomorphism via $P_{j,k} \mapsto a_{j,k}$ (mapping (1) in \cref{fig:soubasis}): 
                    \begin{align}
                    \gxyp{n} = \text{span}\left( \{ P_{j,k} \}_{j<k} \right) \cong \so{n}. 
                    \end{align}
                    
                \item $\gxyce{n} \cong \so{n} \oplus \so{n}$.
            
                    In the cycle case, we can now go from nodes $j$ to $k$ in two directions: generate $P_{jk}$ by going directly between (clockwise) or $\Popp{jk}$ by ``looping around'' (anti-clockwise), see \cref{fig:soubasis}. However, if we begin at node $j$ and proceed along the cycle all the way back to the start, the terms cancel out to zero. That is, for even $n$:
                    \begin{equation}
                        [iXY_{j,j+1}, [\dots [iXY_{j-2,j-1}, iXY_{j-1,j}]\dots]] = 0.
                    \end{equation}

                    So, $\gxyce{n} = \text{span}\bigP{\bigB{P_{j,k}, \Popp{\ell,m}}_{j,k;\ell,m}}$. The $\Popp{j,k}$ terms act like symmetric terms of $\su{n}$ (as opposed to the $P_{j,k}$ acting skew-symmetrically). For example, $[\Popp{1,2}, \Popp{2,3}] = P_{1,3}$. This means that $\text{span}\bigP{\bigB{\Popp{\ell,m}}_{\ell,m}}$ is not a Lie algebra itself.
                    
                    One can show that $\Zn \cdot P_{j,k} = \Popp{j,k}$. Define the shifted sets
            
                    \begin{equation}
                        A^+ = \bigB{\frac{1}{2}\bigP{I + \Zn} P_{j,k}: 1 \leq j < k \leq n},\ A^- = \bigB{\frac{1}{2}\bigP{I - \Zn} P_{j,k}: 1 \leq j < k \leq n}.
                    \end{equation}

                    Then $\gxyce{n} = \text{span}(A^+, A^-)$. Moreover, the sets $A^+, A^-$ both individually obey skew-symmetric commutation relations. Lastly, $[A^+, A^-] = \{0\}$. Therefore, we can split the DLA into 
                    \begin{align} 
                    \gxyce{n} = A^+ \oplus A^- \cong \so{n} \oplus \so{n}.
                    \end{align}
            
                \item $\gxyco{n} \cong \su{n}$.
            
                    The main difference between $\gxyce{n}$ and $\gxyco{n}$ is that when $n$ is odd, we now have that looping around through nested commutation does not cancel:
                    \begin{equation}
                        [iXY_{j,j+1}, [\dots [iXY_{j-2,j-1}, iXY_{j-1,j}]\dots]] \neq 0
                    \end{equation}
            
                    These in fact define the $\Dopp{j,j+1}$ terms. We have 
                    \begin{align}
                    \gxyco{n} = \text{span}\bigP{\bigB{P_{j,k}, \Popp{\ell,m}, \Dopp{a,a+1}}_{j,k;\ell,m;a}}.
                    \end{align}
                    The $P_{j,k}$ act as skew-symmetric terms, the $\Popp{\ell,m}$ act as symmetric terms, and the $\Dopp{a}$ as diagonal terms (mapping (3) in \cref{fig:soubasis}). Therefore, $\gxyco{n} \cong \su{n}$.
            
                \item $\gxypz{n} \cong \u{1} \oplus \su{n}$.
            
                    First we understand the 1-dimensional center $\u{1}$. Now that we have access to $\Gz{n}$ terms, we can construct the element 
                    \begin{equation}
                        Z^+ = i\sum_{j \in [n]}Z_j
                    \end{equation}
                    Indeed, $[Z^+, \Gxypz{n}] = \{0\}$ which implies that $[Z^+, \gxypz{n}] = \{0\}$. Therefore, 
                    \begin{align}
                        Z\bigP{\gxypz{n}} = \text{span}(Z^+) \cong \u{1}.
                    \end{align}
            
                    Since $\Gxyp{n} \subset \Gxypz{n}$, we have the skew-symmetric $P_{jk}$ terms here. Since $D_{j,k} = Z_j - Z_{k} \in \text{span}\bigP{\bigB{Z_a}_a}$, we must have $D_{j,k} \in \gxypz{n}$. Moreover, we can now use this to generate \mbox{$Q_{j,k} = [D_{j,k}, P_{j,k}]$.} There are no additional terms in this DLA and so 
                    \begin{align} 
                    \gxypz{n} = \text{span}\bigP{\bigB{Z^+, P_{j,k}, Q_{\ell,m},D_{a,a+1}}_{j,k;\ell,m;a}}.
                    \end{align}
                    The $Q_{\ell,m}$ terms obey symmetric commutation relations and the $D_{a,a+1}$ obey diagonal ones. Thus, there is a copy of $\su{n}$ inside $\gxypz{n}$. Together with the center this gives 
                    \begin{align} 
                    \gxypz{n} \cong \u{1} \oplus \su{n}.
                    \end{align}
            
                    It is worth noting that though $\gxyco{n} \cong \su{n}$, those are different elements than the copy of $\su{n}$ in $\gxypz{n}$. It is not a Lie sub-algebra in the way that $\gxyp{n} \subset \gxypz{n}$.
            
                \item $\gxycz{n} \cong \u{1} \oplus \su{n} \oplus \su{n}$.
            
                    The same center $\text{span}(Z^+) \cong \u{1}$ is applicable in this scenario.
            
                    Since we have a cycle connectivity, we know we have all the $P_{j.k}$ and $\Popp{\ell,m}$ terms. In turn, we can generate all the $Q_{j,k}$ and $\Qopp{\ell,m}$ terms using $D_{j,k}$ and $D_{\ell,m}$, respectively. For odd $n$, we know we can generate $\Dopp{j,j+1}$ as in the $\gxyco{n}$ case. For even $n$, we have that $\Dopp{j,j+1} = [P_{j,j+1}, \Qopp{j,j+1}]$. Thus, we have that 
                    \begin{equation}
                        \gxycz{n} = \text{span}\bigP{ \left\{ Z^{+}, P_{j,k}, \Popp{\ell,m}, Q_{j',k'}, \Qopp{\ell',m'}, D_{a,a+1}, \Dopp{b,b+1} \right\}}
                    \end{equation}
            
                    Recall the trick of using $\Zn$ to map elements to one another. Define $B^+$ and $B^-$ as
                    \begin{equation}
                        B^\pm:= \frac{1}{2}\bigP{I \pm \Zn}\bigB{P_{j,k},Q_{j,k},\Dopp{j,j+1} : 1 \leq j < k \leq n}
                    \end{equation}
            
                     \noindent  It is easy to see that $\Zn \cdot D_{a,a+1} = \Dopp{a,a+1}$. When $n$ is even, we have that $\Zn \cdot P_{j,k} = \Popp{j,k}$ and $\Zn \cdot Q_{\ell,m} = \Qopp{\ell,m}$. When $n$ is odd, we have that $\Zn \cdot P_{j,k} = -\Qopp{j,k}$ and $\Zn \cdot Q_{\ell,m} = \Popp{\ell,m}$. Either way, $\gxycz{n} = \text{span}(B^+, B^-)$. Moreover, $B^+ \cong \su{n}$ with the $\frac{1}{2}\bigP{I + \Zn}P_{j,k}$ acting as the skew-symmetric, $\frac{1}{2}\bigP{I + \Zn}Q_{j,k}$ the symmetric, and $\frac{1}{2}\bigP{I + \Zn}D_{j,j+1}$ the diagonal terms. This is true regardless of the parity of $n$. Equivalently, $B^- \cong \su{n}$. Lastly, $[B^+, B^-] = 0$, giving the full decomposition 
                     \begin{align} 
                     \gxycz{n} \cong \u{1} \oplus \su{n} \oplus \su{n}.
                     \end{align}
                \end{enumerate}

        \subsection{Exponentially-Sized XY-mixer DLAs}\label{subsec:exp_dlas_desc}

            \begin{supptable}[t]
            \centering 
            \begin{tabular}{|c|c|c|c|c|c|}
            \hline
            \textit{Qubits} & $\gxyk{n}$ & $\gxykz{n}$ & $\gxyczzz{n}$ & $ \gxyzzzk{n} $ & $\redglincon(n)$\Tstrut\Bstrut\\ 
            \hline 
            3   &   8   &     17   &    18  &     18 & 16 \Tstrut\Bstrut\\ 
            \hline
            4   &   31   &    66   &    67  &     67 & 65 \Tstrut\Bstrut\\
            \hline
            5   &   123  &   247   &   248  &    248 & 246 \Tstrut\Bstrut\\
            \hline
            6   &   457  &   918   &   919  &    919 & 917 \Tstrut\Bstrut\\
            \hline
            7   &   1712 &   3425  &  3426  &   3426 & 3424 \Tstrut\Bstrut\\
            \hline
            8   &   6429  & 12862 &  12863  &  12863 & 12861 \Tstrut\Bstrut\\
            \hline
            9   &   24305  & 48611 &  48612  &  48612 & 48610 \Tstrut\Bstrut\\
            \hline 
            $n$   &  $\frac{1}{2} \left( \binom{2n}{n} - 3 + (-1)^n \right) + \lfloor \frac{n}{2} \rfloor $ & $ \binom{2n}{n} - n $ & $ \binom{2n}{n} - n + 1 $ & $ \binom{2n}{n} - n + 1 $ & $ \binom{2n}{n} - n - 1 $ \Tstrut\Bstrut\\
            \hline
            \end{tabular}
            \caption{\textbf{Dimensions of exponential sized DLAs for small $n$.} For the XY-mixer clique topology and the XY-mixer path, cycle, and clique topologies with $R_{ZZ}$ clique, we report the dimension of the associated DLA. }
            \label{tab:dlasummary-exp-numerics}
            \end{supptable}

            We describe the Lie algebra of cardinality invariance as a fundamental superalgebra of any XY-mixer algebra and then define basis elements in each exponential Lie algebra to show each of the DLAs grows $\Omega(3^n)$. Based on their fundamental relationship to the Lie algebra of cardinality invariance, as numerically verified for small $n$, we conjecture semisimple decompositions for each exponential DLA. Detailed descriptions of the construction of basis elements and decompositions of the DLAs are provided in \cref{subsec:exp_dlas}. 

            \subsubsection{The Lie Algebra of Imposed Cardinality Invariance}\label{app:glincon}
        
                In this section we describe the Lie algebra associated with skew-Hermitian operators that commute with $ Z^{+} = i\sum_{j=1}^{n} Z_{j} $. Since the XY-mixer and any Pauli-Z string commutes, this algebra is pertinent for our analysis as all XY-mixer topologies studied in this manuscript must be a subalgebra. As discussed in the Methods section, $ Z^{+} $ is associated with the constraint of preserving cardinality or Hamming weight. As such, consider the decomposition of the Hilbert space by Hamming weight 
                \begin{align} 
                \mathcal{H} = \bigoplus_{k=1}^{n} \text{span}\left( F^{k} \right),
                \end{align}
                where $ F^{b} = \left\{ \ket{x} : |x| = b, \ket{x} \in \{ \ket{0}, \ket{1} \}^{\otimes n} \right\} $. Note that $|F^{b}| = \binom{n}{b}$ and $\text{dim}\left(\mathcal{H}\right) = \sum_{k=0}^{n} |F^{k}| = \sum_{k=0}^{n} \binom{n}{k} = 2^{n} $. In particular for $ \ket{x} \in F^{b} $, we have 
                \begin{align}
                \bra{x} Z^{+} \ket{x} &= \bra{x} \left( \sum_{j=1}^{n} I - 2 \ketbra{1}{1}_{j} \right) \ket{x} \\ 
                &= n - 2 \sum_{j=1}^{n} x_{j} \\ 
                &= n - 2 \, b.
                \end{align} 
                Define
                \begin{align}\label{eq:feasible_proj}
                    \hat{F}^{k} = \sum_{\ket{x} \in F^{k}} \ketbra{x}{x},
                \end{align} 
                as the projection operator onto the subspace of hamming weight $k$. Then $ \hat{F}^{k} \, \hat{F}^{k'} = \delta_{k,k'} \hat{F}^{k} $, $ \sum_{k=0}^{n} \hat{F}^{k} = I $ and the embedded constraint operator has the eigensystem decomposition:
                \begin{align} 
                    Z^{+} = i\sum_{k=0}^{n} (n - 2\, k) \hat{F}^{k}.
                \end{align}
        
                Then a Hermitian operator $ H $ that commutes with $Z^{+}$ can be decomposed over the eigenspaces of $Z^{+}$:
                \begin{lemma}\label{lem:eigcommute}
                $[H,Z^{+}] = 0 $ if and only if $ H = \sum_{k=0}^{n} \hat{F}^{k} H \hat{F}^{k} $.
                \end{lemma}
        
                \begin{proof}
                Suppose there exists $ \hat{F}^{k} H \hat{F}^{\ell} \neq 0 $ for $k \neq \ell$, then 
                \begin{align}
                \com{ \hat{F}^{k} H \hat{F}^{\ell}, Z^{+} } &= \sum_{j} (n - 2 j )\com{ \hat{F}^{k} H \hat{F}^{\ell} , \hat{F}^{j} } \\
                &= \left((n - 2k) - (n-2\ell)\right) \hat{F}^{k} H \hat{F}^{\ell} \\ 
                &= 2(\ell-k) \hat{F}^{k} H \hat{F}^{\ell} \\ 
                &\neq 0.
                \end{align}
        
                By the same logic if no such pair $k, \ell$ exists, $ H $ commutes, proving both directions.
                \end{proof}

                Consider the decomposition of $ \hlincon $ over the eigenspace of $Z^{+}$. Let 
                \begin{align}\label{eq:hlincon_k}
                \hlincon^{k} := \text{proj}_{F^{k}} \left( \hlincon \right) = \left\{ i H : i \hat{F}^{k} H \hat{F}^{k} = i H, iH \in \mathfrak{u}(2^n) \right\}, 
                \end{align} 
                be the Lie subalgebra of the skew-Hermitian over the subspace $\text{span}\left( F^{k} \right)$. This is a proper Lie algebra since skew-Hermiticity is preserved and the space is closed under commutation. Note that since $ \hat{F}^{k} H \hat{F}^{k} = H $ and $ \hat{F}^{k} \hat{F}^{k} = \hat{F}^{k} $ we also have $ \hat{F}^{k} H = \hat{F}^{k} \hat{F}^{k} H \hat{F}^{k} = H $ and $ H \hat{F}^{k} = \hat{F}^{k} H \hat{F}^{k} \hat{F}^{k} = H $. To see that the space is closed, we have for $ A, B \in \hlincon^{k} $:
                \begin{align}
                \com{ A, B } &= A B - B A \\ 
                &= \hat{F}^{k} A B \hat{F}^{k} - \hat{F}^{k} B A \hat{F}^{k} \\ 
                &= \hat{F}^{k} \com{ A,B } \hat{F}^{k},
                \end{align}
                and so $[A,B] \in \hlincon^{k}$ as well. Moreover, we define the \textit{traceless} skew-Hermitian operators: 
                \begin{align}\label{eq:glincon_k}
                \glincon^k = \left\{ g : g \in \hlincon^{k}, \text{Tr}(g)=0 \right\},
                \end{align}
                and so we have the decomposition:
                \begin{align}\label{eq:hlincon_k_decomp}
                \hlincon^{k} = \text{span}\left(\{ i \hat{F}^{k} \}\right) \oplus \glincon^k,
                \end{align} 
                since the projection of the identity operator into the space is $ \hat{F}^{k} I \hat{F}^{k} = \hat{F}^{k} \hat{F}^{k} = \hat{F}^{k} $. By dimensionality, we have:
                \begin{align}
                \hlincon^{k} \cong \u{1} \oplus \su{\binom{n}{k}}.
                \end{align}
                
                A fundamental Lie algebra utilized in subsequent sections is that of traceless skew-Hermitian operators up to relative phasing between the feasible subspaces (the reduced centralizer of $Z^{+}$):\begin{align}\label{eq:redglincon}
                    \redglincon &:= \bigoplus_{k=1}^{n-1} \glincon^{k} \\ 
                    &\cong \bigoplus_{k=1}^{n-1} \su{\binom{n}{k}}, 
                \end{align}
                
                since $\glincon^{0} = \glincon^{n} = \{0\} $ (the empty span). Note that the Lie algebra over these feasible subspaces are mutually commutative: $ \com{ \hlincon^{k}, \hlincon^{j} } = \com{ \glincon^{k}, \glincon^{j} } = \{ 0 \} $ if $j \neq k$.
        
                \begin{suppfigure}[!t]
                \centering 
                \includegraphics[width=0.5\textwidth]{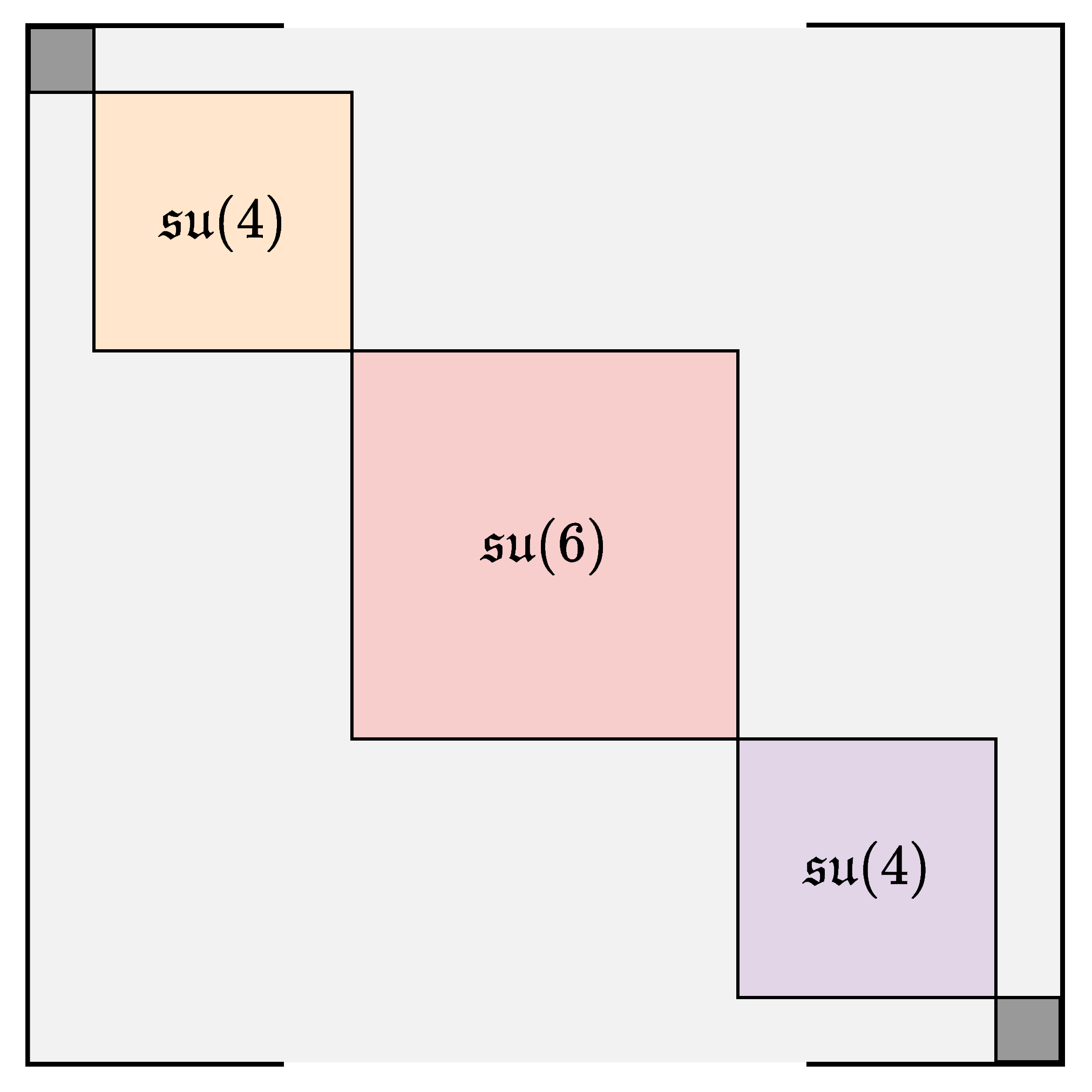}
                \caption{\textbf{Lie Algebra of Cardinality Invariance.} The space of traceless skew-Hermitian operators that commute with $ Z^{+} = i\sum_{j=1}^{n} Z_{j} $ can be depicted as a sum of Lie algebras over each Hamming weight preserving space. Here we depict the space in the case $ i (Z_{1} + Z_{2} + Z_{3} + Z_{4}) = 4 \, i \, \hat{F}^{0} + 2 \, i \, \hat{F}^{1} - 2 \, i \, \hat{F}^{3} - 4 \, i \, \hat{F}^{4} $ (defined in \cref{eq:feasible_proj}) that leads to $ \redglincon(4) = \su{4} \oplus \su{6} \oplus \su{4} $ (defined in \cref{eq:redglincon}) with $4 = \binom{4}{1}, 6 = \binom{4}{2}, 4 = \binom{4}{3}$.  
                }
                \label{fig:glincon}
                \end{suppfigure}
        
                The relationship between these subalgebras over each feasible space and the centralizers $\hlincon \subseteq \u{2^{n}}$ and $ \glincon \subseteq \su{2^{n}} $ is captured by \cref{thm:lincon_decomp} stated in the main text.

                \begin{proof}[Proof \cref{thm:lincon_decomp}]
                It is clear that $ \hlincon^{k} \leq \hlincon $, since $ \hat{F}^{k} H \hat{F}^{k} = H $ implies $ [ H, Z^{+} ] = 0 $, and so suppose there is an operator $ i M \in \hlincon $ not in $\bigoplus_{k=0}^{n} \hlincon^{k} $. By \cref{lem:eigcommute}, $ i M = \sum_{k=0}^{n} i M^{k} $ where $ M^{k} = \hat{F}^{k} M \hat{F}^{k} $ and so each $ iM^k \in \hlincon^{k} $, which implies $i M \in \bigoplus_{k=0}^{n} \hlincon^{k}$. Then by \cref{eq:hlincon_k_decomp} and \cref{eq:redglincon}: 
                \begin{align}
                \hlincon &= \bigoplus_{k=0}^{n} \hlincon^{k} \\ 
                &= \bigoplus_{k=0}^{n} \text{span}\left(\left\{ i \hat{F}^{k} \right\}\right) \oplus \glincon^{k} \\  
                &= \text{span}\left(\left\{ i \hat{F}^k \right\}_{k=0}^{n} \right) \oplus \redglincon \\ 
                &\cong \u{1}^{\oplus n} \oplus \bigoplus_{k=1}^{n-1} \su{\binom{n}{k}}
                \end{align}
        
                This completes the proof for \cref{eq:cong_hlincon}. Since $ \sum_{k=0}^{n} \hat{F}^{k} = I $ and $\text{Tr}(\hat{F}^k) = \binom{n}{k} $, we have:
                \begin{align} 
                \text{span}\left(\Big\{ i \hat{F}^{k} \Big\}_{k=0}^{n} \right) \cap \text{span}^{\perp} \left(\Bigg\{ i \sum_{k=0}^{n} \hat{F}^{k} \Bigg\}\right) = \text{span}\left(\Bigg\{ i \binom{n}{k}^{-1} \hat{F}^{k} - i \binom{n}{k+1}^{-1} \hat{F}^{k+1} \Bigg\}_{k=0}^{n-1} \right),
                \end{align}
                which leads to:
                \begin{align}
                \glincon &= \text{span}\left(\Bigg\{ i \binom{n}{k}^{-1} \hat{F}^{k} - i \binom{n}{k+1}^{-1} \hat{F}^{k+1} \Bigg\}_{k=0}^{n-1}\right) \oplus \redglincon \\ 
                &\cong \u{1}^{\oplus n-1} \oplus \bigoplus_{k=1}^{n-1} \su{\binom{n}{k}} 
                \end{align} 
        
                This completes the proof for \cref{eq:cong_glincon}. 
        
                \end{proof}
        
                The dimension of of the reduced Lie algebra $\redglincon$ is given by: 
                \begin{align}
                \text{dim}\left( \redglincon \right) &= \sum_{k=1}^{n-1} \text{dim}\left( \su{\binom{n}{k}} \right) \\ 
                &= \sum_{k=1}^{n-1} \binom{n}{k}^2 - 1 \\ 
                &= \binom{2 \, n}{n} - n - 1 \label{eq:exact_dims_ker_ad_Zplus} \\ 
                &= 4^{n - \Theta\left( \log(n) \right)}. \label{eq:bound_dims_ker_ad_Zplus}
                \end{align}
        
                In particular, $\text{dim}(\glincon) = \text{dim}(\redglincon) + n - 1 = \binom{2\,n}{n} - 2 $ is an upper bound on any Lie algebra $ \mfg \subseteq \mathfrak{su}(2^n) $ that commutes with $Z^{+}$: if $ \text{ad}_{Z^{+}}(g) = 0 $ for all $ g \in \mfg$, then $ \mfg \leq \glincon$ and so $ \text{dim}(\mfg) \leq \text{dim}(\glincon) $. \cref{fig:glincon} depicts the decomposition of $\redglincon$ for $n=4$.
        
                Let us connect this back to QAOA and prelude the subsequent section. The collection of phase-separating operators $ \mathcal{N} \subseteq \text{span}\left(\{ I_2, Z \}^{\otimes n} \setminus I \right) $ for QAOA commute trivially with $ Z^{+} $, such as $\Gz{n}$ and $\Gzz{n}$. The collection of mixers $ \mathcal{M} \subseteq \su{2^n} \cup \text{span}^{\perp}\left( \{I, Z\}^{\otimes n} \right) $ are selected such that $ [ M, Z^{+} ] = 0 $ for all $ M \in \mathcal{M} $ (see \cref{fig:flowchart}). Then the Lie algebra of QAOA with generators $\mathcal{M}, \mathcal{N}$ is a Lie subalgebra: $\lie{i \mathcal{M}, i \mathcal{N}} \leq \glincon $. This is true for both the polynomial sized DLAs considered in previous sections and the exponential sized DLAs considered in this work. The reduced Lie algebra $ \redglincon $ will play a central role in understanding $ \gxykz{}, \gxyczzz{}$.

            \subsubsection{Basis Elements}\label{subsubsec:exp_basis_els}
                We need to generalize the $P_{j,k}$ from the previous section to allow for arbitrary paths inside of a clique. Fix some $1 \leq j < k \leq n$ and let $ \sigma_{jk} \in \{0,1\}^{n-2} $. \begin{equation}\label{def:Psjk_simple}
                    P_{\sigma_{jk}} := 
                    \begin{cases}
                        XY_{j,k} \prod\limits_{\ell \neq j,k}Z^{\sigma_{jk}(\ell)}_\ell, &\abs{\sigma_{jk}}_1 \text{ odd}\\
                        YX_{j,k} \prod\limits_{\ell \neq j,k}Z^{\sigma_{jk}(\ell)}_\ell, &\abs{\sigma_{jk}}_1 \text{ even}\\
                    \end{cases}
                \end{equation}
                \noindent For example, setting $\sigma_{jk}(\ell) = 1$ when $\ell \in [j,k-1]$ and zero otherwise results in the original $P_{j,k}$ operator. Then $P_{\sigma_{jk}} $ can be constructed through nested $XY$ terms on a path beginning at node $j$ and going to each node $\ell$ where $\sigma_{jk}(\ell) = 1$ and then ending at node $k$. 

                Moreover, we can construct terms:
                \begin{equation}\label{def:Pmsjk_simple}
                    P_{\mu,\sigma_{jk}} := 
                    \begin{cases}
                        XY_{j,k} \prod\limits_{(p,q) \in \mu} YX_{p,q} \prod\limits_{\ell \notin V_{\mu} \cup \{j,k\}} Z^{\sigma_{jk}(\ell)}_\ell, & \abs{\mu} / 2 + \abs{\sigma_{jk}}_1 \text{ odd}\\
                        YX_{j,k} \prod\limits_{(p,q) \in \mu} YX_{p,q} \prod\limits_{\ell \notin V_{\mu} \cup \{j,k\}} Z^{\sigma_{jk}(\ell)}_\ell, & \abs{\mu} / 2 + \abs{\sigma_{jk}}_1 \text{ even}\\
                    \end{cases}
                \end{equation}
                
                \noindent where $\mu \subseteq \{ (p,q) : p,q \in [n] \setminus \{j,k\} , p \neq q \} $ is a set of ordered pairs, $V_{\mu} = \{ p : (p,q) \in \mu \} \cup \{ q : (p,q) \in \mu \} $ is the set of all individual nodes of each pair, and the pairs are disjoint such that $ 2\,|\mu| = |V_{\mu}| $.
        
                The construction follows from an appropriate $P_{\sigma_{jk}'}$ in $\gxyk{n}$ and applying $ XY_{p,q} $ on the appropriate slot for each $(p,q) \in \mu$ where $ \sigma_{jk}'(q) = 0 $ (such that we have identity on this qubit) and $ \sigma_{jk}'(p) = 1 $ (such that we have $ Z $ on this qubit).
            
            \subsubsection{Lie Algebra Decompositions}\label{subsubsec:exp_basis_decomp}
                The following are the conjectured semisimple Lie algebra decompositions, with numerically verified isomorphism for small $n$. Once again we use superscript E and O to denote even and odd $n$ cases. 

                \begin{enumerate}
                    \item $ \gxyke{n} \cong \gxykdecompeven $
                    \item $ \gxyko{n} \cong \gxykdecompodd $
                    \item $ \gxykz{n} \cong \u{1} \oplus \bigoplus_{k=1}^{n-1} \su{\binom{n}{k}}  $
                    \item $ \gxyczzz{n} \cong \u{1}^{\oplus 2}  \oplus  \bigoplus_{k=1}^{n-1} \su{\binom{n}{k}} $
                \end{enumerate}

                All of these DLAs have dimension $ 4^{n - \Theta\left( \log(n) \right)} $. These decompositions are due to a fundamental relationship to the Lie algebras described in \cref{app:glincon} as discussed in \cref{sec:appendix_DLAs_calcs}, we conjecture (and verify for small $n$) that $\gxykz{} = \text{span}\left( Z^{+} \right) \oplus \redglincon $ and $\gxyczzz{} = \text{span}\left( Z^{+} \right) \oplus \text{span}\left( ZZ^{+} \right) \oplus \redglincon $, where $ ZZ^{+} = i \sum_{j<k} Z_{j} Z_{k} $. 
        
    \section{Derivation of Lie Algebras}\label{sec:appendix_DLAs_calcs}
        In this section, we formally prove the Lie algebraic decompositions given in the Results section, building on the description given in \cref{sec:appendix_liealgs}. In \cref{subsec:dla_xy_path,subsec:dla_xy_cycle,subsec:dla_xy_z_path,subsec:dla_xy_z_cycle}, we provide explicit calculations for the polynomial-sized DLAs. In \cref{subsec:exp_dlas}, we handle the exponentially-large DLA calculations and provide motivation for the conjectures.

        \subsection{Lie Algebra of the $XY$ Path}\label{subsec:dla_xy_path}
            \begin{proposition}\label{prop:XY_DLA_path}
                $\gxyp{n} \cong \so{n}$.
            \end{proposition}
            
            This proposition is proven in the following steps.
            \begin{enumerate}
                \item Show that the only terms generated in the DLA process are the $P_{j,k}$ as defined below.
                \item Show that the $\bigB{P_{j,k}}_{j,k}$ is closed under commutation and so is a valid Lie algebra it self, showing that $\gxyp{n} = \Span{\bigB{P_{j,k}}_{j<k}}$. In doing so, we show that they satisfy the skew-symmetric relations \cref{eq:SO1,eq:SO2,eq:SO3}. This shows that $\text{span}\left(\bigB{P_{j,k}}_{j,k} \right)$ is isomorphic to $\so{n}$.
            \end{enumerate}        
            
            \begin{definition}[$P_{jk}$]
                 For any $1 \leq j < k \leq n$, define 
                \begin{align}
                    P_{j,k} &:= \frac{i^{c_{jk}}}{2} \begin{cases}\label{def_Pjk}
                        \AB{XX}{jk} + \AB{YY}{jk} &(k - j \text{ odd}) \\ 
                        \AB{XY}{jk} - \AB{YX}{jk} &(k - j \text{ even}) \\
                    \end{cases} \\
                    c_{jk} &:= 2\floor{\frac{k - j}{2}} + 1 \label{cjk_def}
                \end{align}
            \end{definition}
            
            \begin{lemma}\label{lem:nested_XY_Pij}
               The terms generated in the $\gxyp{n}$ DLA process given by
                \begin{equation}\label{XY_ring_nested_action}
                    ad_{iXY_{j,j+1}} \cdot ad_{iXY_{j+1,j+2}} \cdots ad_{iXY_{k-1,k}} = P_{j,k}
                \end{equation}    
            \end{lemma}
            
            Before we prove this statement, we give the following fact about the $c_{jk}$ terms. This is helpful when working with the $P_{j,k}$ terms.
            
            \begin{lemma}[$c_{jk}$ values, rephrased]\label{lem:c_jk}
                \begin{itemize}
                    \item[]
                    \item If $k - j$ is even, then $c_{jk} = k - j +1$. Also, $c_{j-1,k} = c_{jk}$. 
                    \item If $k - j$ is odd, then $c_{jk} = k - j$. Also, $c_{j-1,k} = c_{jk} +2$.
                \end{itemize}
            \end{lemma}
    
            \noindent The previous lemma gives the exact values of $c_{jk}$ but since they are used as exponents on $i$, we primarily care about their value modulo 4. This can be seen in the following proof and beyond. The proof of this lemma is omitted since it is a simple observation.
                
            \begin{proof}[Proof (\cref{lem:nested_XY_Pij})]
            
                We prove by inducting on $k-j$.
                
                As a base case, we have $iXY_{j,j+1} = \frac{i}{2}(\AB{XX}{j,j+1} + \AB{YY}{j,j+1}) = P_{j,j+1}$ which matches \cref{def_Pjk} since $(j+1) - j = 1$ is odd and $c_{j,j+1} = 1$. 
                
                Now assume by induction on $k-j$ for $j > 1$ that $ad_{iXY_{j,j+1}} \cdots ad_{iXY_{k-1,k}} = P_{j,k}$. If $k - j$ is even, one has
                \begin{align}
                ad_{iXY_{j-1,j}} \cdots ad_{iXY_{k-1,k}}  &= \left[iXY_{j-1,j},\ ad_{iXY_{j,j+1}} \cdots ad_{iXY_{k-1,k}} \right] \\
                     &\stackrel{IH}{=} \left[iXY_{j-1,j},P_{j,k} \right]\\
                     &= \frac{i^{1 + c_{jk}}}{4}[X_{j-1}X_{j} + Y_{j-1}Y_{j},\ \AB{XY}{jk} - \AB{YX}{jk}] \\[5pt]
                    &=  \frac{i^{1 + c_{jk}}}{4}\left(\begin{aligned}
                        &[X_{j-1}X_{j} ,\ \AB{XY}{jk} ] - [X_{j-1}X_{j},\ \AB{YX}{jk}] \\
                        &+[Y_{j-1}Y_{j},\ \AB{XY}{jk} ] - [Y_{j-1}Y_{j},\  \AB{YX}{jk}]
                    \end{aligned} \right) \\[5pt]
                    &=  \frac{i^{1 + c_{jk}}}{4}\left(\begin{aligned}
                        &X_{j-1}[X,X]_j Z_{[j+1:k-1]}Y_k -  X_{j-1}[X,Y]_j Z_{[j+1:k-1]}X_k \\
                        &+Y_{j-1}[Y,X]_j Z_{[j+1:k-1]}Y_k - Y_{j-1}[Y,Y]_j Z_{[j+1:k-1]}X_k  
                    \end{aligned} \right) \\[5pt] 
                    &=  \frac{i^{1 + c_{jk}}}{4}\left(\begin{aligned}
                        &0 - 2iX_{j-1}Z_j Z_{[j+1:k-1]}X_k \\
                        &-2iZ Y_{j-1}Z_j Z_{[j+1:k-1]}Y_k - 0
                    \end{aligned} \right) \\[5pt] 
                    &=-\frac{i^{2 + c_{jk}}}{2} (\AB{XX}{j-1,k} + \AB{YY}{j-1,k}) \\
                    &=\frac{i^{c_{jk}}}{2} (\AB{XX}{j-1,k} + \AB{YY}{j-1,k}) \\
                    &=\frac{i^{c_{j-1,k}}}{2} (\AB{XX}{j-1,k} + \AB{YY}{j-1,k}) \label{eq1}
                \end{align}
            
                \noindent where we have used \cref{lem:c_jk} in the last line for $c_{j-1,k} = c_{jk}$ since $k - j$ is even. Moreover, \cref{eq1} agrees with \cref{def_Pjk} since $k - (j -1) = k - j + 1$ is necessarily odd.
            
                For $k - j$ odd:
                \begin{align}
                    ad_{iXY_{j-1,j}} \cdots ad_{iXY_{k-1,k}}  &\stackrel{IH}{=} \com{iXY_{j-1,j},P_{j,k} }\\
                     &= \frac{i^{1 + c_{jk}}}{4}[X_{j-1}X_{j} + Y_{j-1}Y_{j},\ \AB{XX}{jk} + \AB{YY}{jk}] \\[5pt]
                    &= \frac{i^{1 + c_{jk}}}{4}\left(\begin{aligned}
                        &[X_{j-1}X_{j} ,\ \AB{XX}{jk} ] + [X_{j-1}X_{j},\ \AB{YY}{jk}] \\
                        &+[Y_{j-1}Y_{j},\ \AB{XX}{jk} ] + [Y_{j-1}Y_{j},\  \AB{YY}{jk}]
                    \end{aligned} \right) \\[5pt]
                    &= \frac{i^{1 + c_{jk}}}{4}\left(\begin{aligned}
                        & X_{j-1}[X,X]_j Z_{[j+1:k-1]}X_k + X_{j-1}[X,Y]_j Z_{[j+1:k-1]}Y_k \\
                        &+Y_{j-1}[Y,X]_j Z_{[j+1:k-1]}X_k  + Y_{j-1}[Y,X]_j Z_{[j+1:k-1]}Y_k 
                    \end{aligned} \right) \\[5pt] 
                    &= \frac{i^{1 + c_{jk}}}{4}\left(\begin{aligned}
                        &2iX_{j-1}Z_j Z_{[j+1:k-1]}X_k - 0 \\
                        &0 + 2iZ Y_{j-1}Z_j Z_{[j+1:k-1]}Y_k
                    \end{aligned} \right) \\[5pt] 
                    &=\frac{i^{2 + c_{jk}}}{2}(\AB{XX}{j-1,k} + \AB{YY}{j-1,k}) \\
                    &=\frac{i^{c_{j-1,k}}}{2}P_{j-1,k}
                \end{align}
            
                \noindent where we once again used \cref{lem:c_jk} for $c_{j-1,k} = 2 + c_{jk}$ since $k - j$ is odd.
            
                In either case, we have shown that $ad_{iXY_{j-1,j}} \cdots ad_{iXY_{k-1,k}} = P_{j-1,k}$.
            \end{proof}

            We are now equipped to prove that the $P_{j,k}$ satisfy the skew-symmetric commutation relations \cref{eq:SO1,eq:SO2,eq:SO3}.
            
            \begin{proof}[Proof (\cref{prop:XY_DLA_path})]\label{pf:P_son}
                Due to the definition of $P_{j,k}$ being two cases with respect to the parity of $n$, for each SO relation, we need to show 4 cases: (i) when both terms are even, (ii) when both terms are odd, (iii) when the first is even and the second is odd, and (iv) when the first is odd and the second is even. We only include one case per commutation relationship for brevity; the others follow from almost identical calculation.
                
                \begin{enumerate}
                    \item[SO1:] Assume $k-j$ odd, $\ell - k$ even. We use \cref{lem:ABCD_jkl_coms} freely throughout these calculations.  
                        \begin{align}
                            \left[P_{j,k}, P_{k,\ell} \right] &=\frac{i^{c_{jk} + c_{k,\ell}}}{4} [\AB{XX}{jk} + \AB{YY}{jk}, \AB{XY}{k,\ell} - \AB{YX}{k,\ell}]\\[5pt]
                            &=\frac{i^{c_{jk} + c_{k,\ell}}}{4}\left( \begin{aligned}
                                &[\AB{XX}{jk}, \AB{XY}{k,\ell}] - [\AB{XX}{jk} , \AB{YX}{k,\ell}] \\
                                &+ [ \AB{YY}{jk}, \AB{XY}{k,\ell}] - [\AB{YY}{jk}, \AB{YX}{k,\ell}] 
                            \end{aligned} \right)\\[5pt]
                            &=\frac{i^{c_{jk} + c_{k,\ell}}}{4}\left( \begin{aligned}
                                &0 - X_j Z_{[j+1:k-1]} [X,Y]_k Z_{[k+1:\ell-1]} X_{\ell} \\
                                &+ Y_j Z_{[j+1:k-1]} [Y,X]_k Z_{[k+1:\ell-1]} Y_{\ell} - 0
                            \end{aligned} \right)\\[5pt]
                            &=\frac{i^{c_{jk} + c_{k,\ell}}}{4}\left( \begin{aligned}
                                &- X_j Z_{[j+1:k-1]} (2iZ)_k Z_{[k+1:\ell-1]} X_{\ell} \\
                                &+ Y_j Z_{[j+1:k-1]} (-2iZ)_k Z_{[k+1:\ell-1]} Y_{\ell} 
                            \end{aligned} \right)\\[5pt]
                            &=\frac{i^{3 + c_{jk} + c_{k,\ell}}}{2}(\AB{XX}{j,\ell} + \AB{YY}{j,\ell})
                        \end{align}
                    
                        \noindent Using \cref{lem:c_jk}: 
                        \begin{equation}
                            3 + c_{jk} + c_{k,\ell} = 3 + (k - j +1) + (\ell - k) = 4 + \ell - j \equiv \ell - j \bmod 4 = c_{j, \ell}
                        \end{equation}
                        
                        \noindent since $\ell - j = \overbrace{(\ell - k)}^{\text{even}} - \overbrace{(k - j)}^{\text{odd}}$ is odd. Also, since this is the exponent on $i$, all that matters is behavior $\mod 4$. Thus, $\left[P_{j,k}, P_{k,\ell} \right] = P_{j,\ell}$ as desired in SO1.
            
                    \item[SO2:] Assume $\ell-j$ odd, $k - j$ even. Then 
                        \begin{align}
                            \left[P_{j,\ell}, P_{j, k} \right] &= \frac{i^{c_{j,\ell} + c_{j, k}}}{4}[\AB{XX}{j,\ell} + \AB{YY}{j,\ell}, \AB{XY}{jk} - \AB{YX}{jk}]\\[5pt]
                            &=\frac{i^{c_{j,\ell} + c_{j, k}}}{4}\left( \begin{aligned}
                                &[\AB{XX}{j,\ell}, \AB{XY}{jk}] - [\AB{XX}{j,\ell} , \AB{YX}{jk}] \\
                                &+ [ \AB{YY}{j,\ell}, \AB{XY}{jk}] - [\AB{YY}{j,\ell}, \AB{YX}{jk}] 
                            \end{aligned} \right)\\[5pt]
                            &=\frac{i^{c_{j,\ell} + c_{j, k}}}{4} \bigP{
                            \begin{aligned}
                                &\com{X_jZ_k, X_jY_k}Z_{[k+1:\ell-1]}X_\ell - \com{X_jZ_k, Y_jX_k}Z_{[k+1:\ell-1]}X_\ell \\
                                &+\com{Y_jZ_k, X_jY_k}Z_{[k+1:\ell-1]}Y_\ell - \com{Y_jZ_k, Y_jX_k}Z_{[k+1:\ell-1]}Y_\ell
                            \end{aligned}
                            }\\[5pt]
                            &=\frac{i^{c_{j,\ell} + c_{j, k}}}{4} \bigP{
                            \begin{aligned}
                                &-2iX_kZ_{[k+1:\ell-1]}X_\ell - 0 \\
                                &+0 - 2iY_kZ_{[k+1:\ell-1]}Y_\ell
                            \end{aligned}
                            }\\[5pt]
                            &= \frac{i^{3 + c_{j,\ell} + c_{j, k}}}{2}(\AB{XX}{k,\ell} + \AB{YY}{k,\ell})
                        \end{align}
    
                        Observe that
                        \begin{align}
                            3 + c_{j,\ell} + c_{j, k} = 3 + (\ell-j) + (k - j + 1) &\equiv \ell + k - 2j \bmod 4 \\
                            &\equiv \ell - k \bmod 4
                        \end{align}
    
                        \noindent where we subtract $2(k - j)$ in the last line which we can do since $k-j$ is even (and thus $2(k-j) \equiv 0 \bmod 4$). Therefore, $\left[P_{j,\ell}, P_{j, k} \right] = P_{k,\ell}$ as desired in SO2.
            
                    \item[SO2:] Assume $\ell - k$ and $\ell - j$ are odd. Then
                    \begin{align}
                            \left[P_{k,\ell}, P_{j, \ell} \right] &= \frac{i^{c_{k,\ell} + c_{j, \ell}}}{4}[\AB{XX}{k,\ell} + \AB{YY}{k,\ell}, \AB{XX}{j,\ell} + \AB{YY}{j,\ell}]\\[5pt]
                            &=\frac{i^{c_{k,\ell} + c_{j, \ell}}}{4}\bigP{
                            \begin{aligned}
                                &[\AB{XX}{k,\ell}, \AB{XX}{j,\ell}] + [\AB{XX}{k,\ell} , \AB{YY}{j,\ell}] \\
                                &+[\AB{YY}{k,\ell}, \AB{XX}{j,\ell}] + [\AB{YY}{k,\ell}, \AB{YY}{j,\ell}]
                            \end{aligned}
                            }\\[5pt]
                            &=\frac{i^{c_{k,\ell} + c_{j, \ell}}}{4}\bigP{
                            \begin{aligned}
                                &X_jZ_{[j+1:k-1]}(XZ)_k\{X,X\}_k + 0\\
                                &+ 0 + Y_jZ_{[j+1:k-1]}(YZ)_k\{Y,Y\}_k
                            \end{aligned}
                            }\\[5pt]
                            &=\frac{i^{c_{k,\ell} + c_{j, \ell}}}{4}\bigP{X_jZ_{[j+1:k-1]}(-2iY)_k + Y_jZ_{[j+1:k-1]}(2iX)_k}\\[5pt]
                    \end{align}
                \end{enumerate}
            
                Thus, the mapping $P_{j,k} \mapsto a_{j,k}$ is a Lie algebra isomorphism from $\gxyp{n}$ to $\so{n}$.
            \end{proof}

        \subsection{Lie Algebra of the $XY$ Cycle}\label{subsec:dla_xy_cycle}
            The DLA generated by $\Gxyc{n}$ is different depending on the parity of $n$ so we must handle $\gxyce{n}$ and $\gxyco{n}$ separately.
            
            Define operators $\Popp{jk}$ which serve are generated in the cycle DLA that are not present in the path case by considering looping around the opposite direction of $P_{jk}$
            
            \begin{definition}[$\Popp{jk}$]
                For any $1 \leq j < k \leq n$
                \begin{equation}
                    \Popp{jk} :=\label{def_Pjk_minus} \frac{i^{\copp{jk}}}{2}\begin{cases}
                        \ABopp{XX}{jk} + \ABopp{YY}{jk} &(n - k + j \text{ odd})\\ 
                        \ABopp{XY}{jk} - \ABopp{YX}{jk} &(n - k + j \text{ even})\\
                    \end{cases}
                \end{equation}
                
                \noindent with corresponding coefficients as
                \begin{equation}
                    c_{jk}^- := 2\ceil*{\frac{n-k + j}{2}} -1 \label{def_cjk_minus}
                \end{equation}
            \end{definition}
            
            \noindent Note that these technically depend on $n$ but we drop this for simplicity. Similar to \cref{lem:c_jk}, the following fact about the $c_{jk}^-$ is helpful when doing calculations.
            \begin{lemma}\label{lem:c_jk_minus}
                \begin{itemize}
                    \item[]
                    \item If $n -k +j$ is even, then $c_{jk}^- = n-k + j -1$. Also, $c_{j-1,k}^- = c_{jk}^-$.
                    \item If $n -k + j$ is odd, then $c_{jk}^- = n-k + j$. Also, $c_{j-1,k}^- = c_{jk}^- - 2$.
                \end{itemize}
            \end{lemma}
            
            \begin{lemma}\label{lem:nested_XY_opposite}
                For $1 \leq j < k \leq n$\begin{equation}\label{XY_ring_nested_opposite_action}
                    ad_{iXY_{j,j-1}} \cdot ad_{iXY_{j-1,j-2}} \cdots ad_{iXY_{k+2,k+1}} ad_{iXY_{k+1,k}} =  \Popp{jk}
                \end{equation}
            \end{lemma}
            
            Some notes:
            \begin{enumerate}
                \item The subscripts in \cref{lem:nested_XY_opposite} is modulo $n$. In particular, we go down from $j$ to 1 then loop around to $n$, continuing down to $k$.
                \item The definition of $\Popp{jk}$ is hiding a bit of complexity. For example, we have a $\ABopp{XX}{jk} + \ABopp{YY}{jk}$ term whenever $n$ and $k-j$ are even \textit{or} when $n$ and $k-j$ are odd.
            \end{enumerate}
            
            \begin{proof}[(Proof of \cref{lem:nested_XY_opposite})]
                We prove this using induction on $n - k + j$. As a base case, consider when $j=1$ and $k = n$. Then we have
                \begin{equation}
                    XY_{1n}= \frac{i}{2}\left(X_1X_n + Y_1Y_n\right) = \frac{i}{2}\left(\ABopp{XX}{1n} + \ABopp{YY}{1n}\right)  = \Popp{1n}
                \end{equation}
            
                \noindent as desired.
            
                Next, assume that our equation holds for some $n - k + j$ for $j > 1$, ie, is true for the $(j,k)$ case. Consider the $(j+1,k)$ case. 
                
                Assume that $n - k + j$ is even.
                \begin{align}
                    ad_{iXY_{j+1,j}} ad_{iXY_{j,j-1}} \cdots ad_{iXY_{k+1,k}} &\stackrel{IH}{=} \frac{i^{\copp{jk}}}{2}ad_{iXY_{j+1,j}}\left(\ABopp{XY}{jk} - \ABopp{YX}{jk} \right) \\
                    &= \frac{i^{1 + \copp{jk}}}{4}\left[XY_{j+1,j}, \ABopp{XY}{jk} - \ABopp{YX}{jk} \right] \\[5pt]
                    &= \frac{i^{1 + \copp{jk}}}{4}\left(
                    \begin{aligned}
                        &[X_jX_{j+1},\ \ABopp{XY}{jk} ] - [X_{j}X_{j+1} ,\ \ABopp{YX}{jk}] \\
                        &+[Y_{j}Y_{j+1},\ \ABopp{XY}{jk} ] - [ Y_{j}Y_{j+1},\  \ABopp{YX}{jk}]
                    \end{aligned}
                    \right)\\[5pt]
                    &=\frac{i^{1 + \copp{jk}}}{4}Z_{[1:j-1]}\left(
                    \begin{aligned}
                        &\com{X,X}_jX_{j+1}Y_k - \com{X,Y}_jX_{j+1}X_k \\
                        &+\com{Y,X}_jY_{j+1}Y_k - \com{Y,Y}_jY_{j+1}X_k
                    \end{aligned}
                    \right)Z_{[k+1:n]}\\[5pt]
                    &=\frac{i^{1 + \copp{jk}}}{4}Z_{[1:j-1]}\left(
                    \begin{aligned}
                        & 0- 2iZ_jX_{j+1}X_k \\
                        &-2iZ_jY_{j+1}Y_k - 0
                    \end{aligned}
                    \right)Z_{[k+1:n]}\\[5pt]
                    &=\frac{i^{3 + \copp{jk}}}{2}\bigP{\ABopp{XX}{j+1,k} + \ABopp{YY}{j+1,k}}
                \end{align}

                Observe that 
                \begin{align}
                    3 + \copp{jk} = 3 + (n-k + j -1) = n -k + j + 2 = \copp{j+1,k}
                \end{align}
    
                \noindent since $n - (k+1) + j$ is odd. So, $ad_{iXY_{j+1,j}} ad_{iXY_{j,j-1}} \cdots ad_{iXY_{k+1,k}} = \Popp{j+1,k}$. 
            
                A similar calculation shows for $n - k + j$ odd.
            \end{proof}

            Next, we are interested in what happens when we nest commutations fully around the cycle, beginning and ending at the same node. In the DLA process of generating $\gxyp{n}$, the ``last'' element is given by
            \begin{equation}
                P_{1n} =  ad_{iXY_{12}} \left( P_{2,n} \right)
            \end{equation}
    
            \noindent Now that we have $XY_{n,1}$ in $\Gxyc{n}$, what happens when we take $ad_{X_{n1}}(P_{1n})$? We show that for even $n$, this is 0. That is, during the DLA process, once we wrap around the full (even) cycle, there are no additional terms. However, for odd $n$ odd, this is non-zero. This in turn leads to more complex DLA for the odd case. By shifting indices, we see that this behavior extends to $[iXY_{j,j+1}, [\dots [iXY_{j-2,j-1}, iXY_{j-1,j}]\dots]]$ for any $j \in [n]$. This leads to the following definition.
        
            \begin{definition}[$\Dopp{j,j+1;n}$]
                Let $Z_{\overline{j}}$ be the operator that applies a Pauli-Z to every qubit \textit{except} at qubit $j$ where it is just $I$. 
                \begin{equation}
                    \Dopp{j,j+1;n} := \frac{i^{d_n}}{2}(\Zbar{j} - \Zbar{j+1})
                \end{equation}
        
                \noindent for $d_n := 2\floor{\frac{n}{2}} + 1$
            \end{definition}
        
            \noindent Note that the subscript is modular $n$. We drop the $n$ in the subscript when it is fixed.  The exponent is used so that we have $+i$ for $n \equiv 0 \pmod 4$ or $n \equiv 1 \pmod 4$ and $-i$ otherwise. This is due to the fact that we are taking nested commutations of $XY$ terms which have a $i$ prefactor; thus, when we take an additional 2 commutations, we pick up a $i^2 = -1$ term. For example, see the difference between two different $\Dopp{1,2}$ for differing $n$ (ignoring 1/2):
            \begin{align}
                \Dopp{1,2;3} &= [iXY_{1,2}, [iXY_{2,3}, iXY_{3,1}]] \\
                \Dopp{1,2;5} &= [iXY_{1,2}, [iXY_{2,3}, [iXY_{3,4}, [iXY_{4,5}, iXY_{5,1}]]]]
            \end{align}
            
            \begin{lemma}\label{lem:XY_fully_nested_coms}
            \begin{equation}\label{eq:XY_fully_nested_coms}
                    \frac{1}{2}ad_{iX_{n1}}(P_{1n}) = \begin{cases}
                        0 &n \text{ even}\\
                        \Dopp{j,j+1;n}  &n \text{ odd}
                    \end{cases}
                \end{equation}
            \end{lemma}
            
            \noindent This lemma is proven separately in even and odd $n$ sections.

            Lastly, we define $\hat{Z}$:
            \begin{equation}
                \hat{Z} := (-1)^{\floor{\frac{n}{2}}} \Zn \label{Zhat}
            \end{equation}
        
            \noindent $\hat{Z}$ as well as the operators $\frac{1}{2}(I + \hat{Z})$ and $\frac{1}{2}(I - \hat{Z})$ are very useful throughout the remaining Lie algebra decompositions. We list some facts here before proceeding.
        
            \begin{lemma}\label{lem:I_plus_Zhat}
                Let $A,B \in \bigB{P_{jk}}_{j,k}$ and $C,D \in \bigB{\Popp{jk}}_{j,k}$.
                \begin{enumerate}
                    \item $\hat{Z}$ is an involution and $ \frac{1}{2}(I + \hat{Z}),\frac{1}{2}(I - \hat{Z})$ are orthogonal projections.
                    \item $  [\hat{Z}, \bigB{A,C}] = \com{\frac{1}{2}(I \pm \hat{Z}), \bigB{A,C}} = \bigB{0} $.
                    \item
                    \begin{align}
                        \com{\frac{1}{2}(I + \hat{Z})\bigB{A,C}, \bigB{B,D}} &= \com{\bigB{A,C}, \frac{1}{2}(I + \hat{Z})\bigB{B,D}} 
                        =  \frac{1}{2}(I + \hat{Z})\com{\bigB{A,C},\bigB{B,D}}\\
                        \com{\frac{1}{2}(I - \hat{Z})\bigB{A,C}, \bigB{B,D}} &= \com{\bigB{A,C}, \frac{1}{2}(I - \hat{Z})\bigB{B,D}} 
                        =  \frac{1}{2}(I - \hat{Z})\com{\bigB{A,C},\bigB{B,D}}
                    \end{align}
        
                    \item 
                    \begin{align}
                        \com{\frac{1}{2}(I + \hat{Z})\bigB{A,C}, \frac{1}{2}(I + \hat{Z})\bigB{B,D}} &= \frac{1}{2}(I + \hat{Z})\com{\bigB{A,C}, \bigB{B,D}} \\
                        \com{\frac{1}{2}(I - \hat{Z})\bigB{A,C}, \frac{1}{2}(I - \hat{Z})\bigB{B,D}} &= \frac{1}{2}(I - \hat{Z})\com{\bigB{A,C}, \bigB{B,D}}
                    \end{align}
                \end{enumerate}
            \end{lemma}
        
             \begin{proof}
                \begin{enumerate}
                    \item[]
                    \item[(1)] These are simple calculations.
                    \item[(2)] First note that $\com{ZZ, \bigB{XX, YY, XY, YX, ZI, IZ}} = 0$. By extension, for any $1 \leq j < k \leq n$, one has $\com{\hat{Z}, \bigB{\AB{AB}{jk}, \ABopp{AB}{jk}}} = 0$ for any $A,B \in \{X,Y\}$. Since $P_{jk}$ and $\Popp{jk}$ are linear combinations of these terms, we have that $\com{\hat{Z}, \bigB{P_{jk}, \Popp{jk}}_{j,k}} = 0$. By extension, $\com{\frac{1}{2}\bigP{ I \pm \hat{Z}}, \bigB{P_{jk}, \Popp{jk}}_{j,k}} = 0$.
                    \item[(3)] This is a direct application of (2).
                    \item[(4)] This is a direct application of (1) and (3).
                \end{enumerate}
            \end{proof}
            
        \subsubsection{Even $n$}\label{subsubsec:even_cycle}
            
            \begin{proposition}\label{prop:XY_DLA_cycle_even}
                $\gxyce{n} \cong \so{n} \oplus \so{n}$.
            \end{proposition}
            
            This proposition is proven in the following steps.
            \begin{enumerate}
                \item Show that the only operators generated by the DLA process are $\bigB{P_{j,k},\Popp{\ell,m}}_{j,k;\ell,m}$. In particular, show that nested commutation around the full cycle starting and ending at any node $j$ results in 0 as seen in \cref{eq:XY_fully_nested_coms} case 1. Thus, no other operators can be constructed by the DLA process.
                \item Show that commutation $\bigB{P_{j,k},\Popp{\ell,m}}_{j,k;\ell,m}$ is closed under commutation and so a Lie algebra itself.  In particular, this shows that $\gxyce{n} = \text{span}\left(\bigB{P_{j,k},\Popp{\ell,m}}_{j,k;\ell,m}\right)$ as $\gxyce{n}$ is the minimal Lie algebra containing generators $\Gxyc{n}$.
                \item In proving the previous this step, we show that $\hat{Z} \cdot P_{jk} = \Popp{jk}$.\footnote{This is only true for even $n$.} Define the sets $A^+ := \bigB{\frac{1}{2}(I + \hat{Z})P_{j,k}}_{j,k}$ and $A^- := \bigB{\frac{1}{2}(I - \hat{Z})P_{j,k}}_{j,k}$. Notice that $A^+ + A^- =\bigB{P_{j,k},\Popp{\ell,m}}_{j,k;\ell,m}$. Moreover, they have zero overlap in the Lie algebra: $\com{A^+,A^-} = \{0\}$. Lastly, we can show that both $A^+$ and $A^-$ are each isomorphic to $\so{n}$. This allows us to decompose the space into a the direct sum of Lie algebras:
                \begin{equation}
                    \gxyce{n} = \text{span}\left(\bigB{P_{j,k},\Popp{\ell,m}}_{j,k;\ell,m}\right) = \lie{A^+} \oplus \lie{A^-} \cong \so{n} \oplus \so{n}
                 \end{equation}
            \end{enumerate}   
            
            We now prove the statements above.
            
            \begin{proof}[(1)]
                When $n$ is even, $n-1$ is odd, so $P_{1n} = \frac{i^{n-1}}{2} \left(\AB{XX}{1n} + \AB{YY}{1n} \right)$. Then
                \begin{align}
                ad_{iXY_{n,1}}(P_{1n}) &=\frac{i^{n}}{4}[X_1X_n + Y_1Y_n, \AB{XX}{1n} + \AB{YY}{1n}]\\
                &=\frac{1}{4}\bigP{[X_1X_n, \AB{XX}{1n}] + [X_1X_n, \AB{YY}{1n}] + [Y_1Y_n, \AB{XX}{1n} ] + [Y_1Y_n,  \AB{YY}{1n}]}\\
                &= 0
            \end{align}
            
                \noindent Since all 4 terms commute.
            
                Thus, the only elements generated in the DLA process are the $P_{j,k}$ and $\Popp{\ell,m}$.
            \end{proof}
            
            \begin{proof}[(2)]
                We begin by showing that $\hat{Z} \cdot P_{j,k} = \Popp{j,k}$.
                
                For $k-j$ even:
                \begin{align}
                    \hat{Z} \cdot P_{j,k} &= i^{2\frac{n}{2}}Z_1 \cdots Z_n \cdot \frac{i^{c_{j,k}}}{2}(\AB{XY}{jk} - \AB{YX}{jk}) \\
                    &=\frac{i^{n+c_{j,k}}}{2}\left(\begin{aligned}
                        &Z_{[1:j-1]} (ZX)_j(ZZ)_{[j+1:k-1]} (ZY)_kZ_{[j+1:n]} \\
                        &+Z_{[1:j-1]} (ZY)_j(ZZ)_{[j+1:k-1]} (ZX)_kZ_{[j+1:n]}
                    \end{aligned}\right)\\[5pt]
                    &=\frac{i^{n+c_{j,k}}}{2}\left(\begin{aligned}
                        &Z_{[1:j-1]} (iY)_j(-iX)_kZ_{[j+1:n]} \\
                        &-Z_{[1:j-1]} (-iX)_j(iY)_k Z_{[j+1:n]}
                    \end{aligned}\right)\\[5pt]
                    &=\frac{i^{n+c_{j,k} +2}}{2}\left(Z_{[1:j-1]} X_jY_k Z_{[j+1:n]} - Z_{[1:j-1]} Y_jX_k Z_{[j+1:n]}\right)\\
                    &=\frac{i^{n+c_{j,k} +2}}{2}\bigP{\ABopp{XY}{jk} - \ABopp{YX}{jk}}
                \end{align}
                
                \noindent We can now work the exponent to get
                \begin{align}
                    n+c_{j,k} +2 = n + (k - j + 1) + 2  &\equiv n - k + j + 3 \bmod 4\\
                    &\equiv n - k + j - 1 \bmod 4\\
                    &=\copp{jk}
                \end{align}
        
                \noindent where the first equivalence is from $2(k-j) \equiv 0 \mod 4$. Thus,
                \begin{equation}
                    \hat{Z} \cdot P_{j,k} = \frac{i^{\copp{jk}}}{2}\bigP{\ABopp{XY}{jk} - \ABopp{YX}{jk}} = \Popp{jk}\label{Z_times_Pjk}
                \end{equation}
        
                \noindent A similar calculation works for $k-j$ odd.
            
                Since $\hat{Z}$ commutes with both $P_{j,k}$ and $\Popp{\ell,m}$, we have the following:
                \begin{align}
                    [P_{j,k}, P_{\ell,m}^-] &= \hat{Z} \cdot [P_{j,k}, P_{\ell,m}] \label{comm_P_Pminus}\\
                    [P_{j,k}^-, P_{\ell,m}] &= \hat{Z} \cdot [P_{j,k}, P_{\ell,m}] \label{comm_Pmimus_P}\\
                    [P_{j,k}^-, P_{\ell,m}^-] &= [P_{j,k}, P_{\ell,m}] \label{comm_Pmimus_P_minus}
                \end{align}
                
                \noindent In particular, the first two are some $\Popp{a,b}$ depending on the individual indices. Therefore, $\bigB{P_{j,k},\Popp{\ell,m}}_{j,k;\ell,m}$ is closed under commutation and $\text{span}\left(\bigB{P_{j,k},\Popp{\ell,m}}_{j,k;\ell,m}\right)$ is thus a Lie algebra. In particular, it must be $\gxyce{n}$ since this is the smallest Lie algebra containing $\Gxyc{n}$.
            \end{proof}

            \begin{proof}[(3)]\label{pf:Aplus_Aminus_cong_son}
                By \cref{lem:I_plus_Zhat} item 4, commutations in $A^+$ or $A^-$ reduce to commutations of $\{P_{jk}\}_{j,k}$. For example, taking $A^+_{j,k}, A^+_{k,\ell} \in A^+$, we have
                \begin{align}
                    \com{A^+_{j,k}, A^+_{k,\ell}} &= \com{\frac{1}{2}(I + \hat{Z})P_{j,k}, \frac{1}{2}(I + \hat{Z})P_{\ell,m}} \\
                    &= \frac{1}{2}(I + \hat{Z})\com{P_{j,k}, P_{\ell,m}}\\
                    &=\frac{1}{2}(I + \hat{Z})P_{j,\ell}\\
                    &=A^+_{j,\ell}
                \end{align}            
                
                \noindent Thus, they both obey all the skew-symmetric relations \cref{eq:SO1,eq:SO2,eq:SO3} and so $A^+$ and $A^-$ are themselves isomorphic to $\so{n}$.
        
                Lastly, by combining \cref{lem:I_plus_Zhat} items (1) and (3), we see that these sets have no overlap. For example, for any $A^+_{j,k} \in A^+, A^-_{\ell,m} \in A^-$:
                \begin{align}
                    [A^+_{j,k}, A^-_{\ell,m}] &= \com{\frac{1}{2}(I + \hat{Z})P_{j,k}, \frac{1}{2}(I - \hat{Z})P_{\ell,m}}\\
                    &=\frac{1}{2}(I + \hat{Z}) \cdot \frac{1}{2}(I - \hat{Z}) \com{P_{j,k}, P_{\ell,m}}\\
                    &= 0.
                \end{align}
        
                In conclusion, we have that 
                \begin{equation}
                    \gxyce{n} = \text{span}\left(\bigB{P_{j,k},\Popp{\ell,m}}_{j,k;\ell,m}\right) = \lie{A^+} \oplus \lie{A^-} \cong \so{n} \oplus \so{n}
                 \end{equation}
            \end{proof}
        
        \subsubsection{Odd $n$}\label{subsubsec:odd_cycle}
        
            \begin{proposition}\label{prop:XY_DLA_cycle_odd}
                $\gxyco{n} \cong \su{n}$.
            \end{proposition}
        
            This proposition is proven in the following steps.
            \begin{enumerate}
                \item Show that the operators generated by the DLA process are $\bigB{P_{j,k},\Popp{\ell,m},\Dopp{a,a+1}}_{j,k;\ell,m;a}$. In particular, show case 2 of \cref{eq:XY_fully_nested_coms}.
                \item Show that commutation $\bigB{P_{j,k},\Popp{\ell,m},\Dopp{a,a+1}}_{j,k;\ell,m;a}$ is closed under commutation and so a Lie algebra itself.  In particular, this shows that $\gxyco{n} = \text{span}\left(\bigB{P_{j,k},\Popp{\ell,m},\Dopp{a,a+1}}_{j,k;\ell,m;a}\right)$ as $\gxyco{n}$ is the minimal Lie algebra containing generators $\Gxyc{n}$. In showing this, we also show that these operators obey the SU relations with $\bigB{P_{j,k} : 1 \leq j < k \leq n},\bigB{\Popp{}{j,k} : 1 \leq j < k \leq n}$, and $\bigB{\Dopp{j,j+1} : 1 \leq j < n}$ acting as the skew-symmetric, symmetric, and diagonal operators, respectively.
            \end{enumerate}
        
            \begin{proof}[(1)]
                When $n$ is odd, then $n-1$ is even and so $P_{1n} = \frac{i^n}{2}\bigP{\AB{XY}{1n} - \AB{YX}{1n}}$. So:
                \begin{align}
                    \frac{1}{2}ad_{iXY_{n,1}}(P_{1n}) &=\frac{i^{n+1}}{8}[X_1X_n + Y_1Y_n, \AB{XY}{1n} - \AB{YX}{1n}]\\
                    &=\frac{i^{n+1}}{8}\bigP{[X_1X_n, \AB{XY}{1n}] - [X_1X_n, \AB{YX}{1n}] + [Y_1Y_n, \AB{XY}{1n} ] - [Y_1Y_n,  \AB{YX}{1n}]} \label{dn_odd_eq} \\
                    &=\frac{i^{n+1}}{8}\bigP{2iZ_{\overline{1}} - 2iZ_{\overline{n}} - 2iZ_{\overline{n}} + 2iZ_{\overline{1}}}\\
                    &=\frac{i^{n+2}}{2}\bigP{Z_{\overline{1}} - Z_{\overline{n}}}\\
                    &= \frac{i^{n}}{2}(Z_{\overline{n}} - Z_{\overline{1}} ) \\
                    &=\Dopp{n,1;n}.
                \end{align}
    
                By re-ordering indices, we can see that this behavior generates all the $\Dopp{a,a+1;n}$ for any $a \in [n]$.
            \end{proof}
    
            We now show that the operators $\bigB{P_{j,k},\Popp{\ell,m},\Dopp{a,a+1}}_{j,k;\ell,m;a}$ obey $\su{n}$ relations. 
        
            \begin{proof}[(2)]\label{pf:P_Pminus_Dminus_sun}            
                \begin{enumerate}
                    \item[]
                    \item[(SU1-3)] We have already shown that that $\bigB{P_{j,k}}_{j,k}$ satisfy \cref{eq:SO1,eq:SO2,eq:SO3} in $\gxyp{n}$.
                    \item[(SU4)] Let $1 \leq j < k < \ell \leq n$. We consider the case when $n- k + j$ is odd and $n - \ell + k$ is even.
                     \begin{align}
                         \com{\Popp{j,k},\Popp{k,\ell}} &= \frac{i^{c_{jk}^- + c_{k\ell}^-}}{4}\com{\ABopp{XX}{jk} + \ABopp{YY}{jk},\ \ABopp{XY}{k \ell} - \ABopp{YX}{k\ell}} \\[5pt]
                         &=\frac{i^{c_{jk}^- + c_{k\ell}^-}}{4} \bigP{
                         \begin{aligned}
                             &\com{\ABopp{XX}{jk}, \ABopp{XY}{k \ell}} - \com{\ABopp{XX}{jk}, \ABopp{YX}{k\ell}} \\
                             &+\com{\ABopp{YY}{jk},\ \ABopp{XY}{k \ell}} - \com{\ABopp{YY}{jk}, \ABopp{YX}{k\ell}}
                         \end{aligned}
                         }
                     \end{align}
                     Apply \cref{lem:AB_minus_CD_minus_coms} to these four terms.
                     \begin{align}
                         \com{\ABopp{XX}{jk}, \ABopp{XY}{k \ell}} &=\com{X_jX_kZ_\ell,\  Z_jX_kY_\ell}\cdot Z_{[j+1:\ell-1]\setminus \{k\}}\\
                         &=\bigP{(XZ)_j(XX)_k(ZY)_\ell -(ZX)_j(XX)_k(YZ)_\ell}\cdot Z_{[j+1:\ell-1]\setminus \{k\}}\\
                         &=\bigP{(-iY)_j(-iX)_\ell -(iY)_j(iX)_\ell}\cdot Z_{[j+1:\ell-1]\setminus \{k\}}\\
                         &=i^2\bigP{Y_jX_\ell -Y_jX_\ell}\cdot Z_{[j+1:\ell-1]\setminus \{k\}}\\
                         &=0 \\[3pt]
                         \com{\ABopp{XX}{jk}, \ABopp{YX}{k\ell}} &=\com{X_jX_kZ_\ell,\  Z_jY_kX_\ell}\cdot Z_{[j+1:\ell-1]\setminus \{k\}}\\
                        &=\bigP{(XZ)_j(XY)_k(ZX)_\ell -(ZX)_j(YX)_k(XZ)_\ell}\cdot Z_{[j+1:\ell-1]\setminus \{k\}}\\
                        &=\bigP{(-iY)_j(iZ)_k(iY)_\ell -(iY)_j(-iZ)_k(-iY)_\ell}\cdot Z_{[j+1:\ell-1]\setminus \{k\}}\\
                        &=-i^3\bigP{Y_jZ_kY_\ell + Y_jZ_kY_\ell}\cdot Z_{[j+1:\ell-1]\setminus \{k\}}\\
                        &=2i\AB{YY}{j,\ell}\\[3pt]
                          \com{\ABopp{YY}{jk}, \ABopp{XY}{k\ell}} &=\com{Y_jY_kZ_\ell,\  Z_jX_kY_\ell}\cdot Z_{[j+1:\ell-1]\setminus \{k\}}\\
                        &=\bigP{(YZ)_j(YX)_k(ZY)_\ell -(ZY)_j(XY)_k(YZ)_\ell}\cdot Z_{[j+1:\ell-1]\setminus \{k\}}\\
                        &=\bigP{(iX)_j(-iZ)_k(-iX)_\ell -(-iX)_j(iZ)_k(iX)_\ell}\cdot Z_{[j+1:\ell-1]\setminus \{k\}}\\
                        &=i^3\bigP{X_jZ_kX_\ell + X_jZ_kX_\ell}\cdot Z_{[j+1:\ell-1]\setminus \{k\}}\\
                        &=-2i\AB{XX}{j,\ell}\\[3pt]
                         \com{\ABopp{YY}{jk}, \ABopp{YX}{k \ell}} &=\com{Y_jY_kZ_\ell,\  Z_jY_kX_\ell}\cdot Z_{[j+1:\ell-1]\setminus \{k\}}\\
                         &=\bigP{(YZ)_j(YY)_k(ZX)_\ell -(ZY)_j(YY)_k(XZ)_\ell}\cdot Z_{[j+1:\ell-1]\setminus \{k\}}\\
                         &=\bigP{(iX)_j(iY)_\ell -(-iX)_j(-iY)_\ell}\cdot Z_{[j+1:\ell-1]\setminus \{k\}}\\
                         &=i^2\bigP{X_jY_\ell -X_jY_\ell}\cdot Z_{[j+1:\ell-1]\setminus \{k\}}\\
                         &=0
                     \end{align}
                
                    Thus, 
                    \begin{equation}
                        \com{\Popp{j,k},\Popp{k,\ell}} = -\frac{i^{c_{jk}^- + c_{k\ell}^- + 1}}{2}\bigP{\AB{XX}{j,\ell} + \AB{YY}{j,\ell}}
                    \end{equation}
            
                    Lastly, handle the exponent. $n-k+j$ is odd, so $\copp{j,k} = n - k + j$ and $n - \ell +k$ is even, so $\copp{k,\ell} = n - \ell + j  - 1$.
                    \begin{align}
                        c_{jk}^- + c_{k\ell}^- + 1 = (n-k + j) + (n - \ell + k -1) + 1 = 2n + j - \ell \equiv 2 + j - \ell \bmod 4 %
                    \end{align}
                    
                    \noindent where we have used $2n \equiv 2 \bmod $ since $n$ is odd. Since $n$ and $n- k + j$ are odd, $k - j$ is even. Conversely, $\ell - k$ is odd. In turn, $\ell - j = (\ell - k) + (k - j)$ is odd. Therefore, $2(\ell-j) \equiv 2 \bmod 4$. This implies that $2 + j - \ell \equiv \ell - j \bmod 4$ which is $ c_{j,\ell} $ by \cref{lem:c_jk}. Putting this all together: $\com{\Popp{j,k},\Popp{k,\ell}} = P_{j,\ell}$.
    
                    \item[(SU5,6)] Follow similarly as SU5, using \cref{lem:AB_minus_CD_minus_coms} to handle the individual terms and showing that the exponent matches correctly.               
    
                    \item[(SU7)] Assume $k-j$ odd (and so $n-k+j$ even).
                        \begin{align}
                            \com{P_{j,k}, \Popp{j,k}} &= \frac{i^{c_{jk} + \copp{jk}}}{4}\com{\AB{XX}{jk} + \AB{YY}{jk},\ \ABopp{XY}{jk} - \ABopp{YX}{jk}}\\[3pt]
                            &=\frac{i^{c_{jk} + \copp{jk}}}{4}\bigP{
                            \begin{aligned}
                                &\com{\AB{XX}{jk}, \ABopp{XY}{jk}} - \com{\AB{XX}{jk},  \ABopp{YX}{jk}}\\
                                &+\com{\AB{YY}{jk}, \ABopp{XY}{jk}} - \com{\AB{YY}{jk}, \ABopp{YX}{jk}}
                            \end{aligned}
                            }\label{eq:su8}
                        \end{align}
                
                        Apply \cref{lem:AB_CD_minus_com} to these four terms.
                         \begin{align}
                             \com{\AB{XX}{jk}, \ABopp{XY}{jk}} &= \com{\AB{XX}{jk}, \AB{XY}{jk}} \cdot Z_{[n]\setminus \{j,k\}} \\
                             &= \bigP{(XX)_j(XY)_k - (XX)_j(YX)_k}\cdot Z_{[n]\setminus \{j,k\}} \\
                             &= 2i\Zbar{j}\\[3pt]
                             \com{\AB{XX}{jk},  \ABopp{YX}{jk}} &= \com{\AB{XX}{jk}, \AB{YX}{jk}} \cdot Z_{[n]\setminus \{j,k\}} \\
                             &= \bigP{(XY)_j(XX)_k - (YX)_j(XX)_k}\cdot Z_{[n]\setminus \{j,k\}} \\
                             &= 2i\Zbar{k}\\[3pt]
                             \com{\AB{YY}{jk}, \ABopp{XY}{jk}} &= \com{\AB{YY}{jk}, \AB{XY}{jk}} \cdot Z_{[n]\setminus \{j,k\}} \\
                             &= \bigP{(YX)_j(YY)_k - (XY)_j(YY)_k}\cdot Z_{[n]\setminus \{j,k\}} \\
                             &= -2i\Zbar{k}\\[3pt]
                             \com{\AB{YY}{jk}, \ABopp{YX}{jk}} &= \com{\AB{YY}{jk}, \AB{YX}{jk}} \cdot Z_{[n]\setminus \{j,k\}} \\
                             &= \bigP{(YY)_j(YX)_k - (YY)_j(XY)_k}\cdot Z_{[n]\setminus \{j,k\}} \\
                             &= -2i\Zbar{j}\\[3pt]
                         \end{align}
                    
                         \noindent Plugging this into \cref{eq:su8} we get
                         \begin{equation}
                             \com{P_{jk}, \Popp{j,k}} = i^{c_{jk} + \copp{jk} + 1}\bigP{\Zbar{j} - \Zbar{k}}
                         \end{equation}
                    
                        \noindent Handle the exponent on the coefficient:
                        \begin{align}
                            c_{jk} + \copp{jk} + 1 = (k-j) +(n-k+j - 1) + 1 = n = 2\floor{\frac{n}{2}} + 1 = d_n
                        \end{align}
                    
                        \noindent since $n$ is odd. Thus, $\com{P_{jk}, \Popp{j,k}} = 2 \Dopp{j,k}$.
    
                    \item[(SU8,9)] Follow similarly as SU7.
    
                    \item[(SU10)] Assume $\ell - k$ is even.
                        \begin{align}
                            \com{\Dopp{j,k}, P_{k,\ell}} &= \frac{i^{d_n + c_{k,\ell}}}{4}\com{\Zbar{j} - \Zbar{k}, \AB{XY}{k,\ell} - \AB{YX}{k,\ell}}\\
                            &=\frac{i^{d_n + c_{k,\ell}}} {4}\bigP{\begin{aligned}
                                &\com{\Zbar{j}, \AB{XY}{k,\ell}} - \com{\Zbar{j},  \AB{YX}{k,\ell}}\\
                                &-\com{\Zbar{k}, \AB{XY}{k,\ell}} + \com{\Zbar{k}, \AB{YX}{k,\ell}}
                            \end{aligned}}\\
                            &=\frac{i^{d_n + c_{k,\ell}}} {4}\bigP{\begin{aligned}
                                &-Z_{[1:k-1]}X_k \com{Z,Y}_\ell Z_{[\ell+1:n]}\\
                                &+Z_{[1:k-1]}Y_k \com{Z,X}_\ell Z_{[\ell+1:n]}
                            \end{aligned}}\\
                            &=\frac{i^{d_n + c_{k,\ell}}} {4}\bigP{\begin{aligned}
                                &-2iZ_{[1:k-1]}X_k X_\ell Z_{[\ell+1:n]}\\
                                &-2iZ_{[1:k-1]}Y_k Y_\ell Z_{[\ell+1:n]}
                            \end{aligned}}\\
                            &=\frac{i^{d_n + c_{k,\ell} + 3}}{2}\bigP{\ABopp{XX}{k,\ell} + \ABopp{YY}{k,\ell}}
                        \end{align}
    
                        \noindent where we have used \cref{lem:Zbar_AB} in line 3. Working out the exponent:
                        \begin{align}
                            d_n + c_{k,\ell} + 3 = \bigP{2\floor{\frac{n}{2}} + 1} + \bigP{\ell - k + 1} + 3 &= 2\floor{\frac{n}{2}} + \ell - k + 5 \\
                            &\equiv 2\floor{\frac{n}{2}} + \ell - k + 1 \bmod 4 \\
                            &\equiv 2\floor{\frac{n}{2}}  -\ell + k +1 \bmod 4
                        \end{align}
    
                        \noindent where we have used that $-2(k-\ell) \equiv 0 \bmod 4$ since $k-\ell$ is even. Now, $n$ odd implies that $2\floor{\frac{n}{2}} = n-1$, so
                        \begin{equation}
                            2\floor{\frac{n}{2}}  -\ell + k +1  = n - \ell + k = \copp{k,\ell}
                        \end{equation}
    
                        \noindent Thus, $\com{\Dopp{j,k}, P_{k,\ell}} = \Popp{k,\ell}$.
    
                    \item[(SU11-13)] Follows similarly as SU10, where we use \cref{lem:Zbar_AB} on the individual terms and then work out the algebra on the exponent.

                 \end{enumerate}
            \end{proof}

        \subsection{Lie Algebra of the $XY$ Path with $ R_{Z} $}\label{subsec:dla_xy_z_path}        
            This proposition is proven in the following steps.
            \begin{enumerate}
                \item Show that the term $Z^+$ as defined below is in $\gxypz{n}$ and commutes with all the terms in $\Gxypz{n}$, thus showing that we have a 1-dimensional center $Z(\gxypz{n}) = \text{span}\left(\bigB{Z^+}\right)$.
                \item Show that there are three sets of terms generated in the DLA process: $\bigB{P_{jk}}_{j,k}$, $\bigB{Q_{jk}}_{j,k}$, and $\bigB{D_{a}}_{a}$ where the latter two are defined below.
                \item Show that these sets of operators obey the skew-symmetric, symmetric, and diagonal commutation relations, respectively.
            \end{enumerate}
        
            \begin{lemma}\label{lem:Zplus_center}
                The center of $\gxypz{n}$ is 1-dimensional, containing the element $Z^+ := i\sum_v Z_v$. 
            \end{lemma}
            
            \begin{proof}
                Ignoring coefficients:
                \begin{align}
                    [Z_j, XY_{j,k}] &= \com{Z_j, X_j X_k} + \com{Z_j, Y_j Y_k} = \com{Z,X}_j X_k + \com{Z, Y}_j Y_k = 2iY_jX_k - 2i X_j Y_k \\
                    [Z_k, XY_{j,k}] &= \com{Z_k, X_j X_k} + \com{Z_k, Y_j Y_k}= X_j \com{Z,X}_k + Y_j \com{Z, Y}_k = -2iX_jY_k + 2i Y_j X_k
                \end{align}
        
                \noindent In particular, $[Z_k, XY_{j,k}] = -[Z_j, XY_{j,k}]$. Thus for any generator $XY_{j,j+1} \in \Gxypz{n}$:
                \begin{equation}
                    \com{Z^+, XY_{j,j+1}} = \com{Z_j, XY_{j,j+1}} + \com{Z_{j+1}, XY_{j,j+1}} = 0 
                \end{equation}
        
                Clearly $[Z^+, Z_k]$ for any $k$. If an element commutes with all of the generators, then it commutes with the full the DLA. Thus, $Z(\Gxypz{n}) = \text{span}(Z^+)$.
            \end{proof}
        
            \begin{definition}[$D_{j,k}$]
                For any $1 \leq j < k \leq n$, define 
                \begin{equation}\label{def_Djk}
                    D_{j,k} := \frac{i}{2}\bigP{Z_j - Z_k}
                \end{equation}
            \end{definition}
        
            \noindent Due to telescoping behavior, $D_{j,k} \in \text{span}\left(\bigB{D_{a,a+1} : 1 \leq a < n}\right)$ for any $j,k$. This allows us to just take these $n-1$ operators as our diagonal generators. These elements are in $\gxypz{n}$ since they are a linear combination of generators.
        
            \begin{definition}[$Q_{jk}$]
                For any $1 \leq j < k \leq n$, define 
                \begin{align}
                    Q_{j,k} &:= \frac{1}{2}\com{D_{j,k}, P_{j,k}}
                \end{align}
            \end{definition}
        
            \noindent $Q_{j,k}$ are in $\gxypz{n}$ as they are the result of a commutation of two operators in $\gxypz{n}$. The following is a helpful characterization of $Q_{j,k}$ akin to \cref{def_Pjk}.
        
            \begin{lemma}\label{lem:Qjk_paulis}
                \begin{align}
                    Q_{j,k} &= \frac{i^{b_{jk}}}{2} \begin{cases}\label{eq:Qjk_paulis}
                        \AB{YX}{jk} - \AB{XY}{jk}  &(k - j \text{ odd}) \\ 
                         \AB{XX}{jk} + \AB{YY}{jk}&(k - j \text{ even}) \\
                    \end{cases} \\
                    b_{jk} &:= 2\floor{\frac{k - j}{2}} - 1 \label{bjk_def}
                \end{align}
            \end{lemma}
        
            \noindent In fact, the only difference is that the cases flip depending on the parity. E.g. when $k-j$ is odd, $P_{j,k} \propto \AB{XX}{jk} + \AB{YY}{jk}$ and $Q_{j,k} \propto \AB{XY}{jk} - \AB{YX}{jk}$.
    
            We state a lemma about the $b_{jk}$ terms similar to \cref{lem:c_jk,lem:c_jk_minus}.
            \begin{lemma}\label{lem:b_jk}
                \begin{itemize}
                    \item[]
                    \item If $k - j$ is even, then $b_{jk} = k - j  - 1$. Also, $b_{j-1,k} = b_{jk}$. 
                    \item If $k - j$ is odd, then $b_{jk} = k - j - 2$. Also, $b_{j-1,k} = b_{jk} +2$.
                \end{itemize}
            \end{lemma}
        
            \begin{proof}[Proof (\cref{lem:Qjk_paulis})]    
                Assume that $k - j$ is odd.
                \begin{align}
                    \frac{1}{2}\com{D_{j,k}, P_{j,k}} &= \frac{i^{1 + k - j}}{8}\com{Z_j - Z_k,\ \AB{XX}{jk} + \AB{YY}{jk}} \\
                    &=  \frac{i^{1 + k - j}}{8} \bigP{
                    \begin{aligned}
                        &\com{Z_j,\AB{XX}{jk}} + \com{Z_j , \AB{YY}{jk}} \\
                        &-\com{Z_k, \AB{XX}{jk}} - \com{ Z_k,\AB{YY}{jk}}
                    \end{aligned}
                    }\\
                    &=  \frac{i^{1 + k - j}}{8} \bigP{
                    \begin{aligned}
                        &\com{Z,X}_jX_k + \com{Z,Y}_jY_k \\
                        &-X_j\com{Z, X}_k -Y_j\com{Z, Y}_k
                    \end{aligned}
                    }Z_{[j+1:k-1]}\\
                    &=  \frac{i^{1 + k - j}}{8} \bigP{
                    \begin{aligned}
                        &2iY_jX_k -2iX_jY_k \\
                        &-2iX_jY_k +2iY_jX_k
                    \end{aligned}
                    }Z_{[j+1:k-1]}\\
                    &=  \frac{i^{k - j + 2}}{2}\bigP{\AB{YX}{jk} - \AB{XY}{jk}}
                \end{align}
        
                \noindent Note that $k - j + 2 \equiv k - j - 2 \bmod 4 = \bopp{jk}$ and thus $\frac{1}{2}\com{D_{j,k}, P_{j,k}} = Q_{j,k}$.
    
                The $k-j$ even case follows similarly.
            \end{proof}
        
            Next, we show that the $\bigB{P_{jk}, Q_{\ell m}, D_a}$ satisfy $\su{n}$ relations. As in \cref{pf:P_Pminus_Dminus_sun}, we do not show all relations but an important subset of them.
        
            \begin{proof}[(3)]\label{pf:PQD_sun}
                \begin{itemize}
                    \item[]
                    \item[(SU1-3)] We have already shown that that $\bigB{P_{j,k}}_{j,k}$ satisfy \cref{eq:SO1,eq:SO2,eq:SO3} in $\gxyp{n}$.
                    \item[(SU4)] Assume $k - j$ is odd and $\ell - k$ even. Using \cref{lem:AB_minus_CD_minus_coms},
                        \begin{align}
                            \com{Q_{jk}, Q_{k\ell}} &= \frac{i^{b_{jk} + b_{k\ell}}}{4}\com{\AB{YX}{jk} - \AB{XY}{jk},\ \AB{XX}{k\ell} + \AB{YY}{k\ell}} \\
                            &= \frac{i^{(k - j - 2) + (\ell - k - 1)}}{4} \bigP{
                            \begin{aligned}
                                &\com{\AB{YX}{jk} , \AB{XX}{k\ell} } + \com{\AB{YX}{jk},  \AB{YY}{k\ell}} \\
                                &-\com{\AB{XY}{jk}, \AB{XX}{k\ell}} - \com{\AB{XY}{jk}, \AB{YY}{k\ell}}
                            \end{aligned}
                            }\\
                            &= \frac{i^{\ell -  j - 3}}{4} \bigP{
                            \begin{aligned}
                                & Y_j  [X,Y]_k  Y_{\ell} \\
                                &-X_j [Y,X]_k  X_{\ell}
                            \end{aligned}
                            }Z_{[j+1:\ell-1]\setminus k}\\
                            &= \frac{i^{\ell -  j - 3}}{4} \bigP{2iY_jZ_k Y_{\ell} +2iX_j Z_kX_{\ell}} Z_{[j+1:\ell-1]\setminus k}\\
                            &=\frac{i^{\ell -  j - 2}}{2}\bigP{\AB{XX}{j\ell} + \AB{YY}{j\ell}}
                            \\
                            &=-\frac{i^{\ell -  j}}{2}\bigP{\AB{XX}{j\ell} + \AB{YY}{j\ell}}
                            \\
                            &=-P_{j,\ell}
                        \end{align}
    
                        \noindent Using \cref{lem:c_jk} on the exponent.
                        
                    \item[(SU5,6)] Follow similarly as SU5, using \cref{lem:AB_minus_CD_minus_coms} to handle the individual terms and showing that the exponent matches correctly.
    
                    \item[(SU7)] Assume $k-j$ odd. Once again using \cref{lem:AB_minus_CD_minus_coms} for the intermediate calculations,
                        \begin{align}
                            \com{P_{jk}, Q_{jk}} &= \frac{i^{c_{jk} + b_{jk}}}{4}\com{\AB{XX}{jk} + \AB{YY}{jk}, \AB{YX}{jk} - \AB{XY}{jk}} \\
                            &=\frac{i^{c_{jk} + b_{jk}}}{4}\bigP{\begin{aligned}
                                &\com{\AB{XX}{jk}, \AB{YX}{jk}} - \com{\AB{XX}{jk} ,  \AB{XY}{jk}}\\
                                &+\com{\AB{YY}{jk}, \AB{YX}{jk}} - \com{\AB{YY}{jk},  \AB{XY}{jk}}
                            \end{aligned}}\\
                            &=\frac{i^{c_{jk} + b_{jk}}}{4}\bigP{\begin{aligned}
                                &2iZ_j - 2iZ_k\\
                                &-2iZ_k + 2iZ_j
                            \end{aligned}}\\
                            &=i^{c_{jk} + b_{jk} + 1}\bigP{Z_j - Z_k}
                        \end{align}
    
                        \noindent Using \cref{lem:c_jk,lem:b_jk} and that $2(k-j) \equiv 2 \bmod 4$:
                        \begin{equation}
                            c_{jk} + b_{jk} + 1 = (k - j) + (k - j - 2) + 1 = 2(k-j) - 1 \equiv 1 \bmod 4
                        \end{equation}
    
                        \noindent Thus, $\com{P_{jk}, Q_{jk}} = 2 D_{j,k}$.
    
                    \item[(SU8,9)] Follow similarly as SU7.

                    \item[(SU10)] Assume $\ell - k$ is even. 
                        \begin{align}
                            \com{D_{j,k}, P_{k,\ell}} &= \frac{i^{1 + c_{k,\ell}}}{4}\com{Z_j - Z_{k}, \AB{XY}{k,\ell} - \AB{YX}{k,\ell}}\\
                            &= -\frac{i^{1 + c_{k,\ell}}}{4}\bigP{\com{Z_{k}, \AB{XY}{k,\ell}} -\com{Z_{k}, \AB{YX}{k,\ell}}}\\
                            &= -\frac{i^{1 + c_{k,\ell}}}{4}\bigP{2iY_{k} Z_{[j+2:\ell-1]}Y_\ell +2iX_{k} Z_{[j+2:\ell-1]}X_\ell}\label{eq:su16_pqd}\\
                            &= -\frac{i^{2 + c_{k,\ell}}}{2}\bigP{\AB{XX}{k,\ell} + \AB{YY}{k,\ell} }
                        \end{align}
        
                        \noindent where we have used \cref{lem:com_Z_AB} in \cref{eq:su16_pqd}. Working out the exponent:
                        \begin{align}
                            2 + c_{k,\ell} = 2 + (\ell - k + 1) &= \ell - k + 3 \\
                            &\equiv \ell - k - 1 \bmod 4 \\
                            &= b_{k,\ell}
                        \end{align}
    
                        \noindent Thus $\com{D_{j,k}, P_{k,\ell}} = -Q_{k, \ell}$
    
                    \item[(SU11-13)] Follows similarly as SU10.

                \end{itemize}
            \end{proof}

        \subsection{Lie Algebra of the $XY$ Cycle with $ R_{Z} $}\label{subsec:dla_xy_z_cycle}
        
            The last polynomial-sized DLA we discuss in this manuscript is the $XY$ cycle with $R_Z$ rotations. Define the generator set $\Gxycz{n} := \Gxyc{n} \cup i\Gz{n}$ which induces the DLA $\gxycz{n}$.
        
            \begin{proposition}\label{prop:XY_Z_DLA_cycle}
                $\gxycz{n} \cong \u{1} \oplus \su{n} \oplus \su{n}$.
            \end{proposition}
        
            This proposition is proven in the following steps.
        
            \begin{enumerate}
                \item $\gxycz{n}$ has the same center $\Span{Z^+}$ for the same reasons as in \cref{lem:Zplus_center}.
                \item Every term we have seen thus far is in this DLA: $P_{jk}, \Popp{jk}, D_{jk}, \Dopp{jk}, Q_{jk}$, as well as $\Qopp{jk}$ as described below.
                \item Similar to \cref{prop:XY_DLA_cycle_even}, we need to split our space up using the $\hat{Z}$ operator. Define the sets
                \begin{equation}\label{eq:B_plus_minus}
                    B^\pm := \frac{1}{2}(I \pm \hat{Z}) \bigB{P_{jk}, Q_{\ell m}, D_{a} : 1 \leq j < k \leq n; 1 \leq \ell < m \leq n; 1 \leq a < n }
                \end{equation}
    
                First, we show that $\Span{B^+} + \Span{B^-} + \Span{Z^+} = \gxycz{n}$ and that $[B^+, B^-] = 0$. From this, we have  $\gxycz{n} = \lie{B^+} \oplus \lie{B^-} \oplus \u{1}$. We then show that both $B^+$ and $B^-$ are isomorphic to $\su{n}$. Similar to the XY-cycle cases, we have slightly different behavior depending on the parity of $n$ though this does not change the end result.
            \end{enumerate}

            Now that we have closed boundary conditions, we have the $\bigB{\Popp{jk}}$ terms. This results in the last set of operators that can be generated.
        
            \begin{definition}[$\Qopp{jk}$]
                For any $1 \leq j < k \leq n$, define 
                \begin{align}
                    \Qopp{jk} &:= \frac{1}{2}\com{D_{jk}, \Popp{jk}}
                \end{align}
            \end{definition}
        
            Similar to the previous scenarios, we can write $\Qopp{jk}$ in the Pauli basis as follows.
        
            \begin{lemma}\label{lem:Qjk_minus_paulis}
                \begin{align}
                    \Qopp{jk} &= \frac{i^{\bopp{jk}}}{2} \begin{cases}\label{eq:Qjk_minus_paulis}
                        \ABopp{XY}{jk} - \ABopp{YX}{jk}  &(n - k + j \text{ odd}) \\ 
                         \ABopp{XX}{jk} + \ABopp{YY}{jk}&(n - k + j \text{ even}) \\
                    \end{cases} \\
                    \bopp{jk} &:= 2\floor{\frac{n - k + j}{2}} + 1 \label{b_opp_jk_def}
                \end{align}
            \end{lemma}
        
            \noindent Just the terms in $P_{jk}$ and $Q_{jk}$ are flipped based on parity of $k - j$, the same behavior occurs with $\Popp{jk}$ and $\Qopp{jk}$.

            We state a lemma about the $\bopp{jk}$ terms similar to previous lemmas.
            \begin{lemma}\label{lem:b_jk_minus}
                \begin{itemize}
                    \item[]
                    \item If $n -k +j$ is even, then $\bopp{j,k} = n-k + j +1$. Also, $\bopp{j,k} = \bopp{j+1,k}$.
                    \item If $n -k + j$ is odd, then $\bopp{jk} = n-k + j$. Also, $\bopp{jk} = \bopp{j+1,k} - 2$.
                \end{itemize}
            \end{lemma}
    
            \begin{proof}[Proof of \cref{lem:Qjk_minus_paulis}]
                Assume that $k-j$ is odd and that $n - k + j$ is even.
                \begin{align}
                    \frac{1}{2}\com{D_{jk}, \Popp{jk}} &= \frac{i^{1 + \copp{jk}}}{8}\com{Z_j - Z_k,\ \ABopp{XY}{jk} - \ABopp{YX}{jk}} \\
                    &=\frac{i^{1 + \copp{jk}}}{8}\bigP{\begin{aligned}
                        &\com{Z_j , \ABopp{XY}{jk}} - \com{Z_j, \ABopp{YX}{jk}}\\
                        &-\com{Z_k, \ABopp{XY}{jk}} + \com{Z_k, \ABopp{YX}{jk}}
                    \end{aligned}}\\
                    &=\frac{i^{1 + \copp{jk}}}{8}Z_{[1:j-1]}\bigP{\begin{aligned}
                        &[Z,X]_j Y_k - [Z,Y]_j X_k\\
                        &-X_j [Z,Y]_k + Y_j [Z,X]_k
                    \end{aligned}}Z_{[k+1:n]}\\
                    &=\frac{i^{1 + \copp{jk}}}{8}Z_{[1:j-1]}\bigP{\begin{aligned}
                        &2iY_j Y_k +2iX_j X_k\\
                        &+2iX_j X_k + 2iY_j Y_k
                    \end{aligned}}Z_{[k+1:n]}\\
                    &=\frac{i^{2 + \copp{jk}}}{2}\bigP{\ABopp{XX}{jk} + \ABopp{YY}{jk}}
                \end{align}
    
                Now, $2 + \copp{jk} = 2 + (n-k + j -1) = n - k + j +1 = \bopp{jk}$ and so we have shown that $\frac{1}{2}\com{D_{jk}, \Popp{jk}} = \frac{i^{\bopp{jk}}}{2}\bigP{\ABopp{XX}{jk} + \ABopp{YY}{jk}}$. The other cases follow similarly.
            \end{proof}
        
            We have the following proposition which extends the cases in \cref{lem:I_plus_Zhat}.
            \begin{proposition}\label{prop:Z_hat_commutes_D_P_and_Q}
                Let $A,B \in \bigB{Q_{jk}}_{j,k}$ and $C,D \in \bigB{\Qopp{jk}}_{j,k}$.
                \begin{enumerate}
                    \item $  [\hat{Z}, \bigB{A,C}] = \com{\frac{1}{2}(I \pm \hat{Z}), \bigB{A,C}} = \bigB{0} $.
                    \item 
                    \begin{align}
                        \com{\frac{1}{2}(I + \hat{Z})\bigB{A,C}, \bigB{B,D}} &= \com{\bigB{A,C}, \frac{1}{2}(I + \hat{Z})\bigB{B,D}} 
                        =  \frac{1}{2}(I + \hat{Z})\com{\bigB{A,C},\bigB{B,D}}\\
                        \com{\frac{1}{2}(I - \hat{Z})\bigB{A,C}, \bigB{B,D}} &= \com{\bigB{A,C}, \frac{1}{2}(I - \hat{Z})\bigB{B,D}} 
                        =  \frac{1}{2}(I - \hat{Z})\com{\bigB{A,C},\bigB{B,D}}
                    \end{align}
        
                    \item 
                    \begin{align}
                        \com{\frac{1}{2}(I + \hat{Z})\bigB{A,C}, \frac{1}{2}(I + \hat{Z})\bigB{B,D}} &= \frac{1}{2}(I + \hat{Z})\com{\bigB{A,C}, \bigB{B,D}} \\
                        \com{\frac{1}{2}(I - \hat{Z})\bigB{A,C}, \frac{1}{2}(I - \hat{Z})\bigB{B,D}} &= \frac{1}{2}(I - \hat{Z})\com{\bigB{A,C}, \bigB{B,D}}
                    \end{align}
                \end{enumerate}
            \end{proposition}
        
            The proof follows exactly as in \cref{lem:I_plus_Zhat}.

            \noindent We know that $\hat{Z} \cdot P_{jk} = \Popp{jk}$ for even $n$. These remaining scenarios are summarized in the following lemma.
        
            \begin{lemma}\label{lem:Z_hat_times_D_P_and_Q}
                For any $1 \leq j < k \leq n$,
                \begin{align}\label{eq:Z_hat_times_D}
                    \hat{Z} \cdot D_{jk} &= \Dopp{jk}\\
                    \hat{Z} \cdot P_{jk} &= \begin{cases}
                        \Popp{jk} &n \text{ even}\\
                        -\Qopp{jk} &n \text{ odd}
                    \end{cases}\label{eq:Z_hat_times_P}\\
                    \hat{Z} \cdot Q_{jk} &= \begin{cases}
                        \Qopp{jk} &n \text{ even}\\
                        \Popp{jk} &n \text{ odd}
                    \end{cases}\label{eq:Z_hat_times_Q}
                \end{align}
            \end{lemma}
    
            Note that since $\hat{Z}$ is an involution, we also have the other direction e.g. $\hat{Z} \cdot \Popp{jk} = \begin{cases}
                    P_{jk} &n \text{ even}\\
                    Q_{jk} &n \text{ odd}
                \end{cases}$.
    
            \begin{proof}[(Proof \cref{lem:Z_hat_times_D_P_and_Q})]
                To begin:
                \begin{equation}
                    \hat{Z} \cdot D_{jk} = (-1)^{\floor{\frac{n}{2}}} \Zn \cdot \frac{i}{2}\bigP{Z_j - Z_k} = \frac{i^{2\floor{\frac{n}{2}} + 1}}{2}\bigP{\Zbar{j} - \Zbar{k}} = \Dopp{jk}
                \end{equation}
    
                Assume that $n$ is even. We know $\hat{Z} \cdot P_{jk} = \Popp{jk}$. Assume that $k - j$ is odd. This implies that $n - k + j$ is also odd. Then
                \begin{align}
                    \hat{Z} \cdot Q_{jk} &= (-1)^{\floor{\frac{n}{2}}} \Zn \cdot \frac{i^{b_{jk}}}{2}\bigP{\AB{YX}{jk} - \AB{XY}{jk}} \\
                    &=\frac{i^{2\floor{\frac{n}{2}} + b_{jk}}}{2}Z_{[1:j-1]}\bigP{(ZY)_j (ZZ)_{[j+1:k-1]} (ZX)_k - (ZX)_j (ZZ)_{[j+1:k-1]} (ZY)_k }Z_{[k+1:n]} \\
                    &=\frac{i^{2\floor{\frac{n}{2}} + b_{jk}}}{2}Z_{[1:j-1]}\bigP{-i^2 X_jY_k + i^2Y_j  X_k }Z_{[k+1:n]}\\
                    &=\frac{i^{2\floor{\frac{n}{2}} + b_{jk}}}{2}Z_{[1:j-1]}\bigP{i^4 X_jY_k - i^4Y_j  X_k }Z_{[k+1:n]}\\
                    &=\frac{i^{2\floor{\frac{n}{2}} + b_{jk} }}{2}\bigP{\ABopp{XY}{jk} - \ABopp{YX}{jk}}\label{eq:Zhat_Qjk}
                \end{align}
    
                \noindent We now handle the exponent. First note that $n$ is even and so $2\floor{\frac{n}{2}} = n$. Since $k -j$ is odd, we have that $b_{jk} = k - j - 2$. Furthermore, $2(k-j) \equiv 2 \mod 4$. Thus,
                \begin{align}
                    2\floor{\frac{n}{2}} + b_{jk}  = n + k - j - 2 \equiv  n-k + j \bmod 4 = \bopp{jk}
                \end{align}
    
                \noindent Putting this together, we have that $\hat{Z} \cdot Q_{jk} = \Qopp{jk}$.
    
                Next, assume that $n$ is odd. Still assume $k -j$ is odd. Note that this implies that $n - k + j$ is now even.
                \begin{align}
                    \hat{Z}\cdot P_{jk} &= (-1)^{\floor{\frac{n}{2}}} \Zn \cdot\frac{i^{c_{jk}}}{2} \bigP{\AB{XX}{jk} + \AB{YY}{jk}} \\
                    &=\frac{i^{2\floor{\frac{n}{2}} + c_{jk}}}{2}Z_{[1:j-1]}\bigP{(ZX)_j (ZZ)_{[j+1:k-1]} (ZX)_k + (ZY)_j (ZZ)_{[j+1:k-1]} (ZY)_k }Z_{[k+1:n]} \\
                    &=\frac{i^{2\floor{\frac{n}{2}} + c_{jk}}}{2}Z_{[1:j-1]}\bigP{i^2 Y_jY_k + i^2X_j X_k }Z_{[k+1:n]}\\
                    &=-\frac{i^{2\floor{\frac{n}{2}} + c_{jk}}}{2}\bigP{\ABopp{XX}{jk} + \ABopp{YY}{jk}}
                \end{align}
    
                \noindent Since $n$ is odd, $2\floor{\frac{n}{2}} = n -1$. Since $k-j$ is odd, $c_{jk} = k - j$. So,
                \begin{equation}
                    2\floor{\frac{n}{2}} + c_{jk} = n + k - j - 1 \equiv n-k + j +1 \bmod 4 = \bopp{jk}
                \end{equation}
    
                \noindent Thus, $\hat{Z}\cdot P_{jk} = -\Qopp{jk}$.
    
                Lastly, note that \cref{eq:Zhat_Qjk} is still valid: $\hat{Z} \cdot Q_{jk} = \frac{i^{2\floor{\frac{n}{2}} + b_{jk} }}{2}\bigP{\ABopp{XY}{jk} - \ABopp{YX}{jk}}$. The difference here is that $2\floor{\frac{n}{2}} = n -1$ and $b_{jk} =  k - j  - 2$. So,
                \begin{equation}
                    2\floor{\frac{n}{2}} + b_{jk} = n - k - j - 3 \equiv n-k + j -1 \bmod 4 = \copp{jk}
                \end{equation}
    
                \noindent Therefore, $\hat{Z} \cdot Q_{jk} = \frac{i^{\copp{jk}}}{2}\bigP{\ABopp{XY}{jk} - \ABopp{YX}{jk}} = \Popp{jk}$.
            \end{proof}
    
            \begin{corollary}
                $\gxycz{n} = \Span{B^+} + \Span{B^-} + \Span{Z^+}$
            \end{corollary}
    
            This follows immediately from \cref{lem:Z_hat_times_D_P_and_Q}.

            \begin{proposition}\label{prop:B_cong_sun}
                $\lie{B^\pm} \cong \su{n}$.
            \end{proposition}

            \begin{proof}
                We claim that $\lie{B^+} \cong \su{n}$ by arguing that $P^+_{jk} := \frac{1}{2}(I + \hat{Z})P_{jk}$ act skew-symmetrically, $ Q^+_{jk} := \frac{1}{2}(I + \hat{Z})Q_{jk} $ act symmetrically, and the $ D^+_{jk} := \frac{1}{2}(I + \hat{Z})D_{jk} $ act diagonally in $\su{n}$. The $B^-$ follows identically. Just as showing $\lie{A^\pm} \cong \so{n}$ reduced to showing that the $P_{jk}$ obeyed $\so{n}$ relations using \cref{lem:I_plus_Zhat}, the calculations here reduce to $\bigB{P_{jk}, Q_{\ell m}, D_a}_{j,k;\ell,m;a}$ via \cref{lem:Z_hat_times_D_P_and_Q,prop:Z_hat_commutes_D_P_and_Q}. For example, to show that these operators obey SU12:
                \begin{align}
                    \com{P_{jk}^+,D^+_{k \ell}} &= \com{\frac{1}{2}\bigP{I + \hat{Z}}P_{jk},\ \frac{1}{2}\bigP{I + \hat{Z}}D_{k\ell}}  \\
                        &=\frac{1}{2}\bigP{I + \hat{Z}}\com{P_{jk},D_{k\ell}} \\
                        &=\frac{1}{2}\bigP{I + \hat{Z}}Q_{jk}\\
                        &=Q_{jk}^+
                \end{align}
    
                All of the calculations follow in the exact same way.
            \end{proof}

        \subsection{Exponentially Large XY-mixer Dynamic Lie Algebras}\label{subsec:exp_dlas}
                            
            We now describe all of the associated DLAs with exponentially large dimension. These include the Lie algebras used to describe QAOA with the XY-mixer for constrained optimization. More surprisingly, the clique topology of the XY-mixer with no additional generators is also exponential. The main results of this section are summarized by \cref{thm:lincon_decomp} and \cref{conj:expliealg} in the Results section and expanded on in \cref{subsec:exp_dlas_desc}. 
            
            \subsubsection{Fully-connected $XY$ Interactions}\label{subsubsec:xyclique}
        
                The fully-connected $XY$-mixer DLA $\gxyk{n} = \lie{i \Gxyk{n}}$ for $ \Gxyk{n} = \{ XY_{j,k} \}_{j<k}^{n} $ is exponential size despite its symmetry. In this section we show a subset of operators exists in $ \Gxyk{n} $ such that $ \text{dim}\left(\gxyk{n}\right) = \Omega(3^n) $. Note that this is a strict lower bound and the DLA themselves seem to be larger than this bound gives and so we present a conjecture of DLA decomposition.

                Fix a pair of vertices $1 \leq j < k \leq n$. Consider $P_{j,k}$ operator as defined in \cref{def_Pjk}. This is constructed by nested commutations of $i XY_{\ell, \ell+1}$ for $j \leq \ell < k$ (see \cref{XY_ring_nested_action}) while $\Popp{j,k}$ as given in \cref{def_Pjk_minus} can be constructed by nesting $i XY_{\ell, \ell+1}$ for $\ell \in [n] \setminus [j,k]$ (see \cref{XY_ring_nested_opposite_action}). In both cases, the operators are constructed on a path from $j$ to $k$ and for each edge included in the path, we nest another $XY$ term. Since a clique of size $n$ contains all cycles of length less than $n$, we can nest from $j$ to $k$ on any path $(j,n_2),(n_2,n_3),\ldots,(n_{N},k)$. The number of unique operators defined are given by the path with the ordering $n_2<n_3<\ldots<n_{N}$:
                \begin{align}
                ad_{iXY_{j,n_1}} \cdot ad_{iXY_{n_1,n_2}} \cdots ad_{iXY_{n_N,k}}.
                \end{align}
                
                We generalize the definition of $c_{jk}$ from \cref{cjk_def} given for $P_{j,k}$ for paths on the clique as:
                \begin{align}\label{def:cjk_paths}
                c_{jk} := 2\floor{\frac{|\sigma_{jk}|+1}{2}} + 1. 
                \end{align}
                In the case that $ \sigma_{jk}(\ell) = 1 $ only for $ \ell \in [j-k] $, then this expression matches \cref{cjk_def}. Then for $ \sigma_{jk} \in \{0,1\}^{n-2} $, each of the following is in the DLA:
                \begin{equation}\label{def:Psjk}
                    P_{\sigma_{jk}} := \frac{i^{c_{jk}}}{2}
                    \begin{cases}
                        XY_{j,k} \prod\limits_{\ell \neq j,k}Z^{\sigma_{jk}(\ell)}_\ell, &\abs{\sigma_{jk}}_1 \text{ odd}\\
                        YX_{j,k} \prod\limits_{\ell \neq j,k}Z^{\sigma_{jk}(\ell)}_\ell, &\abs{\sigma_{jk}}_1 \text{ even}\\
                    \end{cases}
                \end{equation}
                \noindent For example, setting $\sigma_{jk}(\ell) = 1$ when $\ell \in [j,k-1]$ and zero otherwise results in the original $P_{j,k}$ operator. Then $P_{\sigma_{jk}} $ can be constructed through nested $XY$ terms on a path beginning at node $j$ and going to each node $\ell$ where $\sigma_{jk}(\ell) = 1$ and then ending at node $k$. 
                
                Moreover, we can construct terms:
                \begin{equation}\label{def:Pmsjk}
                    P_{\mu,\sigma_{jk}} := \frac{i^{c_{jk}}}{2}
                    \begin{cases}
                        XY_{j,k} \prod\limits_{(p,q) \in \mu} YX_{p,q} \prod\limits_{\ell \notin V_{\mu} \cup \{j,k\}} Z^{\sigma_{jk}(\ell)}_\ell, & \abs{\mu} / 2 + \abs{\sigma_{jk}}_1 \text{ odd}\\
                        YX_{j,k} \prod\limits_{(p,q) \in \mu} YX_{p,q} \prod\limits_{\ell \notin V_{\mu} \cup \{j,k\}} Z^{\sigma_{jk}(\ell)}_\ell, & \abs{\mu} / 2 + \abs{\sigma_{jk}}_1 \text{ even}\\
                    \end{cases}
                \end{equation}
                
                \noindent where $\mu \subseteq \{ (p,q) : p,q \in [n] \setminus \{j,k\} , p \neq q \} $ is a set of ordered pairs, $V_{\mu} = \{ p : (p,q) \in \mu \} \cup \{ q : (p,q) \in \mu \} $ is the set of all individual nodes of each pair, and the pairs are disjoint such that $ 2\,|\mu| = |V_{\mu}| $.
        
                The construction follows from an appropriate $P_{\sigma_{jk}'}$ in $\gxyk{n}$ and applying $ XY_{p,q} $ on the appropriate slot for each $(p,q) \in \mu$ where $ \sigma_{jk}'(q) = 0 $ (such that we have identity on this qubit) and $ \sigma_{jk}'(p) = 1 $ (such that we have $ Z $ on this qubit). Without loss of generality, assume $ \abs{\sigma_{jk}'}_1 $ is odd. Then we construct:
                \begin{align} 
                P_{\{(p_1,q_1)\}, \sigma_{jk}' - e_{q_1}} &= \com{ i XY_{p_1,q_1} , P_{\sigma_{jk}'} } / 2 \\
                &= \frac{i^{c_{jk}+1}}{4} XY_{j,k}  \com{ X_{p_1} X_{q_1} + Y_{p_1} Y_{q_1}, I_{p_1} Z_{q_1} } \prod\limits_{\ell \notin \{p_1,q_1,j,k\}} Z^{\sigma_{jk}'(\ell)}_\ell \\ 
                &= \frac{i^{c_{jk}}}{2} XY_{j,k} YX_{p_1,q_1} \prod\limits_{\ell \notin \{p_1,q_1,j,k\}} Z^{\sigma_{jk}'(\ell)}_\ell ,
                \end{align} 
                \noindent where $ e_{q_{1}}(q_1) = 1 $ and zero otherwise. Recall that $i^{M+4} = i^{M}$ and the elementary commutation relationships in \cref{subsec:defs_paulis_XY}. Suppose we have construct $P_{\mu',\sigma_{jk}' - \mu_{jk}'}$ where $ \mu' = \{ (p_1, q_1), \ldots, (p_{r-1}, q_{r-1}) \} $ and $ \mu_{jk}'(\ell) = 1$ ($ \mu_{jk}'(\ell) = 0$) if $ \ell \in \{ q_1, \ldots, q_{r-1} \} $ ($ \ell \notin \{ q_1, \ldots, q_{r-1} \} $). Then we can construct: 
                \begin{align}
                P_{\mu' + \{ (p_r, q_r) \}, \sigma_{jk}' - \mu_{jk}' - e_{q_r}} &=  \com{ i XY_{p_r,q_r} , P_{\mu', \sigma_{jk}' - \mu_{jk}'} } / 2 \\
                &= \frac{i^{c_{jk}+1}}{4} XY_{j,k} \com{ X_{p_r} X_{q_r} + Y_{p_r} Y_{q_r}, I_{p_r} Z_{q_r} } / 2
                \prod\limits_{(p,q) \in \mu'} YX_{p,q}
                \prod\limits_{\ell \notin V_{\mu'} \cup \{j,k\}} Z^{\sigma_{jk}'(\ell)}_\ell \\ 
                &= \frac{i^{c_{jk}}}{2} XY_{j,k} YX_{p_r,q_r} \prod\limits_{(p,q) \in \mu'} YX_{p,q}
                \prod\limits_{\ell \notin V_{\mu'} \cup \{j,k\}} Z^{\sigma_{jk}'(\ell)}_\ell .
                \end{align}
                
                Following this procedure, we can generate any $ P_{\mu, \sigma_{jk}} $, therefore each is in $ \gxyk{n} $.

                How many such terms are there? Any unique $j,k$ pair can be selected from $n$, any $p_{l+1}, q_{l+1}$ pair can be selected from $ n - 2 \, l $ available indices up to any $l < n / 2 - 1$, and any $\sigma$ can be selected over remaining indices $ \{0,1\}^{n-2\,l-2} $. In total, we have at least $ \sum_{l=1}^{n/2} \binom{n}{2\,l} \, 2^{n-2\,l} = ((2+1)^n + (2-1)^n)/2 = \Omega\left( 3^n \right) $ basis elements in $ \gxyk{n} $.

                \cref{tab:dlasummary-exp-numerics} reports the dimension of $ \gxyk{n} $ up to $ n = 9 $. Through numerical experiments we verify that, for small $n$, $ \text{proj}_{F^{k}} \left( \gxyk{n} \right) = \text{span}\left(\{ \hat{F}^{k} g \hat{F}^{k} : g \in \gxyk{n} \} \right) \cong \su{\binom{n}{k}} $ for all $ k $ except $ k = n/2 $ (in the case that $ n $ is even). In this even case, $ \text{dim}\left( \text{proj}_{F^{n/2}} \left( \gxyk{n} \right) \right) = 2\,(\binom{n}{n/2}^2/2 -1) $, which suggests congruence with $ \su{\frac{1}{2}\binom{n}{n/2}} \oplus \su{\frac{1}{2}\binom{n}{n/2}} $. We therefore split the classification into even and odd: $ \gxyke{n} $ in the case $ n $ is even and $ \gxyko{n} $ in the case $ n $ is odd. Based on the dimension of the Lie algebra and stability of classification as $n$ scales, we conjecture the following:
                \begin{conjecture}\label{conj:gxyk}
                    The classification of $ \gxyk{n} $ depends on parity. In the case that $ n $ is odd, \mbox{$ \gxyko{n} \cong \gxykdecompodd $}; in the case that $ n $ is even, $ \gxyke{n} \cong \gxykdecompeven $. Then, the dimension of the Lie algebra scales as $ \text{dim}\left(\gxyk{n}\right) =  \frac{1}{2} \left( \binom{2n}{n} - 3 + (-1)^{n+1} \right) + \lfloor \frac{n}{2} \rfloor  = 4^{n - \Theta\left( \log(n) \right)} $. 
                \end{conjecture}     
            
            \subsubsection{Fully-connected $XY$ Interactions with $R_z$}\label{subsubsec:xyclique_z}
                    
                Since the generates of XY clique are included, the Lie algebra of the XY clique is a Lie subalgebra $\gxyk{n} \leq \gxykz{n} $, and therefore $ \text{dim}\left( \gxyk{n} \right) = \Omega\left(3^n\right) $. \cref{tab:dlasummary-exp-numerics} reports the dimension of $ \gxykz{n} $ up to $ n = 9 $. Recognize that $ Z^{+} \in \gxyk{n} $ and so $ \gxykz{n} = \text{span}\left( Z^{+}\right) \oplus [ \gxykz{n}, \gxykz{n} ] $. Through numerical experiments we verify that, for small $n$,
                \begin{align} 
                \text{proj}_{F^{k}}\left( \com{ \gxykz{n}, \gxykz{n} } \right) = \text{span}\left( \{ \hat{F}^{k} g \hat{F}^{k} : g \in \com{ \gxykz{n}, \gxykz{n} } \} \right) \cong \su{\binom{n}{k}}.
                \end{align}  
                
                Based on the dimension of the Lie algebra and that this decomposition will hold as $n$ scales, we conjecture: 
                \begin{conjecture}\label{conj:gxykz}
                    $ \gxykz{n} = \text{span}\left( Z^{+} \right) \, \oplus \, \redglincon \cong \u{1} \oplus \bigoplus_{k=1}^{n-1} \su{\binom{n}{k}} $. As a consequence of \cref{eq:exact_dims_ker_ad_Zplus}, the dimension of the Lie algebra scales as $ \text{dim}\left( \gxykz{n} \right) = \binom{2 \, n}{n} - n = 4^{n - \Theta\left(\log(n)\right)} $.
                \end{conjecture}
            
            \subsubsection{$XY$ Cycle with $R_z$ and Fully-conencted $R_{zz}$ Interactions}\label{subsubsec:xyring_zz}

                It is sufficient to show that each term in $ \Gxyk{n} $ is in $ \gxyczzz{n} $ to conclude $ \gxyk{n} \leq \gxyczzz{n} $ and therefore $\text{dim}\left(\gxyczzz{n}\right) = \Omega\left( 3^n \right)$. Since $ P_{j,k} \in \gxyczzz{n} $, we show how to construct $ XY_{j,k} $ from $ \{ P_{j,k} \} \cup \mathcal{G}_{Z,ZZ} $. Without loss of generality, assume $ |k-j| $ is odd. We have: 
                \begin{align}
                \com{ \frac{i}{2} Z_j Z_{j+1}, P_{j,k} } &= -\frac{1}{4}\Big( \com{ Z_{j} Z_{j+1} , X_{j} Z_{j+1} } X_{k}  + \com{ Z_{j} Z_{j+1} , Y_{j} Z_{j+1} } Y_{k} \Big)  \prod_{\ell=j+2}^{k-1} Z_{\ell} \\ 
                &= i \left( Y_{j} X_{k} - X_{j} Y_{k} \right) / 2 \prod_{\ell=j+2}^{k-1} Z_{\ell}.
                \end{align} 
                \noindent Then 
                \begin{align}
                \com{ -\frac{i}{2} Z_{j},  \com{ \frac{i}{2} Z_{j} Z_{j+1} , P_{j,k} } }  =  XY_{j,k} \prod_{\ell=j+2}^{k-1} Z_{l}. 
                \end{align} 
                By repeating this argument, we can remove each $ Z_{l} $ such that 
                \begin{align}
                \com{ -\frac{i}{2} Z_j \com{ \frac{i}{2} Z_{j} Z_{k-1} , \com{ \cdots , \com{ -\frac{i}{2} Z_j, \com{ \frac{i}{2} Z_{j}, Z_{j+1} , P_{j,k} } } \cdots } } } = XY_{j,k} . 
                \end{align} 
                
                In contrast to previous Lie algebras, $\gxyczzz{n}$ has a center that includes $ ZZ^+ = i \sum_{j<k} Z_{j} Z_{k} $. Clearly $ ZZ^+ $ commutes with operators in $ \mathcal{G}_{Z,ZZ}(n) $, $ ZZ^+ $ also commutes with $ XY_{j,k} $:
                \begin{align}
                \com{ ZZ^+, XY_{j,k}} &= -\sum_{l<m} \com{ Z_{l} Z_{m}, X_j X_k + Y_j Y_k } \\
                &= -\Bigg( \frac{1}{2}\sum_{l \neq k} \com{ Z_{j} Z_{l}, X_j X_k + Y_j Y_k } + \frac{1}{2}\sum_{l \neq k} \com{ Z_{l} Z_{j}, X_j X_k + Y_j Y_k } \\ 
                &\phantom{= } + \frac{1}{2}\sum_{l \neq j} \com{ Z_{k} Z_{l}, X_j X_k + Y_j Y_k } + \frac{1}{2}\sum_{l \neq j} [ Z_{l} Z_{k}, X_j X_k + Y_j Y_k ] \Bigg)\\ 
                &= -\Bigg( \frac{1}{2}\sum_{l \neq k} 2 \, i \, Y_{j} X_{k} Z_{l}  - 2 \, i \,  X_j Y_k Z_{l} + \frac{1}{2} \sum_{l\neq k } 2 \, i \, X_j Y_k Z_{l}  - 2 \, i \,  Y_j X_k Z_{l}  \\ 
                &\phantom{= } + \frac{1}{2} \sum_{l \neq j} 2 \, i \, X_{j} Y_{k} Z_{l} - 2 \, i \, Y_j X_k Z_{l} + \frac{1}{2} \sum_{l \neq j} 2 \, i \, X_{j} Y_{k} Z_{l}  - 2 \, i \,  Y_{j} X_{k} Z_{l} \Bigg) \\ 
                &= 0. 
                \end{align}

                Then as a corollary of \cref{conj:gxykz}:
                \begin{conjecture}\label{conj:gxyczzz}
                    $ \gxyczzz{n} = \text{span}\left(\{ i Z^+, i ZZ^+ \}\right) \oplus \redglincon \cong \u{1}^{\oplus 2} \oplus \bigoplus_{k=1}^{n-1} \su{\binom{n}{k}} $ and, due to \cref{eq:exact_dims_ker_ad_Zplus}, the dimension of the Lie algebra scales as $ \text{dim}\left( \gxyczzz{n} \right) = \binom{2 \, n}{n} - n + 1 = 4^{n - \Theta\left( \log(n) \right)} $.    
                \end{conjecture}

\section{Ising Embeddings for Linear and Quadratic Costs}\label{app:isingembed}

In this section, we show how to formulate $ H_{f} \in \text{span}\left(\mathcal{G}_{Z,ZZ}\right) $ for linear and quadratic programming, focusing on \texttt{Portfolio Optimization}, \texttt{Sparsest k-Subgraph} and \texttt{Graph Partitioning}.

First we consider the embedding of a linear cost $ \bm{\nu} $ such as $ \bm{\nu}^{T}  x  $. Let $ \ket{x} \in \{ \ket{0}, \ket{1} \}^{\otimes n} $ be the computational basis state for a binary string $ x \in \{0,1\}^{n} $. Let $ \sigma_{i}^{0} = \ketbra{0}{0}_{i} = \left( I + Z_{i} \right) / 2 $ and $ \sigma_{i}^{1} = \ketbra{1}{1}_{i} = \left( I - Z_{i} \right) / 2 $.  

\begin{lemma}\label{lem:embedvec}
Given a vector $ \bm{\nu} \in \R^{n} $, $ H_{\bm{\nu}} = \nu_{i} \sigma_{i}^{1} $ has the expectation value of $ \bm{\nu} $ over computational basis states: $ \bra{x} H_{\bm{\nu}} \ket{x} = \bm{\nu}^{T}  x $. 
\end{lemma}

\begin{proof}[(Proof of \cref{lem:embedvec})] 
\begin{align}
\bra{x} \sum_{i=1}^{n} \nu_{i} \sigma_{i}^{1}\ket{x} &= \sum_{i=1}^{n} \nu_{i}\bra{x} \sigma_{i}^{1} \ket{x} \\
&= \sum_{ij} \nu_{i} | \braket{x_{i}}{1} |^2  \\
&= \sum_{i} \nu_{i} x_{i} \\ 
&= \bm{\nu}^{T} x.
\end{align}
\end{proof}

Since $ \sigma_{i}^{1} = (I - Z_{i})/2 $, we write the embedding explicitly over the $ \mathcal{G}_{Z} \cup \{ I \} $ basis as: 
\begin{align}\label{eq:isinglin}
H_{\bm{\nu}} &= \frac{1}{2} \sum_{i=1}^{n} \nu_{i} \left( I - Z_{i} \right) \\ 
&= \frac{\sum_{i=1}^{n} \nu_{i}}{2} I - \frac{1}{2} \sum_{i=1}^{n} \nu_{i} Z_{i}.
\end{align}

Next we consider a quadratic cost $ \bm{A} $ as $ x^{T} \bm{A} x $ appears in quadratic cost functions such as Eq.~\ref{eq:po_iqp}.

\begin{lemma}\label{lem:embedmat}
Given a symmetric matrix $ \bm{A} \in \R^{n\times n} $, $ H_{\bm{A}} = \sum_{ij} \, A_{ij} \sigma_{i}^{1} \sigma_{j}^{1} $ has the expectation value of $ \bm{A} $ over computational basis states: $ \bra{x} H_{\bm{A}} \ket{x} = x^{T} \bm{A} x $. 
\end{lemma}

\begin{proof}[(Proof of \cref{lem:embedmat})] 
\begin{align}
\bra{x} \sum_{i=1}^{n} \sum_{j=1}^{n} A_{ij} \sigma_{i}^{1} \sigma_{j}^{1} \ket{x} &= \sum_{i=1}^{n} \sum_{j=1}^{n} A_{ij}\bra{x} \sigma_{i}^{1} \sigma_{j}^{1} \ket{x} \\
&= \sum_{ij} A_{ij}| \braket{x_{i}}{1} |^2 | \braket{x_{j}}{1} |^2  \\
&= \sum_{ij} A_{ij} x_{i} x_{j} \\ 
&= \sum_{i} x_{i} \left( \sum_{j=1}^{n} A_{ij} x_{j} \right) \\
&= \sum_{i} x_{i} \left( \bm{A}x \right)_{i}  \\ 
&= x^{T}  \bm{A}  x. 
\end{align}
\end{proof}

To write the embedding explicitly over the $ \mathcal{G}_{Z,ZZ} \cup \{ I \} $ basis:
\begin{align}\label{eq:isingquad}
H_{\bm{A}} &= \sum_{i=1}^{n} \sum_{j=1}^{n} A_{ij} \sigma_{i}^{1} \sigma_{j}^{1} \\ 
&= \sum_{i=1}^{n} \sum_{j=1}^{n} A_{ij} \left( I - Z_{i} \right) \left( I - Z_{j} \right) / 4 \\ 
&= \left( \sum_{i=1}^{n} \sum_{j=1}^{n} A_{ij} / 4 \right) I - \left( \sum_{i=1}^{n} \left( \sum_{j=1}^{n} A_{ij} / 4  \right) Z_{i} \right) \\
&\phantom{= } - \left( \sum_{j=1}^{n} \left( \sum_{i=1}^{n} A_{ij} / 4 \right) Z_{j} \right) + \sum_{i=1}^{n} \sum_{j=1}^{n} A_{ij} Z_{i} Z_{j} / 4  \\
&= \left( \sum_{i=1}^{n} \sum_{j=1}^{n} A_{ij} / 4 \right) I - \frac{1}{2}\sum_{i=1}^{n}\sum_{j=1}^n A_{ij} Z_{i} + \frac{1}{4}\sum_{i=1}^{n} \sum_{j=1}^{n} A_{ij} Z_{i} Z_{j} & (A_{ij}=A_{ji}) \\
&= \left( \frac{1}{2} \sum_{i < j} A_{ij} + \frac{1}{4} \sum_{i} A_{ii} \right) I - \frac{1}{2} \sum_{i=1}^{n} \left( \sum_{j=1}^{n} A_{ij} \right) Z_{i} + \frac{1}{2} \sum_{i < j }^{n} A_{ij} Z_{i} Z_{j}. & (A_{ij}=A_{ji})
\end{align}

Note that coefficients on the identity operator are constants for all wavefunctions. Then consider the reduced cost over single qubit operators: 
\begin{align}
\text{Proj}_{\Gz{n}} \left( H_{\bm{\nu}} + H_{\bm{A}} \right) = \frac{1}{2} \sum_{i=1}^{n} \left( \nu_{i} - \left( \sum_{j=1}^{n} A_{ij} \right) \right) Z_{i}. 
\end{align}

Note that $ \bm{\nu} = 0 $ for \texttt{Sparsest k-Subgraph} and $ \sum_{j=1}^{n} A_{ij} $ is the degree of node $ i $. For \texttt{Portfolio Optimization}, the projection is the expected return of an asset plus the risk profile of the specific asset with respect to all assets.

Consider the cost $ \sum_{ij} A_{ij} (1-x_{i})x_{j} = \sum_{ij} A_{ij} x_{i} x_{j} - A_{ij} x_{j} $ associated with minimizing a cut for a given adjacency matrix $ \bm{A} $. We consider the associated Hamiltonian:
\begin{align}
H_{\text{MinCut}} &= \sum_{i=1}^{n} \sum_{j=1}^{n} A_{ij} \sigma_{i}^{0} \sigma_{j}^{1} \\ 
&= \sum_{i=1}^{n} \sum_{j=1}^{n} A_{ij} (I + Z_{i})(I - Z_{j}) / 4 \\ 
&= \sum_{i=1}^{n} \sum_{j=1}^{n} A_{ij} (I + Z_{i} - Z_{j} - Z_{i} Z_{j}) / 4 \\ 
&= \sum_{i=1}^{n} \sum_{j=1}^{n} A_{ij} (I - Z_{i} Z_{j}) / 4 & (A_{ij} = A_{ji})\\ 
&= \frac{|E|}{2} I - \frac{1}{2} \sum_{i<j} Z_i Z_j.
\end{align}

Then $ H_{\text{MinCut}} $ is the cost Hamiltonian for \texttt{Graph Partitioning}. Notice that $ \text{Proj}_{\mathcal{G}_{Z}} \left( H_{\text{MinCut}} \right) = 0 $.

\section{Common Commutation Calculations}\label{sec:com_calcs}
    
    Here we show the nitty-gritty of all necessary commutation calculations involving the operators $\AB{AB}{j,k}$, $\ABopp{CD}{a,b}$, etc. Let $A,B,C,D \in \{X,Y\}$ and let $j < k < \ell$. Note that $\com{ZZ, \bigB{XX,XY,YX,YY}} = 0$.

    \begin{lemma}\label{lem:ABCD_jkl_coms}
        \begin{enumerate}
            \item[]
            \item $\com{\AB{AB}{j,k}, \AB{CD}{k,\ell}} = A_j Z_{[j+1:k-1]} [B,C]_k Z_{[k+1:\ell-1]} D_{\ell}$. If $B = C$, this is 0.
            \item $\com{\AB{AB}{j,\ell}, \AB{CD}{j,k}} = \com{A_jZ_k, C_jD_k}Z_{[k+1:\ell-1]}B_\ell$.
            \item $\com{\AB{AB}{k,\ell}, \AB{CD}{j,\ell}} = C_j Z_{[j+1:k-1]}\com{A_k  B_\ell, Z_k  D_\ell}$.
            \item $\com{\AB{AB}{jk}, \AB{CD}{jk}} = \com{A_jB_k,C_jD_k}$.
        \end{enumerate}
    \end{lemma}

    \begin{proof}
        \begin{align}
            \com{\AB{AB}{j,k}, \AB{CD}{k,\ell}} &= \AB{AB}{j,k} \AB{CD}{k,\ell} - \AB{CD}{k,\ell} \AB{AB}{j,k} \\
            &=A_j Z_{[j+1:k-1]} (BC)_k Z_{[k+1:\ell-1]} D_\ell - A_j Z_{[j+1:k-1]} (CB)_k Z_{[k+1:\ell-1]} D_\ell\\
            &=A_j Z_{[j+1:k-1]} [B,C]_k Z_{[k+1:\ell-1]} D_{\ell}.\\
            \com{\AB{AB}{j,\ell}, \AB{CD}{j,k}} &= \com{A_j Z_{[j+1:\ell - 1]} B_\ell,\ C_j Z_{[j + 1:k - 1]} D_k} \\
            &= \com{A_j Z_{[j+1:k-1]} Z_k,\ C_j Z_{[j + 1:k - 1]} D_k}Z_{[k+1:\ell-1]}B_\ell \\
            &= \com{A_j Z_k,\ C_jD_k}Z_{[k+1:\ell-1]}B_\ell.\\
            \com{\AB{AB}{k,\ell}, \AB{CD}{j,\ell}} &=\com{A_k Z_{[k+1:\ell-1]} B_\ell, C_j Z_{[j+1:\ell-1]} D_\ell} \\
            &=C_j Z_{[j+1:k-1]}\com{A_k Z_{[k+1:\ell-1]} B_\ell, Z_k Z_{[k+1:\ell-1]} D_\ell}\\
            &=C_j Z_{[j+1:k-1]}\com{A_k  B_\ell, Z_k  D_\ell}\\
            \com{\AB{AB}{jk}, \AB{CD}{jk}} &= \com{A_j Z_{[j+1:k-1]} B_k, C_j Z_{[j+1:k-1]} D_k}  =\com{A_j  B_k, C_j D_k}
        \end{align}
    \end{proof}
    
    \begin{lemma}\label{lem:AB_minus_CD_minus_coms}
        \begin{enumerate}
            \item[]
            \item $\com{\ABopp{AB}{jk}, \ABopp{CD}{k \ell}} = \com{A_jB_kZ_\ell,\  Z_jC_kD_\ell}\cdot Z_{[j+1:\ell-1]\setminus \{k\}}$
            \item $\com{\ABopp{AB}{jk}, \ABopp{CD}{j \ell}} = \com{A_j Z_\ell , C_jD_\ell} B_kZ_{[k+1:\ell-1]}$
            \item $\com{\ABopp{AB}{j\ell}, \ABopp{CD}{k \ell}} = \com{A_jB_\ell, Z_{j}D_\ell } Z_{[j+1:k-1]}C_k$.
            \item $\com{\ABopp{AB}{jk}, \ABopp{CD}{jk}} = \com{A_jB_k,C_jD_k}$
        \end{enumerate}
    \end{lemma}

    \begin{proof}
        \begin{align}
            \com{\ABopp{AB}{jk}, \ABopp{CD}{k \ell}} &= \com{Z_{[1:j-1]}A_jB_k Z_{[k+1:\ell]} Z_{[\ell+1:n]}, Z_{[1:j-1]}Z_{[j:k-1]}C_kD_\ell Z_{[\ell+1:n]}} \\
            &= \com{A_jB_k Z_\ell , Z_jC_kD_\ell }Z_{[j+1:\ell-1]\setminus k}.\\
            \com{\ABopp{AB}{jk}, \ABopp{CD}{j \ell}} &= \com{Z_{[1:j-1]}A_jB_k Z_{[k+1:n]}, Z_{[1:j-1]}C_jD_\ell Z_{[\ell+1:n]}}\\
            &= \com{A_j Z_\ell Z_{[\ell+1:n]}, C_jD_\ell Z_{[\ell+1:n]}} B_kZ_{[k+1:\ell-1]}\\
            &= \com{A_j Z_\ell , C_jD_\ell} B_kZ_{[k+1:\ell-1]}.\\
            \com{\ABopp{AB}{j\ell}, \ABopp{CD}{k \ell}} &= \com{Z_{[1:j-1]}A_jB_\ell Z_{[\ell+1:n]}, Z_{[1:k-1]}C_kD_\ell Z_{[\ell+1:n]}} \\
            &=\com{Z_{[1:j-1]}A_jB_\ell, Z_{[1:j]}D_\ell } Z_{[j+1:k-1]}C_k\\
            &=\com{A_jB_\ell, Z_{j}D_\ell } Z_{[j+1:k-1]}C_k.\\
            \com{\ABopp{AB}{jk}, \ABopp{CD}{jk}} &= \com{Z_{[1:j-1]}A_jB_k Z_{[k+1:n]}, Z_{[1:j-1]}C_jD_k Z_{[k+1:n]}} \\
            &=\com{A_jB_k,C_jD_k}.
        \end{align}
    \end{proof}

    \begin{lemma}\label{lem:AB_CD_minus_com}
        \begin{enumerate}
            \item[]
            \item $\com{\AB{AB}{jk}, \ABopp{CD}{jk}}  =\com{A_jB_{k}, C_jD_{k}} \cdot Z_{[n]\setminus \{j,k\}}$
            
            \item $\com{\AB{AB}{jk}, \ABopp{CD}{k\ell}} = \com{A_jB_k,Z_jC_k}D_\ell Z_{[1:j-1]}Z_{[\ell+1:n]}$
            
            \item $\com{\AB{AB}{jk}, \ABopp{CD}{j\ell}} = [A,C]_j Z_{[1:j-1]}Z_{[j+1:k-1]}B_kD_\ell Z_{[\ell+1:n]}$
            
            \item $\com{\AB{AB}{j\ell}, \ABopp{CD}{k\ell}} = \com{A_jZ_kB_\ell, Z_jC_kD_\ell} Z_{[1:j-1]}Z_{[k+1:\ell-1]}Z_{[\ell+1:n]}$
            
            \item $\com{\AB{AB}{k\ell}, \ABopp{CD}{jk}} = [A_kB_\ell,D_kZ_\ell] Z_{[1:j-1]} C_j  Z_{[\ell+1:n]}$
        \end{enumerate}
    \end{lemma}
    
    \begin{proof}
        \begin{align}
            \com{\AB{AB}{jk}, \ABopp{CD}{jk}} &= \com{A_j Z_{[j+1:k-1]} B_k,\ Z_{[1:j-1]} C_j D_k Z_{[k+1:n]}} \\
            &= \begin{aligned}
                &A_j Z_{[j+1:k-1]} B_k \cdot Z_{[1:j-1]}C_j D_k Z_{[k+1:n]}\\
                &-Z_{[1:j-1]}C_j D_k Z_{[k+1:n]}\cdot A_j Z_{[j+1:k-1]} B_k
            \end{aligned}\\
            &= \begin{aligned}
                &Z_{[1:j-1]} (AC)_j Z_{[j+1:k-1]} (BD)_k Z_{[k+1:n]}\\
                &-Z_{[1:j-1]}(CA)_jZ_{[k+1:n]} (DB)_k Z_{[j+1:k-1]}
            \end{aligned}\\
            &= \com{A_jB_{k}, C_jD_{k}} \cdot Z_{[n]\setminus \{j,k\}}\\
            \com{\AB{AB}{jk}, \ABopp{CD}{k\ell}} &= \com{A_j Z_{[j+1:k-1]} B_k,\ Z_{[1:k-1]}C_k D_\ell Z_{[\ell+1:n]}}\\[3pt]
            &=\begin{aligned}
                &A_j Z_{[j+1:k-1]} B_k \cdot Z_{[1:k-1]}C_k D_\ell Z_{[\ell+1:n]} \\
                &- Z_{[1:k-1]}C_k D_\ell Z_{[\ell+1:n]}\cdot A_j Z_{[j+1:k-1]} B_k
            \end{aligned}\\[3pt]
            &=\begin{aligned}
                &Z_{[1:j-1]}(AZ)_j (ZZ)_{[j+1:k-1]} (BC)_k D_\ell Z_{[\ell+1:n]} \\
                &- Z_{[1:j-1]}(ZA)_j (ZZ)_{[j+1:k-1]} (CB)_k D_\ell Z_{[\ell+1:n]}
            \end{aligned}\\
            &=\com{A_jB_k,Z_jC_k}D_\ell Z_{[1:j-1]}Z_{[\ell+1:n]}\\
            \com{\AB{AB}{jk}, \ABopp{CD}{j\ell}} &= \com{A_j Z_{[j+1:k-1]} B_k,\ Z_{[1:j-1]}C_j D_\ell Z_{[\ell+1:n]}}\\[3pt]
            &=\begin{aligned}
                &A_j Z_{[j+1:k-1]} B_k\cdot Z_{[1:j-1]}C_j D_\ell Z_{[\ell+1:n]} \\
                &- Z_{[1:j-1]}C_j D_\ell Z_{[\ell+1:n]}\cdot A_j Z_{[j+1:k-1]} B_k
            \end{aligned}\\[3pt]
            &=\begin{aligned}
                &Z_{[1:j-1]}(AC)_j Z_{[j+1:k-1]}B_kD_\ell Z_{[\ell+1:n]} \\
                &-Z_{[1:j-1]}(CA)_j Z_{[j+1:k-1]}B_kD_\ell Z_{[\ell+1:n]}
            \end{aligned}\\
            &=[A,C]_j Z_{[1:j-1]}Z_{[j+1:k-1]}B_kD_\ell Z_{[\ell+1:n]}\\
            \com{\AB{AB}{j\ell}, \ABopp{CD}{k\ell}} &= \com{A_j Z_{[j+1:\ell-1]}B_\ell,\ Z_{[1:k-1]}C_k D_\ell Z_{[\ell+1:n]}}\\[3pt]
            &=\begin{aligned}
                &A_j Z_{[j+1:\ell-1]}B_\ell\cdot Z_{[1:k-1]}C_k D_\ell Z_{[\ell+1:n]}\\
                &-Z_{[1:k]}C_k D_\ell Z_{[\ell+1:n]}\cdot A_j Z_{[j+1:\ell-1]}B_\ell
            \end{aligned}\\[3pt]
            &=\begin{aligned}
                &Z_{[1:j-1]} (AZ)_j (ZZ)_{[j+1:k-1]} (ZC)_k Z_{[k+1:\ell-1]} (BD)_\ell Z_{[\ell+1:n]}\\
                &-Z_{[1:j-1]} (ZA)_j (ZZ)_{[j+1:k-1]} (CZ)_k Z_{[k+1:\ell-1]} (DB)_\ell Z_{[\ell+1:n]}
            \end{aligned}\\
            &=\com{A_jZ_kB_\ell, Z_jC_kD_\ell} Z_{[1:j-1]}Z_{[k+1:\ell-1]}Z_{[\ell+1:n]}\\
            \com{\AB{AB}{k\ell}, \ABopp{CD}{jk}} &= \com{A_k Z_{[k+1:\ell-1]}B_\ell,\ Z_{[1:j-1]}C_j D_k Z_{[k+1:n]}}\\[3pt]
            &=\begin{aligned}
                &A_k Z_{[k+1:\ell-1]}B_\ell\cdot Z_{[1:j-1]}C_j D_k Z_{[k+1:n]}\\
                &-Z_{[1:j-1]}C_j D_k Z_{[k+1:n]}\cdot A_k Z_{[k+1:\ell-1]}B_\ell
            \end{aligned}\\[3pt]
            &=\begin{aligned}
                &Z_{[1:j-1]} C_j (AD)_k (ZZ)_{[k+1:\ell-1]}(BZ)_\ell Z_{[\ell+1:n]}\\
                &-Z_{[1:j-1]} C_j (DA)_k (ZZ)_{[k+1:\ell-1]}(ZB)_\ell Z_{[\ell+1:n]}
            \end{aligned}\\
            &=[A_kB_\ell,D_kZ_\ell] Z_{[1:j-1]} C_j  Z_{[\ell+1:n]}
        \end{align}
    \end{proof}

    \begin{lemma}\label{lem:Zbar_AB}
        \begin{itemize}
            \item[]
            \item $\com{\Zbar{k}, \AB{AB}{k,\ell}} = Z_{[1:k-1]}A_k \com{Z,B}_\ell Z_{[\ell+1:n]}$
            \item $\com{\Zbar{\ell}, \AB{AB}{k,\ell}} = Z_{[1:k-1]}\com{Z,A}_k B_\ell Z_{[\ell+1:n]}$
            \item $\com{\Zbar{j}, \AB{AB}{k,\ell}} = 0$ for $j \neq k,\ell$.
        \end{itemize}
    \end{lemma}
    
    \begin{proof}
        \begin{align}
            \com{\Zbar{k}, \AB{AB}{k,\ell}} &= \Zbar{k}\AB{AB}{k,\ell} - \AB{AB}{k,\ell}\Zbar{k}\\
            &=\begin{aligned}
                &Z_{[1:k-1]}A_k (ZZ)_{[k+1:\ell-1]}(ZB)_\ell Z_{[\ell+1:n]} \\
                &-Z_{[1:k-1]}A_k (ZZ)_{[k+1:\ell-1]}(BZ)_\ell Z_{[\ell+1:n]}
            \end{aligned}\\
            &=Z_{[1:k-1]}A_k \com{Z,B}_\ell Z_{[\ell+1:n]}\\
            \com{\Zbar{\ell}, \AB{AB}{k,\ell}} &= \Zbar{\ell}\AB{AB}{k,\ell} - \AB{AB}{k,\ell}\Zbar{\ell}\\
            &=\begin{aligned}
                &Z_{[1:k-1]}(ZA)_k (ZZ)_{[k+1:\ell-1]}B_\ell Z_{[\ell+1:n]} \\
                &-Z_{[1:k-1]}(AZ)_k (ZZ)_{[k+1:\ell-1]}B_\ell Z_{[\ell+1:n]}
            \end{aligned}\\
            &=Z_{[1:k-1]}\com{Z,A}_k B_\ell Z_{[\ell+1:n]}\\
        \end{align}

        Assume that $j < k$.
        \begin{align}
            \com{\Zbar{j}, \AB{AB}{k,\ell}} &= \Zbar{j}\AB{AB}{k,\ell} - \AB{AB}{k,\ell}\Zbar{j}\\
            &=\begin{aligned}
                &Z_{[1:j-1]}I_j Z_{[j+1:k-1]}(ZA)_k (ZZ)_{[k+1:\ell-1]}(ZB)_\ell Z_{[\ell+1:n]} \\
                &-Z_{[1:j-1]}I_j Z_{[j+1:k-1]}(AZ)_k (ZZ)_{[k+1:\ell-1]}(BZ)_\ell Z_{[\ell+1:n]}
            \end{aligned}\\
            &=Z_{[1:j-1]}I_j Z_{[j+1:k-1]} \com{ZZ,AB}_{k,\ell} Z_{[\ell+1:n]} \\
            &=0
        \end{align}
        \noindent The other cases follow similarly.
    \end{proof}

    \begin{lemma}\label{lem:com_Z_AB}
        \begin{enumerate}
            \item[]
            \item $[Z_j, \AB{AB}{jk}] = [Z,A]_j Z_{[j+1:k-1]}B_k$
            \item $[Z_k, \AB{AB}{jk}] = A_j Z_{[j+1:k-1]}[Z,B]_k$
            \item $[Z_j, \ABopp{AB}{jk}] = Z_{[1:j-1]}[Z,A]_j B_kZ_{[k+1:n]}$
            \item $[Z_k, \ABopp{AB}{jk}] = Z_{[1:j-1]}A_j [Z,B]_kZ_{[k+1:n]}$
        \end{enumerate}
    \end{lemma}

\end{document}